\documentclass[a4paper,11pt]{scrartcl}
\makeatletter
\DeclareOldFontCommand{\rm}{\normalfont\rmfamily}{\mathrm}
\DeclareOldFontCommand{\sf}{\normalfont\sffamily}{\mathsf}
\DeclareOldFontCommand{\tt}{\normalfont\ttfamily}{\mathtt}
\DeclareOldFontCommand{\bf}{\normalfont\bfseries}{\mathbf}
\DeclareOldFontCommand{\it}{\normalfont\itshape}{\mathit}
\DeclareOldFontCommand{\sl}{\normalfont\slshape}{\@nomath\sl}
\DeclareOldFontCommand{\sc}{\normalfont\scshape}{\@nomath\sc}
\makeatother

\usepackage{float}
\usepackage{amsmath,amssymb,amsthm}
\usepackage{epsfig}
\usepackage{subfig}
\usepackage{xspace}
\usepackage{hyperref}
\RequirePackage{fix-cm}

\numberwithin{equation}{section}
\numberwithin{figure}{section}
\numberwithin{table}{section}

\newcommand{\ie}{i.e.\xspace}
\newcommand{\eg}{e.g.\xspace}
\newcommand{\wlg}{w.l.o.g.\ }

\newcommand{\NN}{\mathbb{N}}
\newcommand{\ZZ}{\mathbb{Z}}
\newcommand{\RR}{\mathbb{R}}
\newcommand{\RRp}{\mathbb{R}_{\ge0}}
\renewcommand{\Pr}[1]{\mbox{\rm\bf Pr}\left[#1\right]}
\newcommand{\Ex}[1]{\mbox{\rm\bf E}\left[#1\right]}

\newcommand{\tra}{^{\operatorname{T}}}
\newcommand{\OPT}{\mathrm{Opt}}

\DeclareMathOperator{\Vol}{Vol}
\DeclareMathOperator{\dist}{{\mathsf d}}

\newcommand{\NP}{{\sf NP}\xspace}

\newtheorem{theorem}{Theorem}
\newtheorem{definition}[theorem]{Definition}
\newtheorem{lemma}[theorem]{Lemma}

\newtheorem{property}[theorem]{Property}

\title{\vspace{-5ex}

Worst Case and Probabilistic Analysis
        of the 2-Opt Algorithm for the TSP\thanks{An extended abstract of this work has
        appeared in the Proceedings of the 18th ACM-SIAM Symposium on Discrete Algorithms (SODA 2007). The results of this extended abstract have been split into two journal articles~\cite{EnglertRV14Algorithmica} and~\cite{EnglertRV16ACM}. This report is an updated version of~\cite{EnglertRV14Algorithmica}, in which two minor errors in the proofs of Lemma~8 and Lemma~9 have been corrected. We thank Bodo Manthey for pointing out these errors.

        This work was supported in part by the EU within the 6th Framework
        Programme under contract 001907 (DELIS), by DFG grants VO~889/2 and WE~2842/1, and by EPSRC grant EP/F043333/1. 

We thank the referee of~\cite{EnglertRV14Algorithmica} for her/his extraordinary efforts and many helpful suggestions.}}
\author{Matthias Englert\thanks{DIMAP and Dept.~of Computer Science, University of Warwick \href{mailto:englert@dcs.warwick.ac.uk}{englert@dcs.warwick.ac.uk}}
\and
Heiko R\"oglin\thanks{Dept.~of Computer Science, University of Bonn \href{mailto:roeglin@cs.uni-bonn.de}{roeglin@cs.uni-bonn.de}}
\and Berthold V\"ocking\thanks{Dept.~of Computer Science, RWTH Aachen University}
}

\date{\vspace{-5ex}}

\begin{document}

\maketitle

\begin{abstract}
2-Opt is probably the most basic local search heuristic for 
the TSP\@. This heuristic achieves amazingly good results on ``real 
world'' Euclidean instances both with respect to running time and 
approximation ratio. There are numerous experimental studies on the 
performance of 2-Opt. However, the theoretical knowledge about this 
heuristic is still very limited. Not even its worst case running time 
on 2-dimensional Euclidean instances was known so far. We clarify this issue by 
presenting, for every $p\in\NN$, a family of $L_p$ instances on which 
2-Opt can take an exponential number of steps.

Previous probabilistic analyses were restricted to instances in which 
$n$ points are placed uniformly at random in the unit square $[0,1]^2$, 
where it was shown that the expected number of steps is bounded by 
$\tilde{O}(n^{10})$ for Euclidean instances. We consider a more 
advanced model of probabilistic instances in which the points can be 
placed independently according to general distributions on $[0,1]^d$, for an 
arbitrary $d\ge 2$. In particular, we allow different distributions for 
different points. We study the expected number of local improvements in 
terms of the number $n$ of points and the maximal density $\phi$ of the 
probability distributions. We show an upper bound on the expected 
length of any 2-Opt improvement path of 
$\tilde{O}(n^{4+1/3}\cdot\phi^{8/3})$. When starting with an initial 
tour computed by an insertion heuristic, the upper bound on the 
expected number of steps improves even to 
$\tilde{O}(n^{4+1/3-1/d}\cdot\phi^{8/3})$. If the distances are 
measured according to the Manhattan metric, then the expected number of 
steps is bounded by $\tilde{O}(n^{4-1/d}\cdot\phi)$. In addition, we 
prove an upper bound of $O(\sqrt[d]{\phi})$ on the expected 
approximation factor with respect to all $L_p$ metrics.
  
Let us remark that our probabilistic analysis covers as special cases 
the uniform input model with $\phi=1$ and a smoothed analysis with 
Gaussian perturbations of standard deviation $\sigma$ with 
$\phi\sim1/\sigma^d$.
\end{abstract}

\section{Introduction}
\label{sec:Introduction}

In the \emph{traveling salesperson problem (TSP)}, we are given a set 
of \emph{vertices} and for each pair of distinct vertices a distance. 
The goal is to find a \emph{tour} of minimum length that visits every 
vertex exactly once and returns to the initial vertex at the end. 
Despite many theoretical analyses and experimental evaluations of the 
TSP, there is still a considerable gap between the theoretical results 
and the experimental observations. One important special case is the 
\emph{Euclidean TSP} in which the vertices are points in $\RR^d$, for 
some $d\in\NN$, and the distances are measured according to the 
Euclidean metric. This special case is known to be \NP-hard in the 
strong sense~\cite{Papadimitriou77}, but it admits a polynomial time 
approximation scheme (PTAS), shown independently in 1996 by 
Arora~\cite{Arora98} and Mitchell~\cite{Mitchell99}. These 
approximation schemes are based on dynamic programming. However, the 
most successful algorithms on practical instances rely on the principle 
of local search and very little is known about their complexity.

The \emph{2-Opt} algorithm is probably the most basic local search 
heuristic for the TSP\@. 2-Opt starts with an arbitrary initial tour 
and incrementally improves this tour by making successive improvements 
that exchange two of the edges in the tour with two other edges. More 
precisely, in each \emph{improving step} the $2$-Opt algorithm selects 
two edges $\{u_1,u_2\}$ and $\{v_1,v_2\}$ from the tour such that 
$u_1,u_2,v_1,v_2$ are distinct and appear in this order in the tour, 
and it replaces these edges by the edges $\{u_1,v_1\}$ and 
$\{u_2,v_2\}$, provided that this change decreases the length of the 
tour. The algorithm terminates in a local optimum in which no further 
improving step is possible. We use the term \emph{2-change} to denote a 
local improvement made by 2-Opt. This simple heuristic performs 
amazingly well on ``real-life'' Euclidean instances like, \eg, the ones 
in the well-known TSPLIB~\cite{Reinelt91}. Usually the 2-Opt heuristic 
needs a clearly subquadratic number of improving steps until it reaches 
a local optimum and the computed solution is within a few percentage 
points of the global optimum~\cite{JohnsonG97}.

There are numerous experimental studies on the performance of 2-Opt. 
However, the theoretical knowledge about this heuristic is still very 
limited. Let us first discuss the number of local improvement steps 
made by 2-Opt before it finds a locally optimal solution. When talking 
about the number of local improvements, it is convenient to consider 
the \emph{state graph}. The vertices in this graph correspond to the 
possible tours and an arc from a vertex $v$ to a vertex $u$ is 
contained if $u$ is obtained from $v$ by performing an improving 2-Opt 
step. On the positive side, van Leeuwen and Schoone consider a 2-Opt 
variant for the Euclidean plane in which only steps are allowed that 
remove a crossing from the tour. Such steps can introduce new 
crossings, but van Leeuwen and Schoone~\cite{vLeeuwenS81} show that 
after $O(n^3)$ steps, 2-Opt finds a tour without any crossing. On 
the negative side, Lueker~\cite{Lueker75} constructs TSP instances 
whose state graphs contain exponentially long paths. Hence, 2-Opt can 
take an exponential number of steps before it finds a locally optimal 
solution. This result is generalized to $k$-Opt, for arbitrary $k\ge2$, 
by Chandra, Karloff, and Tovey~\cite{ChandraKT99}. These negative 
results, however, use arbitrary graphs 
that cannot be embedded into low-dimensional Euclidean space. 
Hence, they leave open the question 
as to whether it is possible to construct Euclidean TSP instances on which
2-Opt can take an exponential number of steps, which has explicitly
been asked by Chandra, Karloff, and Tovey. We resolve this question by constructing 
such instances in the Euclidean plane. In chip design applications, 
often TSP instances arise in which the distances are measured according 
to the Manhattan metric. Also for this metric and for every other $L_p$ 
metric, we construct instances with exponentially long paths in the 
2-Opt state graph.

\begin{theorem}
\label{theorem:LowerBounds}
For every $p\in\{1,2,3,\ldots\}\cup\{\infty\}$ and $n\in\NN=\{1,2,3,\ldots\}$, there is a 
two-dimensional TSP instance with $16n$ vertices in which the distances 
are measured according to the $L_p$ metric and whose state graph 
contains a path of length $2^{n+4}-22$.
\end{theorem}

For Euclidean instances in which $n$ points are placed independently uniformly at 
random in the unit square, Kern~\cite{Kern89} shows that the length of 
the longest path in the state graph is bounded by $O(n^{16})$ with 
probability at least $1-c/n$ for some constant $c$. Chandra, Karloff, and 
Tovey~\cite{ChandraKT99} improve this result by bounding the \emph{expected} 
length of the longest path in the state graph by $O(n^{10}\log{n})$. 
That is, independent of the initial tour and the choice of the local 
improvements, the expected number of 2-changes is bounded by 
$O(n^{10}\log{n})$. For instances in which $n$ points are placed 
uniformly at random in the unit square and the distances are measured 
according to the Manhattan metric, Chandra, Karloff, and Tovey show 
that the expected length of the longest path in the state graph is 
bounded by $O(n^6\log{n})$.

We consider a more general probabilistic input model and improve the 
previously known bounds. The probabilistic model underlying our 
analysis allows different vertices to be placed independently according to 
different continuous probability distributions in the unit hypercube 
$[0,1]^d$, for some constant \emph{dimension} $d\ge2$. The distribution 
of a vertex $v_i$ is defined by a density function $f_i\colon[0,1]^d 
\to[0,\phi]$ for some given $\phi\ge 1$. Our upper bounds depend on the 
number $n$ of vertices and the upper bound $\phi$ on the density. We 
denote instances created by this input model as \emph{$\phi$-perturbed 
Euclidean} or \emph{Manhattan instances}, depending on the underlying 
metric. The parameter $\phi$ can be seen as a parameter specifying how 
close the analysis is to a worst case analysis since the larger $\phi$ 
is, the better can worst case instances be approximated by the 
distributions. For $\phi=1$ and $d=2$, every point has a uniform 
distribution over the unit square, and hence the input model equals the 
uniform model analyzed before. Our results narrow the gap between the 
subquadratic number of improving steps observed in 
experiments~\cite{JohnsonG97} and the upper bounds from the 
probabilistic analysis. With slight modifications, this model also 
covers a smoothed analysis, in which first an adversary specifies the 
positions of the points and after that each position is slightly 
perturbed by adding a Gaussian random variable with small standard 
deviation $\sigma$. In this case, one has to set 
$\phi=1/(\sqrt{2\pi}\sigma)^d$.

We prove the following theorem about the expected length of the longest 
path in the 2-Opt state graph for the three probabilistic input models 
discussed above. It is assumed that the dimension $d\ge 2$ is 
an arbitrary constant.
\begin{theorem}
\label{theorem:runningTime1}
The expected length of the longest path in the 2-Opt state graph
\begin{enumerate}
\setlength{\itemsep}{0em}
\renewcommand{\labelenumi}{\alph{enumi})}
\item is $O(n^4\cdot\phi)$ for $\phi$-perturbed Manhattan instances 
with $n$ points.
\item is $O(n^{4+1/3}\cdot\log( n\phi)\cdot\phi^{8/3})$ for 
$\phi$-perturbed Euclidean instances with $n$ points.
\end{enumerate}
\end{theorem}

Usually, 2-Opt is initialized with a tour computed by some tour 
construction heuristic. One particular class is that of \emph{insertion 
heuristics}, which insert the vertices one after another into the tour. 
We show that also from a theoretical point of view, using such an 
insertion heuristic yields a significant improvement for metric TSP 
instances because the initial tour 2-Opt starts with is much shorter 
than the longest possible tour. In the following theorem, we summarize 
our results on the expected number of local improvements.
\begin{theorem}
\label{theorem:runningTime2}
The expected number of steps performed by 2-Opt
\begin{enumerate}
\setlength{\itemsep}{0em}
\renewcommand{\labelenumi}{\alph{enumi})}
\item is $O(n^{4-1/d}\cdot\log{n}\cdot\phi)$ on $\phi$-perturbed 
Manhattan instances with $n$ points when 2-Opt is initialized with a 
tour obtained by an arbitrary insertion heuristic.
\item is $O(n^{4+1/3-1/d}\cdot\log^{2}(n\phi)\cdot\phi^{8/3})$ on 
$\phi$-perturbed Euclidean instances with $n$ points when 2-Opt is 
initialized with a tour obtained by an arbitrary insertion heuristic.
\end{enumerate}
\end{theorem}
In fact, our analysis shows not only that the expected number of local 
improvements is polynomially bounded but it also shows that the second 
moment and hence the variance is bounded polynomially for 
$\phi$-perturbed Manhattan instances. For the Euclidean 
metric, we cannot bound the variance but the $3/2$-th moment 
polynomially.

In~\cite{EnglertRV07}, we also consider a model in which an arbitrary graph
$G=(V,E)$ is given along with, for each edge $e\in E$, a probability distribution
according to which the edge length $\dist(e)$ is chosen independently of the
other edge lengths. Again, we restrict the choice of distributions to
distributions that can be represented by density functions
$f_e\colon[0,1]\to[0,\phi]$ with maximal density at most $\phi$ for a given
$\phi\ge 1$. We denote inputs created by this input model as
\emph{$\phi$-perturbed graphs}. Observe that in this input model only the
distances are perturbed whereas the graph structure is not changed by the
randomization. This can be useful if one wants to explicitly prohibit certain
edges. However, if the graph $G$ is not complete, one has to initialize 2-Opt
with a Hamiltonian cycle to start with. We prove that in this model the expected
length of the longest path in the 2-Opt state graph is $O(|E|\cdot
n^{1+o(1)}\cdot\phi)$. As the techniques for proving this result are different
from the ones used in this article, we will present it in a separate journal article.

As in the case of running time, the good approximation ratios obtained by 
2-Opt on practical instances cannot be explained by a worst-case 
analysis. In fact, there are quite negative results on the worst-case 
behavior of 2-Opt. For example, Chandra, Karloff, and 
Tovey~\cite{ChandraKT99} show that there are Euclidean instances in the 
plane for which 2-Opt has local optima whose costs are 
$\Omega\left(\frac{\log n}{\log\log n}\right)$ times larger than the 
optimal costs. However, the same authors also show that the expected 
approximation ratio of the worst local optimum for instances with $n$ 
points drawn uniformly at random from the unit square is bounded from 
above by a constant. We generalize their result to our input model in 
which different points can have different distributions with bounded 
density $\phi$ and to all $L_p$ metrics.
\begin{theorem}
\label{theorem:approximation}
Let $p\in\NN\cup\{\infty\}$.
For $\phi$-perturbed $L_p$ instances, the 
expected approximation ratio of the worst tour that is locally optimal 
for 2-Opt is $O(\sqrt[d]{\phi})$.
\end{theorem}

The remainder of the paper is organized as follows. We start by stating 
some basic definitions and notation in 
Section~\ref{sec:Preliminaries}. In Section~\ref{sec:LowerBounds}, we 
present the lower bounds. In Section~\ref{sec:runningTime}, we analyze 
the expected number of local improvements and prove 
Theorems~\ref{theorem:runningTime1} and~\ref{theorem:runningTime2}. 
Finally, in Sections~\ref{sec:approximation} 
and~\ref{sec:SmoothedAnalysis}, we prove 
Theorem~\ref{theorem:approximation} about the expected approximation 
factor and we discuss the relation between our analysis and a
smoothed analysis.
\section{Preliminaries}
\label{sec:Preliminaries}

An \emph{instance of the TSP} consists of a set $V=\{v_1,\ldots,v_n\}$ 
of \emph{vertices} (depending on the context, synonymously referred to 
as \emph{points}) and a symmetric \emph{distance function} $\dist\colon 
V\times V\to\RRp$ that associates with each pair $\{v_i,v_j\}$ of 
distinct vertices a distance $\dist(v_i,v_j)=\dist(v_j,v_i)$. The goal 
is to find a Hamiltonian cycle of minimum length. We also use the term 
\emph{tour} to denote a Hamiltonian cycle. 
We define $\NN=\{1,2,3,\ldots\}$, and
for a natural number $n\in\NN$, we denote the set $\{1,\ldots,n\}$ by $[n]$.

A pair $(V,\dist)$ of a nonempty set $V$ and a function $\dist\colon 
V\times V\to\RRp$ is called a \emph{metric space} if for all $x,y,z\in 
V$ the following properties are satisfied:
\begin{enumerate}
  \renewcommand{\labelenumi}{(\alph{enumi})}
  \setlength{\itemsep}{0em}
  \item $\dist(x,y)=0$ if and only if $x=y$ \emph{(reflexivity)},
  \item $\dist(x,y)=\dist(y,x)$ \emph{(symmetry)}, and
  \item $\dist(x,z)\le \dist(x,y)+\dist(y,z)$ \emph{(triangle inequality)}.
\end{enumerate}
If $(V,\dist)$ is a metric space, then $\dist$ is called a \emph{metric 
on $V$}. A TSP instance with vertices $V$ and distance function $\dist$ 
is called \emph{metric TSP instance} if $(V,\dist)$ is a metric space.

A well-known class of metrics on $\RR^d$ is the class of \emph{$L_p$ 
metrics}. For $p\in\NN$, the distance $\dist_{p}(x,y)$ of two points 
$x\in\RR^d$ and $y\in\RR^d$ with respect to the $L_p$ metric is given 
by $\dist_{p}(x,y) = \sqrt[p]{|x_1-y_1|^p+\cdots+|x_d-y_d|^p}$. The 
$L_1$ metric is often called \emph{Manhattan metric}, and the $L_2$ 
metric is well-known as \emph{Euclidean metric}. For $p\to\infty$, the 
$L_p$ metric converges to the $L_{\infty}$ metric defined by the 
distance function 
$\dist_{\infty}(x,y)=\max\{|x_1-y_1|,\ldots,|x_d-y_d|\}$. A TSP 
instance $(V,\dist)$ with $V\subseteq\RR^d$ in which $\dist$ equals 
$\dist_{p}$ restricted to $V$ is called an \emph{$L_p$ instance}. We 
also use the terms \emph{Manhattan instance} and \emph{Euclidean 
instance} to denote $L_1$ and $L_2$ instances, respectively. 
Furthermore, if $p$ is clear from context, we write $\dist$ instead of 
$\dist_p$.

A \emph{tour construction heuristic} for the TSP incrementally 
constructs a tour and stops as soon as a valid tour is created. 
Usually, a tour constructed by such a heuristic is used as the initial 
solution 2-Opt starts with. A well-known class of tour construction 
heuristics for metric TSP instances are so-called \emph{insertion 
heuristics}. These heuristics insert the vertices into the tour one 
after another, and every vertex is inserted between two consecutive 
vertices in the current tour where it fits best. To make this more 
precise, let $T_i$ denote a subtour on a subset $S_i$ of $i$ vertices, 
and suppose $v\notin S_i$ is the next vertex to be inserted. If $(x,y)$ 
denotes an edge in $T_i$ that minimizes 
$\dist(x,v)+\dist(v,y)-\dist(x,y)$, then the new tour $T_{i+1}$ is 
obtained from $T_i$ by deleting the edge $(x,y)$ and adding the edges 
$(x,v)$ and $(v,y)$. Depending on the order in which the vertices are 
inserted into the tour, one distinguishes between several different 
insertion heuristics. Rosenkrantz et al.~\cite{RosenkrantzSL77} show an 
upper bound of $\lceil\log{n}\rceil+1$ on the approximation factor of 
any insertion heuristic on metric TSP instances. Furthermore, they show 
that two variants which they call \emph{nearest insertion} and 
\emph{cheapest insertion} achieve an approximation ratio of 2 for 
metric TSP instances. The nearest insertion heuristic always inserts 
the vertex with the smallest distance to the current tour (i.e.,
the vertex $v\notin S_i$ that minimizes $\min_{x\in S_i}\dist(x,v)$), and the 
cheapest insertion heuristic always inserts the vertex whose insertion 
leads to the cheapest tour $T_{i+1}$.

\section{Exponential Lower Bounds}
\label{sec:LowerBounds}

In this section, we answer Chandra, Karloff, and 
Tovey's question~\cite{ChandraKT99} as to whether it is possible to 
construct TSP instances in the Euclidean plane on which 2-Opt can take 
an exponential number of steps. We present, for every 
$p\in\NN\cup\{\infty\}$, a family of two-dimensional $L_p$ instances 
with exponentially long sequences of improving 2-changes. In 
Section~\ref{subsec:LB:L2}, we present our construction for the 
Euclidean plane, and in Section~\ref{subsec:LB:Lp} we extend this 
construction to general $L_p$ metrics.

\subsection{Exponential Lower Bound for the Euclidean Plane}
\label{subsec:LB:L2}

In Lueker's construction~\cite{Lueker75} many of the 2-changes remove 
two edges that are far apart in the current tour in the sense that many 
vertices are visited between them.
Our construction differs significantly from the 
previous one as the 2-changes in our construction affect the tour only 
locally. The instances we construct are composed of gadgets of constant 
size. Each of these gadgets has a \emph{zero state} and a \emph{one 
state}, and there exists a sequence of improving 2-changes starting in 
the zero state and eventually leading to the one state. Let 
$G_0,\ldots,G_{n-1}$ denote these gadgets. If gadget $G_i$ with $i>0$ 
has reached state one, then it can be reset to its zero state by gadget 
$G_{i-1}$. The crucial property of our construction is that whenever a 
gadget $G_{i-1}$ changes its state from zero to one, it resets gadget 
$G_i$ twice. Hence, if in the initial tour, gadget $G_{0}$ is in its 
zero state and every other gadget is in state one, then for every $i$ 
with $0\le i\le n-1$, gadget $G_{i}$ performs $2^{i}$ state changes 
from zero to one as, for $i>0$, gadget $G_{i}$ is reset $2^{i}$ times.

Every gadget is composed of 2 subgadgets, which we refer to as \emph{blocks}.
Each of these blocks consists of 4 vertices that are consecutively visited in the
tour. For $i\in\{0,\ldots,n-1\}$ and $j\in[2]$, let ${\cal B}^{i}_1$ and ${\cal
B}^{i}_2$ denote the blocks of gadget $G_i$ and let $A^i_j$, $B^i_j$, $C^i_j$,
and $D^i_j$ denote the four points ${\cal B}^{i}_j$ consists of. If one ignores
certain intermediate configurations that arise when one gadget resets another
one, our construction ensures the following property: The points $A^i_j$,
$B^i_j$, $C^i_j$, and $D^i_j$ are always visited consecutively in the tour
either in the order~$A^i_j B^i_j C^i_j D^i_j$ or in the order~$A^i_j C^i_j B^i_j D^i_j$.

Observe that the change from one of these
configurations to the other corresponds to a single 2-change in which the edges
$A^i_jB^i_j$ and $C^i_jD^i_j$ are replaced by the edges $A^i_jC^i_j$ and
$B^i_jD^i_j$, or vice versa. In the following, we assume that the sum
$\dist(A^i_j,B^i_j)+\dist(C^i_j,D^i_j)$ is strictly smaller than the sum
$\dist(A^i_j,C^i_j)+\dist(B^i_j,D^i_j)$, and we refer to the configuration
$A^i_jB^i_jC^i_jD^i_j$ as the \emph{short state} of the block and to the
configuration $A^i_jC^i_jB^i_jD^i_j$ as the \emph{long state}. Another property
of our construction is that neither the order in which the blocks are visited nor
the order of the gadgets is changed during the sequence of 2-changes. Again with
the exception of the intermediate configurations, the order in which the blocks
are visited is ${\cal B}^{0}_1{\cal B}^{0}_2{\cal B}^{1}_1{\cal B}^{1}_2 \cdots
{\cal B}^{n-1}_1{\cal B}^{n-1}_2$ (see Figure~\ref{fig:LB:Gadgets}).

\begin{figure}[H]
\begin{center}
\includegraphics[width=\textwidth]{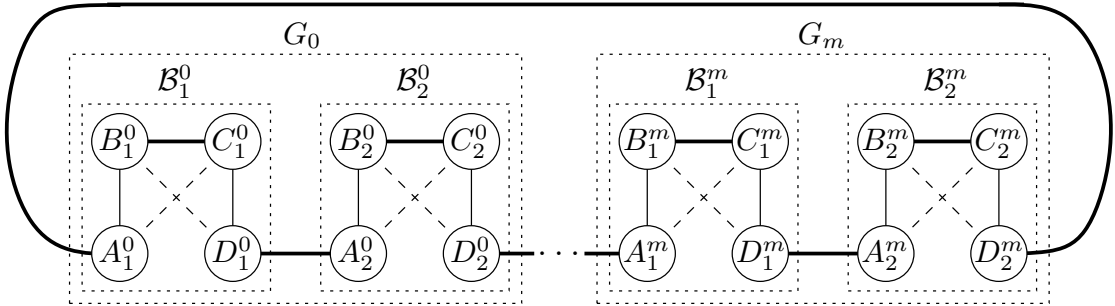}
\caption{In the illustration, we use $m$ to denote $n-1$. Every tour 
that occurs in the sequence of 2-changes contains the thick edges. For 
each block, either both solid or both dashed edges are contained. In 
the former case the block is in its short state; in the latter case the 
block is in its long state.}
\label{fig:LB:Gadgets}
\end{center}
\end{figure}

Due to the aforementioned properties, we can describe every 
non-interme\-di\-ate tour that occurs during the sequence of 2-changes 
completely by specifying for every block if it is in its short state 
or in its long state. In the following, we denote the state of a 
gadget $G_i$ by a pair $(x_1,x_2)$ with $x_j\in\{S,L\}$, meaning 
that block ${\cal B}^i_j$ is in its short state if and only if 
$x_j=S$. Since every gadget consists of two blocks, there are four 
possible states for each gadget. However, only three of them appear in 
the sequence of 2-changes, namely $(L,L)$, $(S,L)$, and $(S,S)$. We 
call state $(L,L)$ the \emph{zero state} and state $(S,S)$ the \emph{one 
state}. In order to guarantee the existence of an exponentially long 
sequence of 2-changes, the gadgets we construct possess the following 
property.
\begin{property}\label{property:LB:Sequence}
If, for $i\in\{0,\ldots,n-2\}$, gadget $G_i$ is in state $(L,L)$ 
(or $(S,L)$, respectively) and gadget $G_{i+1}$ is in state $(S,S)$, then there exists 
a sequence of seven consecutive 2-changes terminating with gadget $G_i$ 
being in state $(S,L)$ (or state $(S,S)$, respectively) and gadget $G_{i+1}$ 
in state $(L,L)$. In this sequence only edges of and between the 
gadgets $G_i$ and $G_{i+1}$ are involved.
\end{property}
We describe in Section~\ref{subsubsec:LB:L2Sequence} how sequences of seven consecutive 2-changes with
the desired properties can be constructed. 
Then we show in Section~\ref{subsubsec:LB:L2Points} that the gadgets can be embedded into the Euclidean plane such that
all of these 2-changes are improving. 
If Property~\ref{property:LB:Sequence} is satisfied and if in the initial tour gadget 
$G_0$ is in its zero state $(L,L)$ and every other gadget is in its 
one state $(S,S)$, then there exists an exponentially long sequence of 
2-changes in which gadget $G_i$ changes $2^i$ times from state zero to 
state one, as the following lemma shows. 
An example with three gadgets
is also depicted in Figure~\ref{fig:LB:Example}.

\begin{figure}[p]
\begin{center}
\includegraphics[scale=0.8]{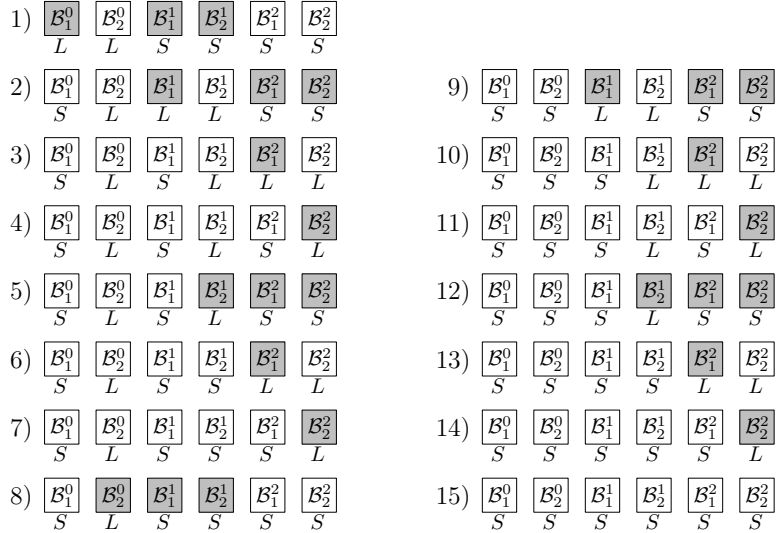}
\caption{This figure shows an example with three gadgets. It shows the 15
configurations that these gadgets assume during the sequence of 2-changes,
excluding the intermediate configurations that arise when one gadget resets 
another one. Gadgets that are involved in the transformation from configuration~$i$
to configuration~$i+1$ are shown in gray. For example, in the step from the
first to the second configuration, the first block ${\cal B}_1^0$ of gadget $G_0$
resets the two blocks of gadget $G_1$. That is, these three blocks follow the
sequence of seven 2-changes from Property~\ref{property:LB:Sequence}. On the other
hand, in the step from the third to the fourth configuration only the first block~${\cal B}_1^2$ of gadget $G_2$
is involved. It changes from its long state to its short state by a single 2-change. 
As this figure shows an example with three gadgets, the total number of 2-changes performed
according to Lemma~\ref{lemma:LB:ExpSequence} is~$2^{3+3-0}-14=50$. This is indeed the case
because 6 of the 14 shown steps correspond to sequences of seven 2-changes
while 8 steps correspond to single 2-changes.}
\label{fig:LB:Example}
\end{center}
\end{figure}

\begin{lemma}
\label{lemma:LB:ExpSequence}
If, for $i\in\{0,\ldots,n-1\}$, gadget $G_i$ is in the zero state 
$(L,L)$ and all gadgets $G_j$ with $j>i$ are in the one state $(S,S)$, 
then there exists a sequence of $2^{n+3-i}-14$ consecutive 2-changes in 
which only edges of and between the gadgets $G_j$ with $j\ge i$ are 
involved and that terminates in a state in which all gadgets $G_j$ with 
$j\ge i$ are in the one state~$(S,S)$.
\end{lemma}
\begin{proof}
We prove the lemma by induction on $i$. If gadget $G_{n-1}$ is in state 
$(L,L)$, then it can change its state with two 2-changes to $(S,S)$ 
without affecting the other gadgets. This is true because the two blocks
of gadget~$G_{n-1}$ can, one after another, change from their long
state~$A^{n-1}_jC^{n-1}_jB^{n-1}_jD^{n-1}_j$ to their short state~$A^{n-1}_jB^{n-1}_jC^{n-1}_jD^{n-1}_j$
by a single 2-change. Hence, the lemma is true for $i=n-1$ because $2^{n+3-(n-1)}-14=2$.

Now assume that the lemma is true for $i+1$ and consider a 
state in which gadget $G_i$ is in state $(L,L)$ and all gadgets $G_j$ 
with $j>i$ are in state $(S,S)$. Due to 
Property~\ref{property:LB:Sequence}, there exists a sequence of seven 
consecutive 2-changes in which only edges of and between $G_i$ and 
$G_{i+1}$ are involved, terminating with $G_i$ being in state $(S,L)$ 
and $G_{i+1}$ being in state $(L,L)$. By the induction hypothesis there 
exists a sequence of $(2^{n+2-i}-14)$ 2-changes after which all gadgets 
$G_j$ with $j>i$ are in state $(S,S)$. Then, due to 
Property~\ref{property:LB:Sequence}, there exists a sequence of seven 
consecutive 2-changes in which only $G_i$ changes its state from 
$(S,L)$ to $(S,S)$ while resetting gadget $G_{i+1}$ again from $(S,S)$ 
to $(L,L)$. Hence, we can apply the induction hypothesis again, 
yielding that after another $(2^{n+2-i}-14)$ 2-changes all gadgets $G_j$ 
with $j\ge i$ are in state $(S,S)$. This concludes the proof as the 
number of 2-changes performed is $14+2(2^{n+2-i}-14)=2^{n+3-i}-14$.
\end{proof}
In particular, this implies that, given 
Property~\ref{property:LB:Sequence}, one can construct instances 
consisting of $2n$ gadgets, \ie, $16n$ points, whose state graphs 
contain paths of length $2^{2n+3}-14>2^{n+4}-22$, as desired in 
Theorem~\ref{theorem:LowerBounds}.

\subsubsection{Detailed description of the sequence of steps}
\label{subsubsec:LB:L2Sequence}

Now we describe in detail how a sequence of 2-changes satisfying 
Property~\ref{property:LB:Sequence} can be constructed. First, we 
assume that gadget $G_i$ is in state $(S,L)$ and that gadget $G_{i+1}$ 
is in state $(S,S)$. Under this assumption, there are three consecutive 
blocks, namely ${\cal B}^{i}_2$, ${\cal B}^{i+1}_1$, and ${\cal 
B}^{i+1}_2$, such that the leftmost one ${\cal B}^{i}_2$ is in its long state, and the 
other blocks are in their short states. We need to find a sequence of 
2-changes in which only edges of and between these three blocks are 
involved and after which ${\cal B}^{i}_2$ is in its short state and the 
other blocks are in their long states. Remember that when the edges 
$\{u_1,u_2\}$ and $\{v_1,v_2\}$ are removed from the tour and the 
vertices appear in the order $u_1,u_2,v_1,v_2$ in the current tour, 
then the edges $\{u_1,v_1\}$ and $\{u_2,v_2\}$ are added to the tour 
and the subtour between $u_1$ and $v_2$ is visited in reverse order. 
If, \eg, the current tour corresponds to the permutation 
$(1,2,3,4,5,6,7)$ and the edges $\{1,2\}$ and $\{5,6\}$ are removed, 
then the new tour is $(1,5,4,3,2,6,7)$. The following sequence of 
2-changes, which is also shown in Figure~\ref{fig:LB:Reset},
has the desired properties. Brackets indicate the edges that 
are removed from the tour.

\begin{figure}[H]
\begin{center}
\subfloat[]{\includegraphics[scale=0.38]{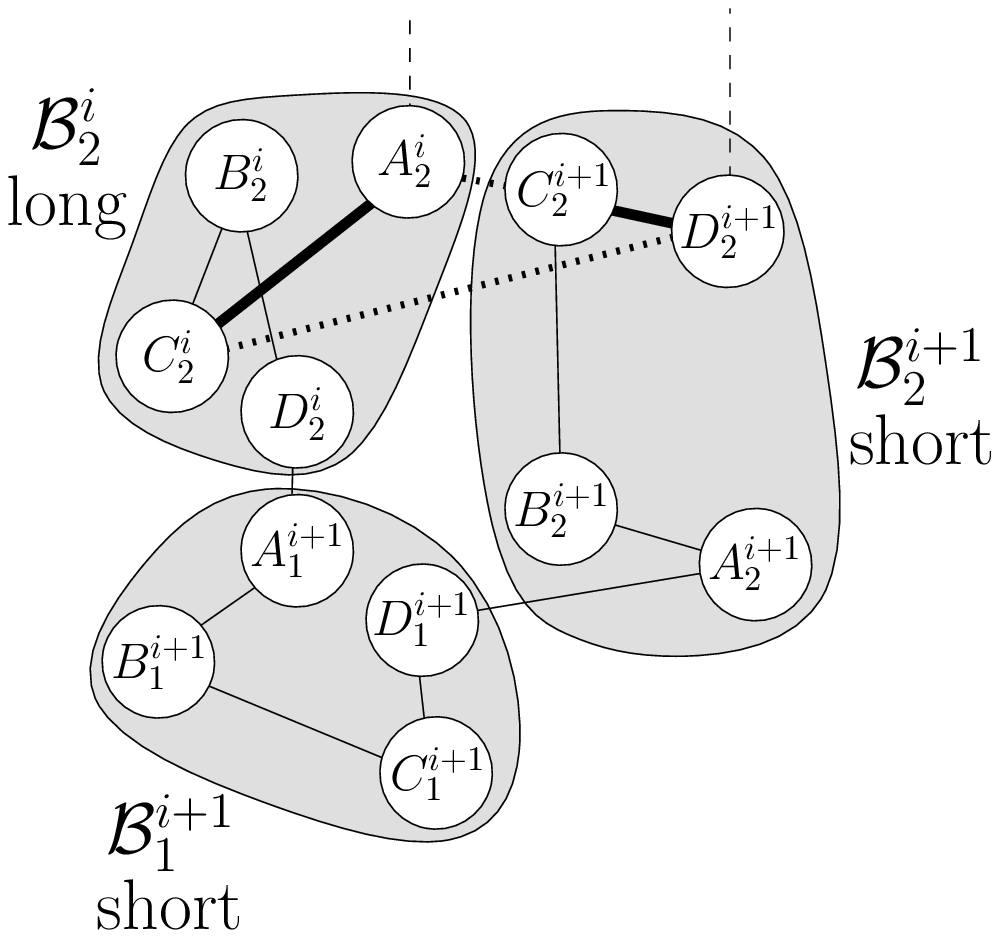}}\hspace{0.1cm}
\subfloat[]{\includegraphics[scale=0.38]{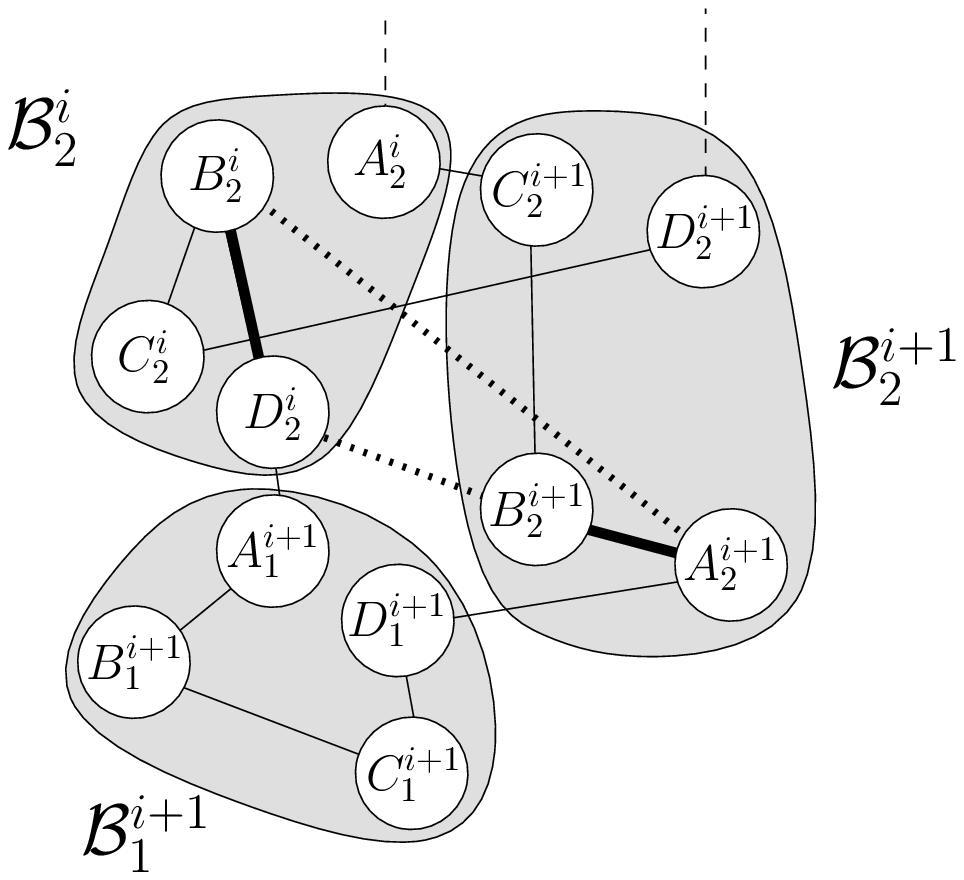}}\hspace{0.1cm}
\subfloat[]{\includegraphics[scale=0.38]{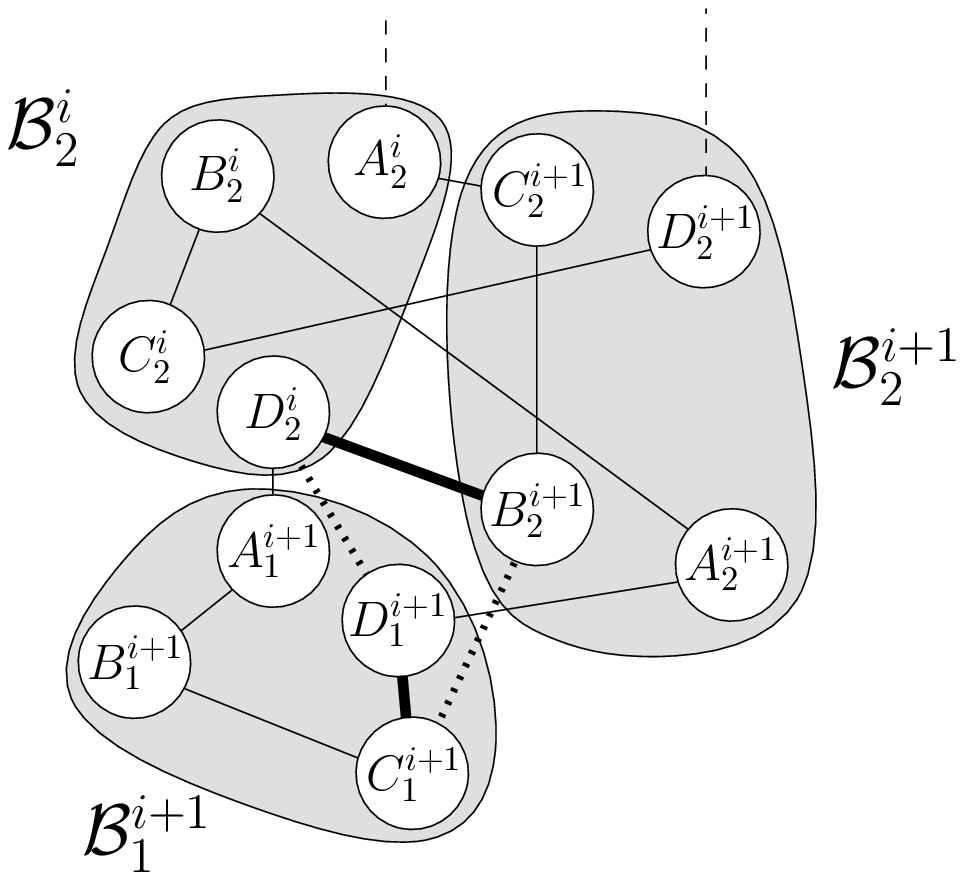}}\\
\subfloat[]{\includegraphics[scale=0.38]{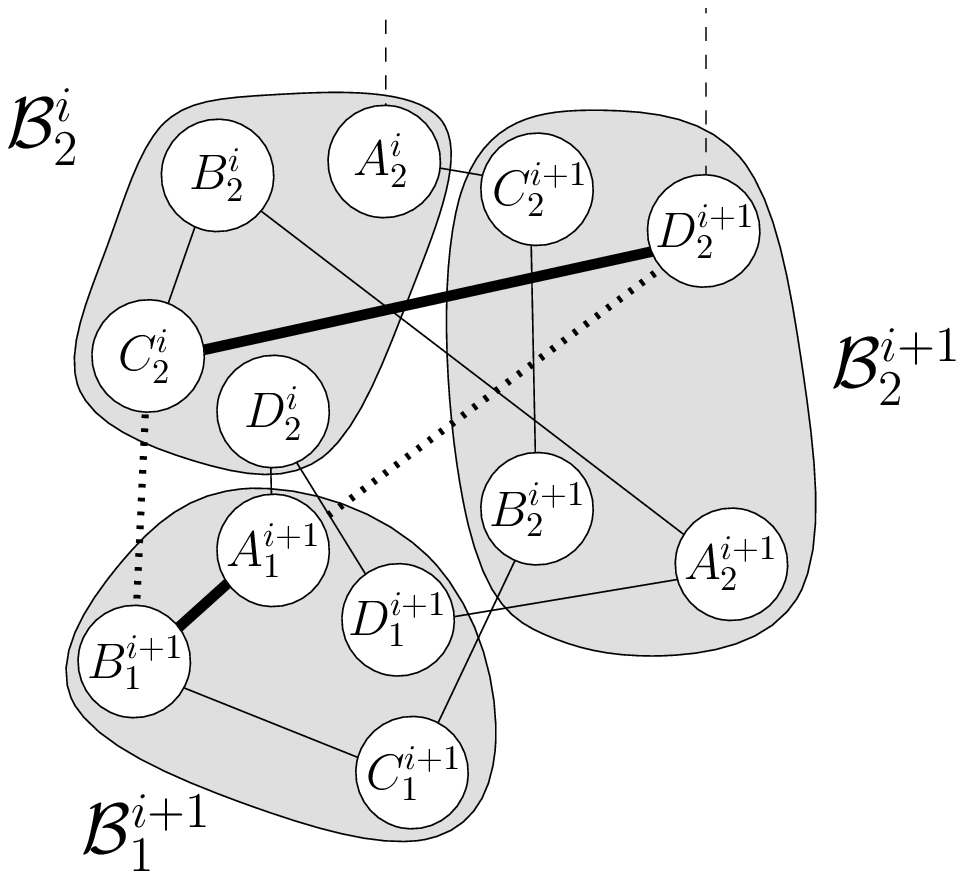}}\hspace{0.1cm}
\subfloat[]{\includegraphics[scale=0.38]{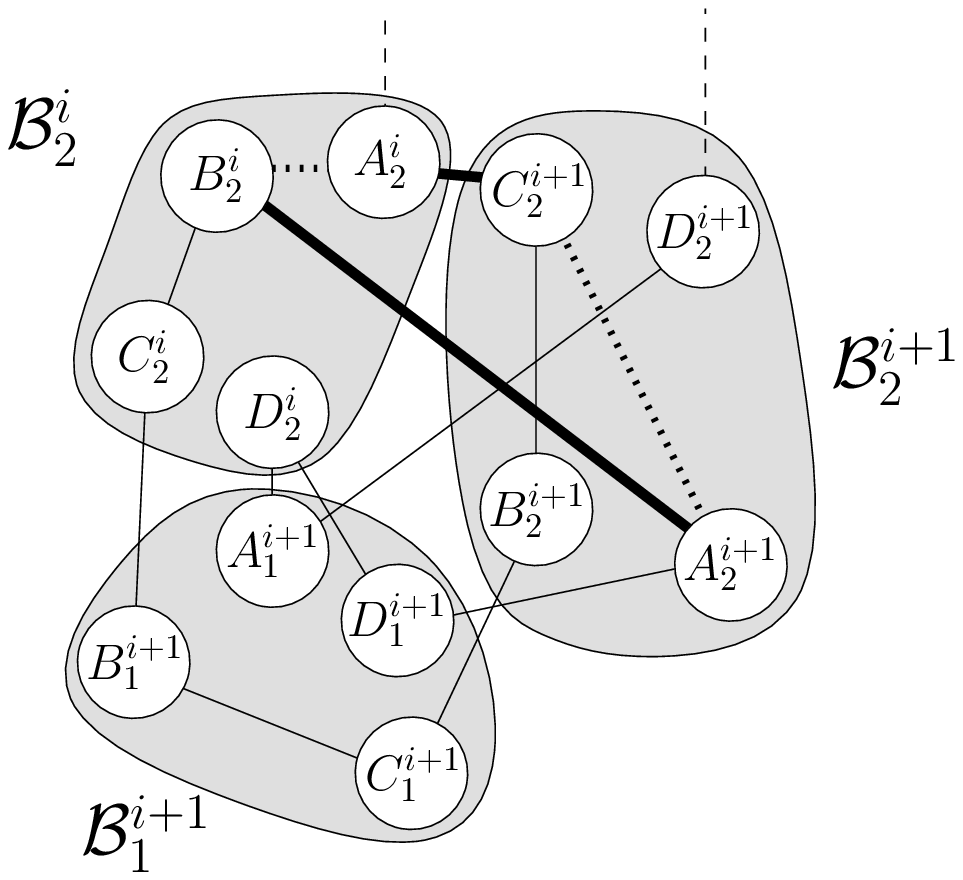}}\hspace{0.1cm}
\subfloat[]{\includegraphics[scale=0.38]{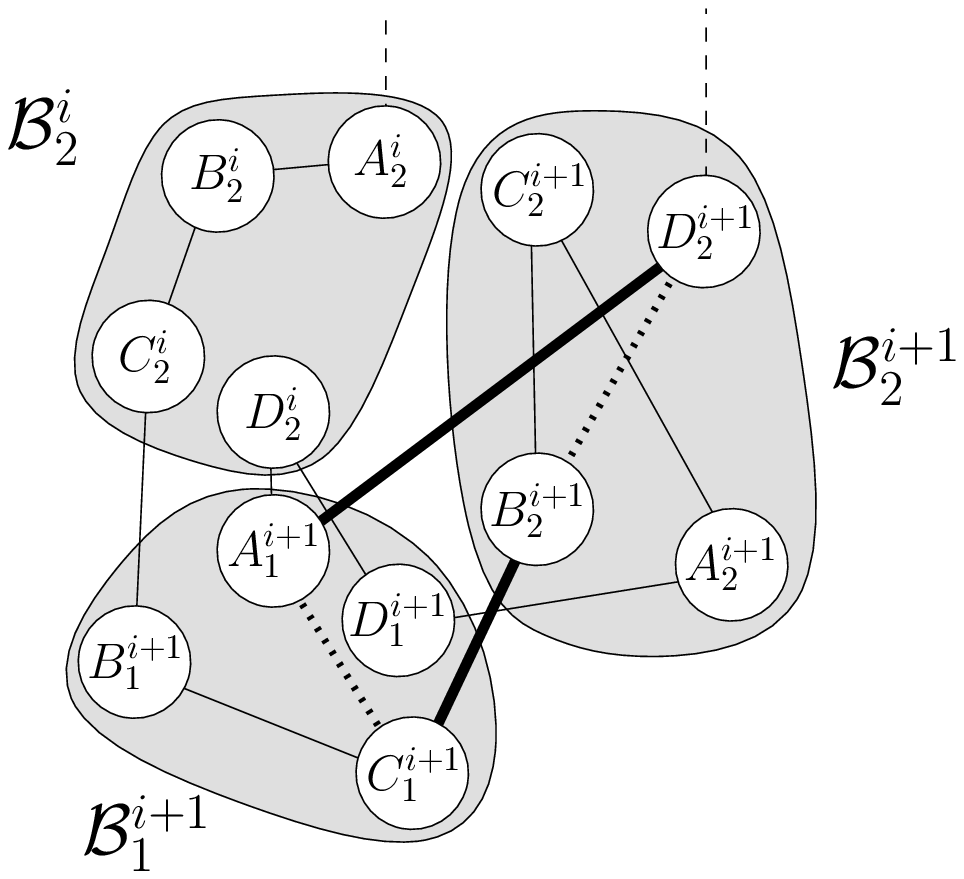}}\\
\subfloat[]{\includegraphics[scale=0.38]{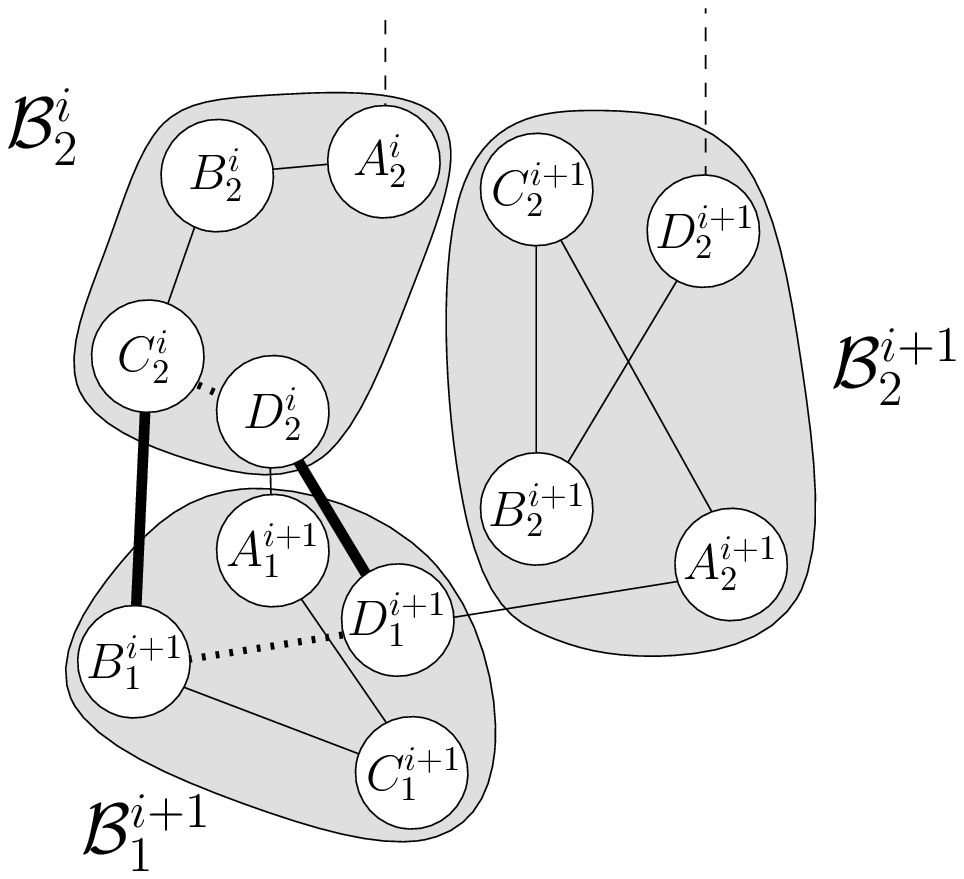}}\hspace{0.1cm}
\subfloat[]{\includegraphics[scale=0.38]{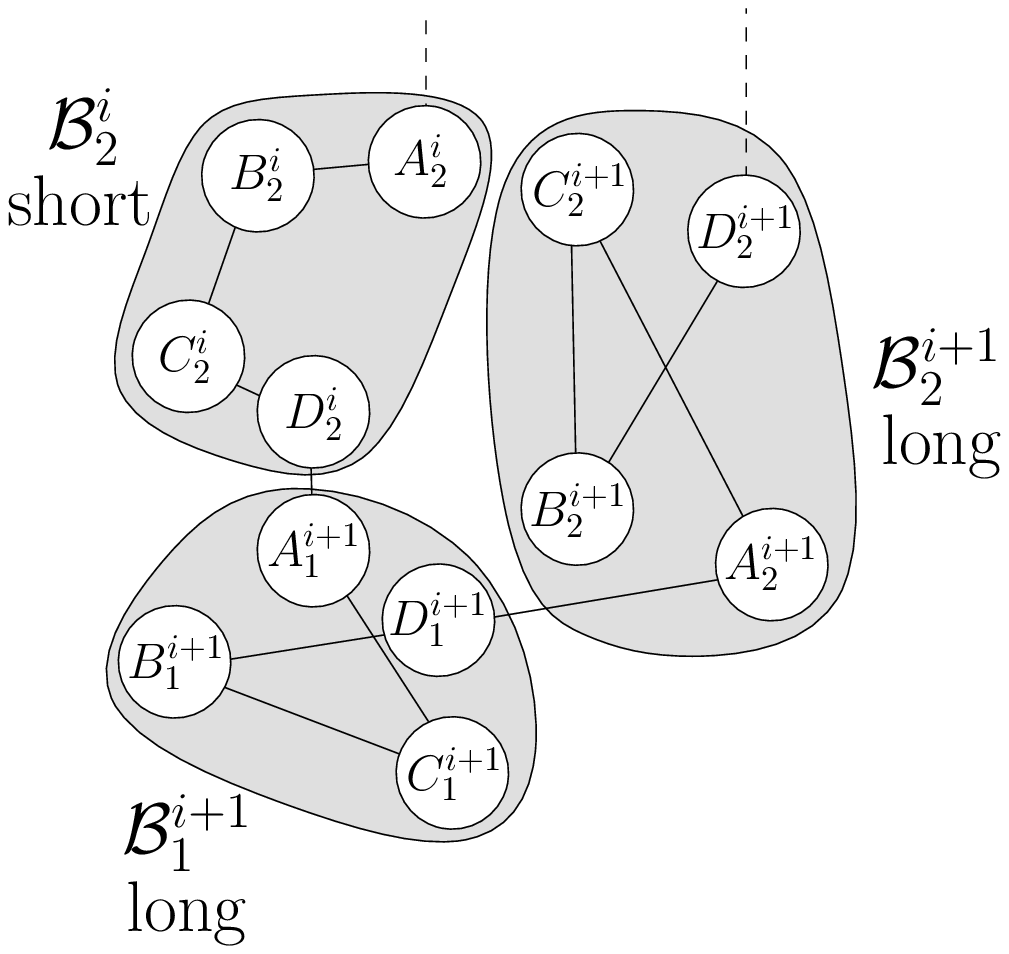}}
\caption{This figure shows the sequence of seven consecutive 2-changes from
Property~\ref{property:LB:Sequence}. In each step the thick edges are removed
from the tour, and the dotted edges are added to the tour. 
It shows how block ${\cal B}_2^i$
switches from its long to its short state while resetting the blocks
${\cal B}_1^{i+1}$ and ${\cal B}_2^{i+1}$ from their short to their long states.
This figure is only schematic and it does not show the actual geometric embedding
of the points into the Euclidean plane.}
\label{fig:LB:Reset}
\end{center}
\end{figure}

\begin{center}
\small
\begin{tabular}{lcccc|cccc|cccc}
& \multicolumn{4}{c}{Long state ACBD} & \multicolumn{4}{c}{Short state ABCD} & \multicolumn{4}{c}{Short state ABCD} \\

1) \hspace{-0.3cm} & $\left[A^{i}_2\right.$
& \hspace{-0.5cm} $\left.C^{i}_2\right]$
& \hspace{-0.5cm} $B^{i}_2$
& \hspace{-0.5cm} $D^{i}_2$ &

  $A^{i+1}_1$
& \hspace{-0.5cm} $B^{i+1}_1$
& \hspace{-0.5cm} $C^{i+1}_1$
& \hspace{-0.5cm} $D^{i+1}_1$ &

  $A^{i+1}_2$
& \hspace{-0.5cm} $B^{i+1}_2$
& \hspace{-0.5cm} $\left[C^{i+1}_2\right.$
& \hspace{-0.5cm} $\left.D^{i+1}_2\right]$ \\[0.1cm]

2) \hspace{-0.3cm} & $A^{i}_2$
& \hspace{-0.5cm} $C^{i+1}_2$
& \hspace{-0.5cm} $\left[B^{i+1}_2\right.$
& \hspace{-0.5cm} $\left.A^{i+1}_2\right]$ &

  $D^{i+1}_1$
& \hspace{-0.5cm} $C^{i+1}_1$
& \hspace{-0.5cm} $B^{i+1}_1$
& \hspace{-0.5cm} $A^{i+1}_1$ &

  $\left[D^{i}_2\right.$
& \hspace{-0.5cm} $\left.B^{i}_2\right]$
& \hspace{-0.5cm} $C^{i}_2$
& \hspace{-0.5cm} $D^{i+1}_2$ \\[0.1cm]

3) \hspace{-0.3cm} & $A^{i}_2$
& \hspace{-0.5cm} $C^{i+1}_2$
& \hspace{-0.5cm} $\left[B^{i+1}_2\right.$
& \hspace{-0.5cm} $\left.D^{i}_2\right]$ &

  $A^{i+1}_1$
& \hspace{-0.5cm} $B^{i+1}_1$
& \hspace{-0.5cm} $\left[C^{i+1}_1\right.$
& \hspace{-0.5cm} $\left.D^{i+1}_1\right]$ &

  $A^{i+1}_2$
& \hspace{-0.5cm} $B^{i}_2$
& \hspace{-0.5cm} $C^{i}_2$
& \hspace{-0.5cm} $D^{i+1}_2$ \\[0.1cm]

4) \hspace{-0.3cm} & $A^{i}_2$
& \hspace{-0.5cm} $C^{i+1}_2$
& \hspace{-0.5cm} $B^{i+1}_2$
& \hspace{-0.5cm} $C^{i+1}_1$ &

  $\left[B^{i+1}_1\right.$
& \hspace{-0.5cm} $\left.A^{i+1}_1\right]$
& \hspace{-0.5cm} $D^{i}_2$
& \hspace{-0.5cm} $D^{i+1}_1$ &
  
  $A^{i+1}_2$
& \hspace{-0.5cm} $B^{i}_2$
& \hspace{-0.5cm} $\left[C^{i}_2\right.$
& \hspace{-0.5cm} $\left.D^{i+1}_2\right]$ \\[0.1cm]

5) \hspace{-0.3cm} & $\left[A^{i}_2\right.$
& \hspace{-0.5cm} $\left.C^{i+1}_2\right]$
& \hspace{-0.5cm} $B^{i+1}_2$
& \hspace{-0.5cm} $C^{i+1}_1$ &

  $B^{i+1}_1$
& \hspace{-0.5cm} $C^{i}_2$
& \hspace{-0.5cm} $\left[B^{i}_2\right.$
& \hspace{-0.5cm} $\left.A^{i+1}_2\right]$ &

  $D^{i+1}_1$
& \hspace{-0.5cm} $D^{i}_2$
& \hspace{-0.5cm} $A^{i+1}_1$
& \hspace{-0.5cm} $D^{i+1}_2$ \\[0.1cm]

6) \hspace{-0.3cm} & $A^{i}_2$
& \hspace{-0.5cm} $B^{i}_2$
& \hspace{-0.5cm} $C^{i}_2$
& \hspace{-0.5cm} $B^{i+1}_1$ &

  $\left[C^{i+1}_1\right.$
& \hspace{-0.5cm} $\left.B^{i+1}_2\right]$
& \hspace{-0.5cm} $C^{i+1}_2$
& \hspace{-0.5cm} $A^{i+1}_2$ &

  $D^{i+1}_1$
& \hspace{-0.5cm} $D^{i}_2$
& \hspace{-0.5cm} $\left[A^{i+1}_1\right.$
& \hspace{-0.5cm} $\left.D^{i+1}_2\right]$ \\[0.1cm]

7) \hspace{-0.3cm} &  $A^{i}_2$
& \hspace{-0.5cm} $B^{i}_2$
& \hspace{-0.5cm} $\left[C^{i}_2\right.$
& \hspace{-0.5cm} $\left.B^{i+1}_1\right]$ &

  $C^{i+1}_1$
& \hspace{-0.5cm} $A^{i+1}_1$
& \hspace{-0.5cm} $\left[D^{i}_2\right.$
& \hspace{-0.5cm} $\left.D^{i+1}_1\right]$ &

  $A^{i+1}_2$
& \hspace{-0.5cm} $C^{i+1}_2$
& \hspace{-0.5cm} $B^{i+1}_2$
& \hspace{-0.5cm} $D^{i+1}_2$ \\[0.1cm]

&  $A^{i}_2$
& \hspace{-0.5cm} $B^{i}_2$
& \hspace{-0.5cm} $C^{i}_2$
& \hspace{-0.5cm} $D^{i}_2$ &

  $A^{i+1}_1$
& \hspace{-0.5cm} $C^{i+1}_1$
& \hspace{-0.5cm} $B^{i+1}_1$
& \hspace{-0.5cm} $D^{i+1}_1$ &

  $A^{i+1}_2$
& \hspace{-0.5cm} $C^{i+1}_2$
& \hspace{-0.5cm} $B^{i+1}_2$
& \hspace{-0.5cm} $D^{i+1}_2$\\

& \multicolumn{4}{c}{Short state ABCD} & \multicolumn{4}{c}{Long state ACBD} & \multicolumn{4}{c}{Long state ACBD}
\end{tabular}
\end{center}

Observe that the configurations~2 to~7 do not have the property mentioned at the beginning
of this section that, for every block~${\cal B}^i_j$, the points~$A^i_j$, $B^i_j$, $C^i_j$, and~$D^i_j$ are
visited consecutively either in the order~$A^i_j B^i_j C^i_j D^i_j$ or in the order~$A^i_j C^i_j B^i_j D^i_j$.
The configurations~2 to~7 are exactly the intermediate configurations that we mentioned at the beginning
of this section.

If gadget $G_i$ is in state $(L,L)$ instead of state $(S,L)$, a 
sequence of steps that satisfies Property~\ref{property:LB:Sequence} 
can be constructed analogously. Additionally, one has to take into 
account that the three involved blocks ${\cal B}^{i}_1$, ${\cal 
B}^{i+1}_1$, and ${\cal B}^{i+1}_2$ are not consecutive in the tour but 
that block ${\cal B}^{i}_2$ lies between them. However, one can easily 
verify that this block is not affected by the sequence of 2-changes, as 
after the seven 2-changes have been performed, the block is in the same 
state and at the same position as before.

\subsubsection{Embedding the construction into the Euclidean plane}
\label{subsubsec:LB:L2Points}

The only missing step in the proof of Theorem~\ref{theorem:LowerBounds} 
for the Euclidean plane is to find points such that all of the 
2-changes that we described in the previous section are improving. We 
specify the positions of the points of gadget $G_{n-1}$ and give a rule 
as to how the points of gadget $G_i$ can be derived when all points of gadget 
$G_{i+1}$ have already been placed. In our construction it happens that 
different points have exactly the same coordinates. This is only for 
ease of notation; if one wants to obtain a TSP instance in which 
distinct points have distinct coordinates, one can slightly move these 
points without affecting the property that all 2-changes are improving.
 
For $j\in[2]$, we choose $A^{n-1}_j=(0,0)$, $B^{n-1}_j=(1,0)$, 
$C^{n-1}_j=(-0.1,1.4)$, and $D^{n-1}_j=(-1.1,4.8)$. Then 
$A^{n-1}_jB^{n-1}_jC^{n-1}_jD^{n-1}_j$ is the short state 
and $A^{n-1}_jC^{n-1}_jB^{n-1}_jD^{n-1}_j$ is the long state
because
\[
\dist(A^{n-1}_j,C^{n-1}_j)+\dist(B^{n-1}_j,D^{n-1}_j)
> \dist(A^{n-1}_j,B^{n-1}_j)+\dist(C^{n-1}_j,D^{n-1}_j),
\]
as
\[
  \dist(A^{n-1}_j,C^{n-1}_j)+\dist(B^{n-1}_j,D^{n-1}_j) 
  = \sqrt{0.1^2+1.4^2}+\sqrt{2.1^2+4.8^2} > 6.64
\]
and
\[
  \dist(A^{n-1}_j,B^{n-1}_j)+\dist(C^{n-1}_j,D^{n-1}_j) 
  = \sqrt{1^2+0^2}+\sqrt{1^2+3.4^2} < 4.55.
\]
We place the points of gadget $G_{i}$ as follows (see 
Figure~\ref{figure:counter:PointsL2}):
\begin{enumerate}
\setlength{\itemsep}{0em}
\item Start with the coordinates of the points of gadget $G_{i+1}$.
\item Rotate these points around the origin by $3\pi/2$.
\item Scale each coordinate by a factor of 3.
\item Translate the points by the vector $(-1.2,0.1)$.
\end{enumerate}
For $j\in[2]$, this yields $A^{n-2}_j=(-1.2,0.1)$, 
$B^{n-2}_j=(-1.2,-2.9)$, $C^{n-2}_j=(3,0.4)$, and 
$D^{n-2}_j=(13.2,3.4)$.

\begin{figure}[H]
\begin{center}
\includegraphics{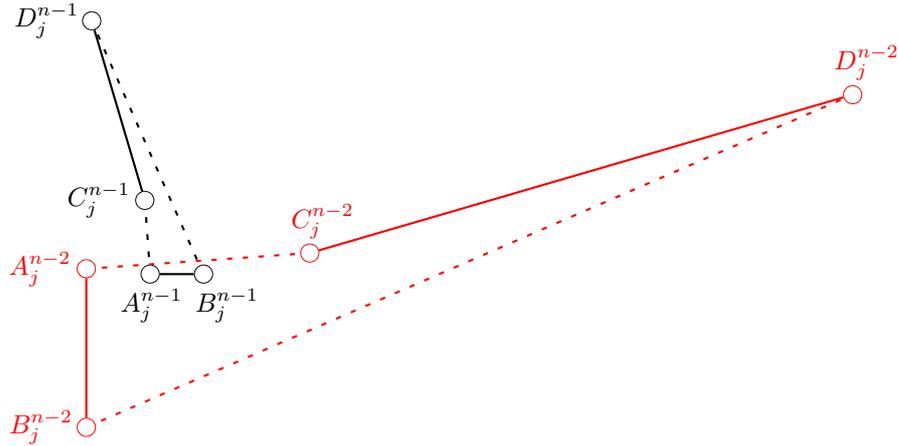}
\end{center}
\caption{This illustration shows the points of the gadgets $G_{n-1}$ 
and $G_{n-2}$. One can see that $G_{n-2}$ is a scaled, rotated, and 
translated copy of $G_{n-1}$.}
\label{figure:counter:PointsL2}
\end{figure}

From this construction it follows that each gadget is a scaled, 
rotated, and translated copy of gadget $G_{n-1}$. If one has a set of 
points in the Euclidean plane that admits certain improving 2-changes, 
then these 2-changes are still improving if one scales, rotates, and 
translates all points in the same manner. Hence, it suffices to show 
that the sequences in which gadget $G_{n-2}$ resets gadget $G_{n-1}$ 
from $(S,S)$ to $(L,L)$ are improving because, for any~$i$, the points
of the gadgets~$G_{i}$ and~$G_{i+1}$ are a scaled, rotated, and translated
copy of the points of the gadgets~$G_{n-2}$ and~$G_{n-1}$.

There are two sequences in which gadget $G_{n-2}$ resets gadget $G_{n-1}$ 
from $(S,S)$ to $(L,L)$: in 
the first one, gadget $G_{n-2}$ changes its state from $(L,L)$ to 
$(S,L)$, in the second one, gadget $G_{n-2}$ changes its state from 
$(S,L)$ to $(S,S)$. Since the coordinates of the points in both blocks 
of gadget $G_{n-2}$ are the same, the inequalities for both sequences 
are also identical. The following inequalities show that the
improvements made by the steps in both sequences are all positive
(see Figure~\ref{fig:LB:Reset} or the table in Section~\ref{subsubsec:LB:L2Sequence}
for the sequence of 2-changes): 

{\small
\begin{alignat*}{5}
  \text{1)\,\,\,} &  \dist(A^{n-2}_2,C^{n-2}_2) && + \dist(C^{n-1}_2,D^{n-1}_2)
                  && - \dist(A^{n-2}_2,C^{n-1}_2) && - \dist(C^{n-2}_2,D^{n-1}_2) \, && > \, 0.03,\\
  \text{2)\,\,\,} &  \dist(B^{n-1}_2,A^{n-1}_2) && + \dist(D^{n-2}_2,B^{n-2}_2)
                  && - \dist(B^{n-1}_2,D^{n-2}_2) && - \dist(A^{n-1}_2,B^{n-2}_2) && > \, 0.91,\\
  \text{3)\,\,\,} &  \dist(B^{n-1}_2,D^{n-2}_2) && + \dist(C^{n-1}_1,D^{n-1}_1)
                  && - \dist(B^{n-1}_2,C^{n-1}_1) && - \dist(D^{n-2}_2,D^{n-1}_1) && > \, 0.06,\\
  \text{4)\,\,\,} &  \dist(B^{n-1}_1,A^{n-1}_1) && + \dist(C^{n-2}_2,D^{n-1}_2)
                  && - \dist(B^{n-1}_1,C^{n-2}_2) && - \dist(A^{n-1}_1,D^{n-1}_2) && > \, 0.05,\\
  \text{5)\,\,\,} &  \dist(A^{n-2}_2,C^{n-1}_2) && + \dist(B^{n-2}_2,A^{n-1}_2)
                  && - \dist(A^{n-2}_2,B^{n-2}_2) && - \dist(C^{n-1}_2,A^{n-1}_2) && > \, 0.43,\\                                    
  \text{6)\,\,\,} &  \dist(C^{n-1}_1,B^{n-1}_2) && + \dist(A^{n-1}_1,D^{n-1}_2)
                  && - \dist(C^{n-1}_1,A^{n-1}_1) && - \dist(B^{n-1}_2,D^{n-1}_2) && > \, 0.06,\\
  \text{7)\,\,\,} &  \dist(C^{n-2}_2,B^{n-1}_1) && + \dist(D^{n-2}_2,D^{n-1}_1)
                  && - \dist(C^{n-2}_2,D^{n-2}_2) && - \dist(B^{n-1}_1,D^{n-1}_1) && > \, 0.53.
\end{alignat*}}
This concludes the proof of 
Theorem~\ref{theorem:LowerBounds} for the Euclidean plane as it shows 
that all 2-changes in Lemma~\ref{lemma:LB:ExpSequence} are improving.

\subsection{Exponential Lower Bound for $L_p$ Metrics}
\label{subsec:LB:Lp}

We were not able to find a set of points in the plane such that all 
2-changes in Lemma~\ref{lemma:LB:ExpSequence} are improving with 
respect to the Manhattan metric. Therefore, we modify the construction 
of the gadgets and the sequence of 2-changes. Our construction for the 
Manhattan metric is based on the construction for the Euclidean plane, 
but it does not possess the property that every gadget resets its 
neighboring gadget twice. This property is only true for half of the 
gadgets. To be more precise, we construct two different types of 
gadgets which we call \emph{reset gadgets} and \emph{propagation 
gadgets}. Reset gadgets perform the same sequence of 2-changes as the 
gadgets that we constructed for the Euclidean plane. Propagation 
gadgets also have the same structure as the gadgets for the Euclidean 
plane, but when such a gadget changes its state from $(L,L)$ to 
$(S,S)$, it resets its neighboring gadget only once. Due to this 
relaxed requirement it is possible to find points in the Manhattan 
plane whose distances satisfy all necessary inequalities. Instead of 
$n$ gadgets, our construction consists of $2n$ gadgets, namely $n$ 
propagation gadgets $G_0^P,\ldots,G_{n-1}^P$ and $n$ reset gadgets 
$G_0^R,\ldots,G_{n-1}^R$. The order in which these gadgets appear in 
the tour is $G_0^PG_0^RG_1^PG_1^R\ldots G_{n-1}^PG_{n-1}^R$.

As before, every gadget consists of two blocks and the order in which 
the blocks and the gadgets are visited does not change during the 
sequence of 2-changes. Consider a reset gadget $G^R_i$ and its 
neighboring propagation gadget $G^P_{i+1}$. We will embed the points
of the gadgets into the Manhattan plane in such a way that 
Property~\ref{property:LB:Sequence} is still satisfied. That is, if 
$G^R_i$ is in state $(L,L)$ (or state $(S,L)$, respectively) and $G^P_{i+1}$ is in state 
$(S,S)$, then there exists a sequence of seven consecutive 2-changes 
resetting gadget $G^P_{i+1}$ to state $(L,L)$ and leaving gadget 
$G^R_i$ in state $(S,L)$ (or $(S,S)$, respectively). The situation is 
different for a propagation gadget $G^P_i$ and its neighboring reset 
gadget $G^R_i$. In this case, if $G^P_i$ is in state $(L,L)$, it first 
changes its state with a single 2-change to $(S,L)$. After that, gadget 
$G^P_i$ changes its state to $(S,S)$ while resetting gadget $G^R_i$ 
from state $(S,S)$ to state $(L,L)$ by a sequence of seven consecutive 
2-changes. In both cases, the sequences of 2-changes in which one block 
changes from its long to its short state while resetting two blocks of
the neighboring gadget from their short to their long states are chosen
analogously to the ones for the Euclidean plane described in 
Section~\ref{subsubsec:LB:L2Sequence}.
An example with three propagation
and three reset gadgets is shown in Figure~\ref{fig:LB:ExampleB}.

In the initial tour, only gadget $G^P_0$ is in state $(L,L)$ and every 
other gadget is in state $(S,S)$. With similar arguments as for the 
Euclidean plane, we can show that gadget $G_i^R$ is reset from its one 
state $(S,S)$ to its zero state $(L,L)$ $2^{i}$ times and that the 
total number of steps is $2^{n+4}-22$. 

\begin{figure}[H]
\begin{center}
\includegraphics[scale=0.79]{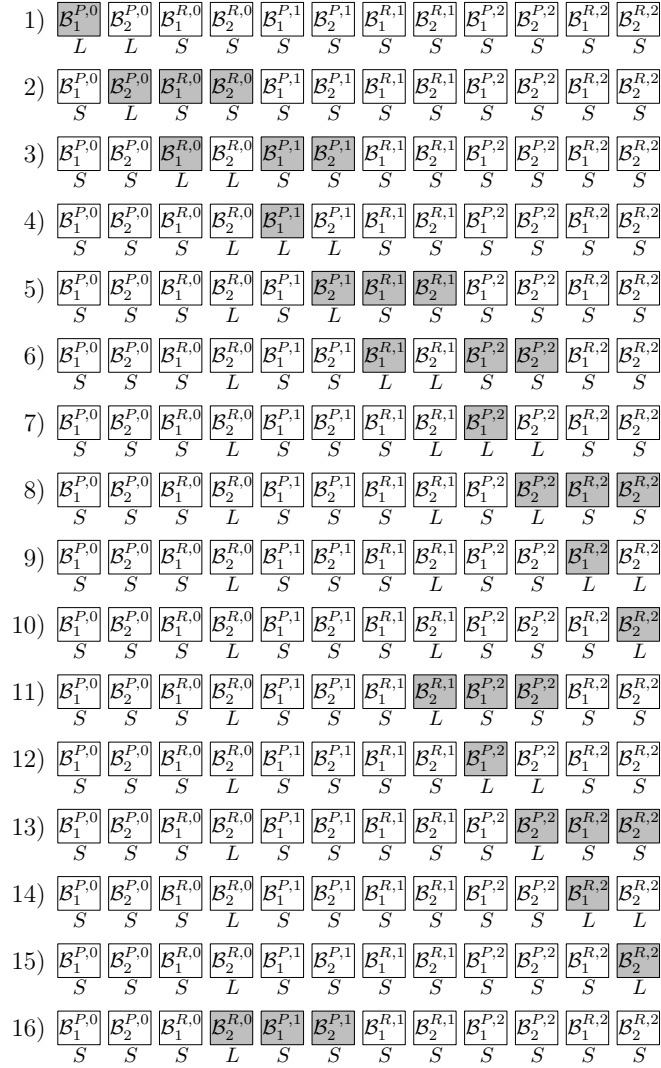}
\caption{This figure shows an example with three propagation and three reset gadgets.
It shows the first 16 configurations that these gadgets assume during the sequence of 2-changes,
excluding the intermediate configurations that arise when one gadget resets 
another one. Gadgets that are involved in the transformation from configuration~$i$
to configuration~$i+1$ are shown in gray. For example, in the step from the
first to the second configuration, the first block ${\cal B}_{1}^{P,0}$ of
the first propagation gadget $G_0^P$ switches from its long to its short state by a 
single 2-change. Then in the step from the second to the third configuration, 
the second block ${\cal B}_{2}^{P,0}$ of the first propagation gadget $G_0^P$
resets the two blocks of the first reset gadget $G_0^R$. That is, these three
blocks follow the sequence of seven 2-changes from Property~\ref{property:LB:Sequence}.}
\label{fig:LB:ExampleB}
\end{center}
\end{figure}

\subsubsection{Embedding the construction into the Manhattan plane}
\label{subsubsec:LB:L1Points}

As in the construction in the Euclidean plane, the points in both 
blocks of a reset gadget $G_i^R$ have the same coordinates. Also in 
this case one can slightly move all the points without affecting the 
inequalities if one wants distinct coordinates for distinct points. 
Again, we choose points for the gadgets $G_{n-1}^P$ and $G_{n-1}^R$ and 
describe how the points of the gadgets $G_i^P$ and $G_i^R$ can be 
chosen when the points of the gadgets $G_{i+1}^P$ and $G_{i+1}^R$ are 
already chosen. For $j\in[2]$, we choose $A^{n-1}_{R,j}=(0,1)$, 
$B^{n-1}_{R,j}=(0,0)$, $C^{n-1}_{R,j}=(-0.7,0.1)$, and 
$D^{n-1}_{R,j}=(-1.2,0.08)$. Furthermore, we choose 
$A^{n-1}_{P,1}=(-2,1.8)$, $B^{n-1}_{P,1}=(-3.3,2.8)$, 
$C^{n-1}_{P,1}=(-1.3,1.4)$, $D^{n-1}_{P,1}=(1.5,0.9)$, 
$A^{n-1}_{P,2}=(-0.7,1.6)$, $B^{n-1}_{P,2}=(-1.5,1.2)$, 
$C^{n-1}_{P,2}=(1.9,-1.5)$, and $D^{n-1}_{P,2}=(-0.8,-1.1)$.

Before we describe how the points of the other gadgets are chosen, we 
first show that the 2-changes within and between the gadgets 
$G_{n-1}^P$ and $G_{n-1}^R$ are improving. For $j\in[2]$, 
$A^{n-1}_{R,j}B^{n-1}_{R,j}C^{n-1}_{R,j}D^{n-1}_{R,j}$ is the short 
state because
\begin{align*}
& \quad \dist(A^{n-1}_{R,j},C^{n-1}_{R,j})+\dist(B^{n-1}_{R,j},D^{n-1}_{R,j})-
(\dist(A^{n-1}_{R,j},B^{n-1}_{R,j})+\dist(C^{n-1}_{R,j},D^{n-1}_{R,j}))\\
& = (0.7+0.9) + (1.2+0.08) - (0+1) - (0.5+0.02)  
= 1.36.
\end{align*}

In the 2-change in which $G_{n-1}^P$ changes its state from $(L,L)$ to
$(S,L)$ the edges~$A^{n-1}_{P,1},C^{n-1}_{P,1}$
and~$B^{n-1}_{P,1},D^{n-1}_{P,1}$ are replaced with the edges~$A^{n-1}_{P,1},B^{n-1}_{P,1}$
and~$C^{n-1}_{P,1},D^{n-1}_{P,1}$. This 2-change is improving because
\begin{align*}
& \quad \dist(A^{n-1}_{P,1},C^{n-1}_{P,1})+\dist(B^{n-1}_{P,1},D^{n-1}_{P,1})-
(\dist(A^{n-1}_{P,1},B^{n-1}_{P,1})+\dist(C^{n-1}_{P,1},D^{n-1}_{P,1}))\\
& = (0.7+0.4)+(4.8+1.9)-(1.3+1)-(2.8+0.5)
= 2.2. 
\end{align*}
The 2-changes in the sequence in which~$G_{n-1}^P$ changes its state from $(S,L)$ to $(S,S)$ while resetting 
$G_{n-1}^R$ are chosen analogously to the ones shown in Figure~\ref{fig:LB:Reset} and in the table in
Section~\ref{subsubsec:LB:L2Sequence}. The only difference is that the involved blocks are
not~${\cal B}^{i}_2$, ${\cal B}^{i+1}_1$, and ${\cal B}^{i+1}_2$ anymore, but the second block of gadget~$G_{n-1}^P$
and the two blocks of gadget~$G_{n-1}^R$, respectively. This gives rise to the following equalities that show that  
the improvements made by the 2-changes in this sequence are all positive:

{\small
 \begin{alignat*}{5}
  \text{1)\,\,\,} &  \dist(A^{n-1}_{P,2},C^{n-1}_{P,2}) && + \dist(C^{n-1}_{R,2},D^{n-1}_{R,2})
                  && - \dist(A^{n-1}_{P,2},C^{n-1}_{R,2}) && - \dist(C^{n-1}_{P,2},D^{n-1}_{R,2}) \, && = \, 0.04,\\
  \text{2)\,\,\,} &  \dist(B^{n-1}_{R,2},A^{n-1}_{R,2}) && + \dist(D^{n-1}_{P,2},B^{n-1}_{P,2})
                  && - \dist(B^{n-1}_{R,2},D^{n-1}_{P,2}) && - \dist(A^{n-1}_{R,2},B^{n-1}_{P,2}) && = \, 0.4,\\
  \text{3)\,\,\,} &  \dist(B^{n-1}_{R,2},D^{n-1}_{P,2}) && + \dist(C^{n-1}_{R,1},D^{n-1}_{R,1})
                  && - \dist(B^{n-1}_{R,2},C^{n-1}_{R,1}) && - \dist(D^{n-1}_{P,2},D^{n-1}_{R,1}) && = \, 0.04,\\
  \text{4)\,\,\,} &  \dist(B^{n-1}_{R,1},A^{n-1}_{R,1}) && + \dist(C^{n-1}_{P,2},D^{n-1}_{R,2})
                  && - \dist(B^{n-1}_{R,1},C^{n-1}_{P,2}) && - \dist(A^{n-1}_{R,1},D^{n-1}_{R,2}) && = \, 0.16,\\
  \text{5)\,\,\,} &  \dist(A^{n-1}_{P,2},C^{n-1}_{R,2}) && + \dist(B^{n-1}_{P,2},A^{n-1}_{R,2})
                  && - \dist(A^{n-1}_{P,2},B^{n-1}_{P,2}) && - \dist(C^{n-1}_{R,2},A^{n-1}_{R,2}) && = \, 0.4,\\                                    
  \text{6)\,\,\,} &  \dist(C^{n-1}_{R,1},B^{n-1}_{R,2}) && + \dist(A^{n-1}_{R,1},D^{n-1}_{R,2})
                  && - \dist(C^{n-1}_{R,1},A^{n-1}_{R,1}) && - \dist(B^{n-1}_{R,2},D^{n-1}_{R,2}) && = \, 0.04,\\
  \text{7)\,\,\,} &  \dist(C^{n-1}_{P,2},B^{n-1}_{R,1}) && + \dist(D^{n-1}_{P,2},D^{n-1}_{R,1})
                  && - \dist(C^{n-1}_{P,2},D^{n-1}_{P,2}) && - \dist(B^{n-1}_{R,1},D^{n-1}_{R,1}) && = \, 0.6.
 \end{alignat*}}

Again, our construction possesses the property that each pair of 
gadgets $G_i^P$ and $G_i^R$ is a scaled and translated version of 
the pair $G_{n-1}^P$ and $G_{n-1}^R$. Since we have relaxed the 
requirements for the gadgets, we do not even need rotations here. 
We place the points of $G_i^P$ and $G_i^R$ as follows:
\begin{enumerate}
\setlength{\itemsep}{0em}
\item Start with the coordinates specified for the points of 
gadgets $G_{i+1}^P$ and $G_{i+1}^R$.
\item Scale each coordinate by a factor of 7.7.
\item Translate the points by the vector $(1.93,0.3)$.
\end{enumerate}
For $j\in[2]$, this yields $A^{n-2}_{R,j}=(1.93,8)$, 
$B^{n-2}_{R,j}=(1.93,0.3)$, $C^{n-2}_{R,j}=(-3.46,1.07)$, and 
$D^{n-2}_{R,j}=(-7.31,0.916)$. 
Additionally, it yields
$A^{n-2}_{P,1}=(-13.47,14.16)$, $B^{n-2}_{P,1}=(-23.48,21.86)$, 
$C^{n-2}_{P,1}=(-8.08,11.08)$, $D^{n-2}_{P,1}=(13.48,7.23)$, 
$A^{n-2}_{P,2}=(-3.46,12.62)$, $B^{n-2}_{P,2}=(-9.62,9.54)$, 
$C^{n-2}_{P,2}=(16.56,-11.25)$, and $D^{n-2}_{P,2}=(-4.23,-8.17)$.

As in our construction for the 
Euclidean plane, it suffices to show that the sequences in which gadget 
$G^R_{n-2}$ resets gadget $G^P_{n-1}$ from $(S,S)$ to $(L,L)$ are 
improving because, for any~$i$, the points of the gadgets~$G^R_{i}$ and $G^P_{i+1}$ are
a scaled and translated copy of the points of the gadgets~$G^R_{n-2}$ and $G^P_{n-1}$.
The 2-changes in these sequences 
are chosen analogously to the ones shown in Figure~\ref{fig:LB:Reset} and in the table in
Section~\ref{subsubsec:LB:L2Sequence}. The only difference is that the involved blocks are
not~${\cal B}^{i}_2$, ${\cal B}^{i+1}_1$, and ${\cal B}^{i+1}_2$ anymore, but one of the blocks of gadget~$G^R_{n-2}$
and the two blocks of gadget~$G^P_{n-1}$, respectively.
As the coordinates of the points in the two blocks of gadget 
$G^R_{n-2}$ are the same, the inequalities for both sequences are also 
identical. The improvements made by the steps in both sequences are

{\small 
\begin{alignat*}{5}
\text{1)\,\,\,} &  \dist(A^{n-2}_{R,2},C^{n-2}_{R,2}) && + \dist(C^{n-1}_{P,2},D^{n-1}_{P,2})
&& - \dist(A^{n-2}_{R,2},C^{n-1}_{P,2}) && - \dist(C^{n-2}_{R,2},D^{n-1}_{P,2}) \, && = \, 1.06\\
\text{2)\,\,\,} &  \dist(B^{n-1}_{P,2},A^{n-1}_{P,2}) && + \dist(D^{n-2}_{R,2},B^{n-2}_{R,2})
&& - \dist(B^{n-1}_{P,2},D^{n-2}_{R,2}) && - \dist(A^{n-1}_{P,2},B^{n-2}_{R,2}) \, && = \, 1.032,\\
\text{3)\,\,\,} &  \dist(B^{n-1}_{P,2},D^{n-2}_{R,2}) && + \dist(C^{n-1}_{P,1},D^{n-1}_{P,1})
&& - \dist(B^{n-1}_{P,2},C^{n-1}_{P,1}) && - \dist(D^{n-2}_{R,2},D^{n-1}_{P,1}) \, && = \, 0.168,\\
\text{4)\,\,\,} &  \dist(B^{n-1}_{P,1},A^{n-1}_{P,1}) && + \dist(C^{n-2}_{R,2},D^{n-1}_{P,2})
&& - \dist(B^{n-1}_{P,1},C^{n-2}_{R,2}) && - \dist(A^{n-1}_{P,1},D^{n-1}_{P,2}) \, && = \, 1.14,\\
\text{5)\,\,\,} &  \dist(A^{n-2}_{R,2},C^{n-1}_{P,2}) && + \dist(B^{n-2}_{R,2},A^{n-1}_{P,2})
&& - \dist(A^{n-2}_{R,2},B^{n-2}_{R,2}) && - \dist(C^{n-1}_{P,2},A^{n-1}_{P,2}) \, && = \, 0.06,\\
\text{6)\,\,\,} &  \dist(C^{n-1}_{P,1},B^{n-1}_{P,2}) && + \dist(A^{n-1}_{P,1},D^{n-1}_{P,2})
&& - \dist(C^{n-1}_{P,1},A^{n-1}_{P,1}) && - \dist(B^{n-1}_{P,2},D^{n-1}_{P,2}) \, && = \, 0.4,\\
\text{7)\,\,\,} &  \dist(C^{n-2}_{R,2},B^{n-1}_{P,1}) && + \dist(D^{n-2}_{R,2},D^{n-1}_{P,1})
&& - \dist(C^{n-2}_{R,2},D^{n-2}_{R,2}) && - \dist(B^{n-1}_{P,1},D^{n-1}_{P,1}) \, && = \, 0.012.
\end{alignat*}}
This concludes the proof of Theorem~\ref{theorem:LowerBounds} for the
Manhattan metric as it shows that all 2-changes are improving.

Let us remark that this also implies Theorem~\ref{theorem:LowerBounds} 
for the $L_{\infty}$ metric because distances with respect to the $L_{\infty}$ 
metric coincide with distances with respect to the Manhattan metric if one rotates 
all points by $\pi/4$ around the origin and scales every coordinate 
by $1/\sqrt{2}$.

\subsubsection{Embedding the construction into general $L_p$ metrics}
\label{subsubsec:LB:LpPoints}

It is also possible to embed our Manhattan construction into the $L_p$ metric for~$p\in\NN$ with~$p\ge 3$.
For $j\in[2]$, we choose $A^{n-1}_{R,j}=(0,1)$, 
$B^{n-1}_{R,j}=(0,0)$, $C^{n-1}_{R,j}=(3.5,3.7)$, and 
$D^{n-1}_{R,j}=(7.8,-3.2)$. Moreover, we choose 
$A^{n-1}_{P,1}=(-2.5,-2.4)$, $B^{n-1}_{P,1}=(-4.7,-7.3)$, 
$C^{n-1}_{P,1}=(-8.6,-4.6)$, $D^{n-1}_{P,1}=(3.7,9.8)$, 
$A^{n-1}_{P,2}=(3.2,2)$, $B^{n-1}_{P,2}=(7.2,7.2)$, 
$C^{n-1}_{P,2}=(-6.5,-1.6)$, and $D^{n-1}_{P,2}=(-1.5,-7.1)$. We place 
the points of $G_i^P$ and $G_i^R$ as follows:
\begin{enumerate}
\setlength{\itemsep}{0em}
\item Start with the coordinates specified for the points of 
gadgets $G_{i+1}^P$ and $G_{i+1}^R$.
\item Rotate these points around the origin by $\pi$.
\item Scale each coordinate by a factor of 7.8.
\item Translate the points by the vector $(7.2,5.3)$.
\end{enumerate}
For $j\in[2]$, this yields $A^{n-2}_{R,j}=(7.2,-2.5)$, 
$B^{n-2}_{R,j}=(7.2,5.3)$, $C^{n-2}_{R,j}=(-20.1,-23.56)$, and 
$D^{n-2}_{R,j}=(-53.64,30.26)$.
Additionally, it yields
$A^{n-2}_{P,1}=(26.7,24.02)$, $B^{n-2}_{P,1}=(43.86,62.24)$, 
$C^{n-2}_{P,1}=(74.28,41.18)$, $D^{n-2}_{P,1}=(-21.66,-71.14)$, 
$A^{n-2}_{P,2}=(-17.76,-10.3)$, $B^{n-2}_{P,2}=(-48.96,-50.86)$, 
$C^{n-2}_{P,2}=(57.9,17.78)$, and $D^{n-2}_{P,2}=(18.9,60.68)$.		

It needs to be shown that the distances of these points when measured 
according to the $L_p$ metric for any~$p\in\NN$ with~$p\ge 3$ satisfy all necessary 
inequalities, that is, all 16 inequalities that we have verified in the
previous section for the Manhattan metric. Let us start by showing that
for $j\in[2]$, $A^{n-1}_{R,j}B^{n-1}_{R,j}C^{n-1}_{R,j}D^{n-1}_{R,j}$ is the short 
state. For this, we have to prove the following inequality for every~$p\in\NN$ with~$p\ge 3$:
\begin{alignat}{1}
 & \dist_p(A^{n-1}_{R,j},C^{n-1}_{R,j})\!+\!\dist_p(B^{n-1}_{R,j},D^{n-1}_{R,j}) \!-\!
(\dist_p(A^{n-1}_{R,j},B^{n-1}_{R,j}) \!+\!\dist_p(C^{n-1}_{R,j},D^{n-1}_{R,j})) > 0 \notag\\
 & \iff \,\,\,  \sqrt[p]{3.5^p+2.7^p} + \sqrt[p]{7.8^p+3.2^p} - \sqrt[p]{0^p+1^p} - \sqrt[p]{4.3^p+6.9^p} \, > 0.
 \label{eqn:LB:Lp1}
\end{alignat}
For $p=\infty$, the inequality is satisfied as the left side equals $3.4$ when
distances are measured according to the $L_{\infty}$ metric. In order to
show that the inequality is also satisfied for every~$p\in\NN$ with~$p\ge 3$, we analyze by
how much the distances $\dist_p$ deviate from the distances $\dist_{\infty}$.
For~$p\in\NN$ with~$p\ge 3$, we obtain
\begin{equation}\textstyle
\begin{split}
  & \sqrt[p]{4.3^p+6.9^p} - 6.9
  = 6.9\cdot\left(\sqrt[p]{1+\left(\frac{4.3}{6.9}\right)^p} - 1\right)\\   
  \le &\, 6.9\cdot\left(\sqrt[3]{1+\left(\frac{4.3}{6.9}\right)^3} - 1\right)
  < 0.52.\label{eqn:LB:Lp2}
  \end{split}
\end{equation}
Hence,
\begin{alignat*}{1}
   & \sqrt[p]{3.5^p+2.7^p} + \sqrt[p]{7.8^p+3.2^p} - \sqrt[p]{0^p+1^p} - \sqrt[p]{4.3^p+6.9^p} \\
   \ge &\, 3.5+7.8-1-6.9 - 0.52 > 0,
\end{alignat*}
which proves that $A^{n-1}_{R,j}B^{n-1}_{R,j}C^{n-1}_{R,j}D^{n-1}_{R,j}$ is the short 
state for every~$p\in\NN$ with~$p\ge 3$.

Next we argue that also the 2-change in which $G_{n-1}^P$ changes its state from
$(L,L)$ to $(S,L)$ is improving. For this, the following inequality needs to be
verified for every~$p\in\NN$ with~$p\ge 3$: 
\begin{alignat*}{1}
&  \dist(A^{n-1}_{P,1},C^{n-1}_{P,1})+\dist(B^{n-1}_{P,1},D^{n-1}_{P,1})-
(\dist(A^{n-1}_{P,1},B^{n-1}_{P,1})-\dist(C^{n-1}_{P,1},D^{n-1}_{P,1}))   > 0\\
& \iff \,\,\,  \sqrt[p]{6.1^p+2.2^p} + \sqrt[p]{8.4^p+17.1^p} - \sqrt[p]{2.2^p+4.9^p} - \sqrt[p]{12.3^p+14.4^p}  > 0.
\end{alignat*}
As before, we obtain for~$p\in\NN$ with~$p\ge3$
\[\textstyle
  \sqrt[p]{2.2^p+4.9^p} - 4.9
  = 4.9\cdot\left(\sqrt[p]{1+\left(\frac{2.2}{4.9}\right)^p} \!-\! 1\right)   
  \le 4.9\cdot\left(\sqrt[3]{1+\left(\frac{2.2}{4.9}\right)^3} \!-\! 1\right)
  < 0.15
\]
and
\begin{alignat*}{1}\textstyle
  & \sqrt[p]{12.3^p+14.4^p} - 14.4
  = 14.4\cdot\left(\sqrt[p]{1+\left(\frac{12.3}{14.4}\right)^p} - 1\right)\\   
  \le &\,14.4\cdot\left(\sqrt[3]{1+\left(\frac{12.3}{14.4}\right)^3} - 1\right)
  < 2.53.
\end{alignat*}
This implies for~$p\in\NN$ with~$p\ge 3$
\begin{align*}
   & \quad \sqrt[p]{6.1^p+2.2^p} + \sqrt[p]{8.4^p+17.1^p} - \sqrt[p]{2.2^p+4.9^p} - \sqrt[p]{12.3^p+14.4^p}\\
   & \ge 6.1+17.1-4.9-0.15-14.4-2.53 > 0,
\end{align*}
which proves that the 2-change in which $G_{n-1}^P$ changes its state from
$(L,L)$ to $(S,L)$ is improving for every~$p\in\NN$ with~$p\ge 3$.

Next we show that 
the improvements made by the 2-changes in the sequence in which 
$G_{n-1}^P$ changes its state from $(S,L)$ to $(S,S)$ while resetting 
$G_{n-1}^R$ are positive. For this we need to verify the following inequalities
for every~$p\in\NN$ with~$p\ge 3$ (observe that these are exactly the same inequalities that we
have verified in Section~\ref{subsubsec:LB:L1Points} for the Manhattan metric):
{\small
\begin{alignat*}{5}
  \text{1)\,\,\,} &  \dist_p(A^{n-1}_{P,2},C^{n-1}_{P,2}) && + \dist_p(C^{n-1}_{R,2},D^{n-1}_{R,2})
                  && - \dist_p(A^{n-1}_{P,2},C^{n-1}_{R,2}) && - \dist_p(C^{n-1}_{P,2},D^{n-1}_{R,2}) \, && > \, 0\\
  \iff \,\,\, & \sqrt[p]{9.7^p+3.6^p} &&+ \sqrt[p]{4.3^p+6.9^p} &&- \sqrt[p]{0.3^p+1.7^p} &&- \sqrt[p]{14.3^p+1.6^p} &&\, > 0,\\                  
  \text{2)\,\,\,} &  \dist_p(B^{n-1}_{R,2},A^{n-1}_{R,2}) && + \dist_p(D^{n-1}_{P,2},B^{n-1}_{P,2})
                  && - \dist_p(B^{n-1}_{R,2},D^{n-1}_{P,2}) && - \dist_p(A^{n-1}_{R,2},B^{n-1}_{P,2}) && > \, 0\\
  \iff \,\,\, & \sqrt[p]{0.0^p+1.0^p} &&+ \sqrt[p]{8.7^p+14.3^p} &&- \sqrt[p]{1.5^p+7.1^p} &&- \sqrt[p]{7.2^p+6.2^p} &&\, > 0,\\
  \text{3)\,\,\,} &  \dist_p(B^{n-1}_{R,2},D^{n-1}_{P,2}) && + \dist_p(C^{n-1}_{R,1},D^{n-1}_{R,1})
                  && - \dist_p(B^{n-1}_{R,2},C^{n-1}_{R,1}) && - \dist_p(D^{n-1}_{P,2},D^{n-1}_{R,1}) && > \, 0\\
  \iff \,\,\, & \sqrt[p]{1.5^p+7.1^p} &&+ \sqrt[p]{4.3^p+6.9^p} &&- \sqrt[p]{3.5^p+3.7^p} &&- \sqrt[p]{9.3^p+3.9^p} &&\, > 0,\\
  \text{4)\,\,\,} &  \dist_p(B^{n-1}_{R,1},A^{n-1}_{R,1}) && + \dist_p(C^{n-1}_{P,2},D^{n-1}_{R,2})
                  && - \dist_p(B^{n-1}_{R,1},C^{n-1}_{P,2}) && - \dist_p(A^{n-1}_{R,1},D^{n-1}_{R,2}) && > \, 0\\
  \iff \,\,\, & \sqrt[p]{0.0^p+1.0^p} &&+ \sqrt[p]{14.3^p+1.6^p} &&- \sqrt[p]{6.5^p+1.6^p} &&- \sqrt[p]{7.8^p+4.2^p} &&\, > 0,\\
  \text{5)\,\,\,} &  \dist_p(A^{n-1}_{P,2},C^{n-1}_{R,2}) && + \dist_p(B^{n-1}_{P,2},A^{n-1}_{R,2})
                  && - \dist_p(A^{n-1}_{P,2},B^{n-1}_{P,2}) && - \dist_p(C^{n-1}_{R,2},A^{n-1}_{R,2}) && > \, 0\\                                    
  \iff \,\,\, & \sqrt[p]{0.3^p+1.7^p} &&+ \sqrt[p]{7.2^p+6.2^p} &&- \sqrt[p]{4.0^p+5.2^p} &&- \sqrt[p]{3.5^p+2.7^p} &&\, > 0,\\
  \text{6)\,\,\,} &  \dist_p(C^{n-1}_{R,1},B^{n-1}_{R,2}) && + \dist_p(A^{n-1}_{R,1},D^{n-1}_{R,2})
                  && - \dist_p(C^{n-1}_{R,1},A^{n-1}_{R,1}) && - \dist_p(B^{n-1}_{R,2},D^{n-1}_{R,2}) && > \, 0\\
  \iff \,\,\, & \sqrt[p]{3.5^p+3.7^p} &&+ \sqrt[p]{7.8^p+4.2^p} &&- \sqrt[p]{3.5^p+2.7^p} &&- \sqrt[p]{7.8^p+3.2^p} &&\, > 0,\\
  \text{7)\,\,\,} &  \dist_p(C^{n-1}_{P,2},B^{n-1}_{R,1}) && + \dist_p(D^{n-1}_{P,2},D^{n-1}_{R,1})
                  && - \dist_p(C^{n-1}_{P,2},D^{n-1}_{P,2}) && - \dist_p(B^{n-1}_{R,1},D^{n-1}_{R,1}) && > \, 0\\
  \iff \,\,\, & \sqrt[p]{6.5^p+1.6^p} &&+ \sqrt[p]{9.3^p+3.9^p} &&- \sqrt[p]{5.0^p+5.5^p} &&- \sqrt[p]{7.8^p+3.2^p} &&\, > 0.\\
\end{alignat*}}
These inequalities can be checked in the same way as Inequality~\eqref{eqn:LB:Lp1}.
Details can be found in Appendix~\ref{app:Inequalities}.

It remains to be shown that the sequences in which gadget 
$G^R_{n-2}$ resets gadget $G^P_{n-1}$ from $(S,S)$ to $(L,L)$, are 
improving. As the coordinates of the points in the two blocks of gadget 
$G^R_{n-2}$ are the same, the inequalities for both sequences are also 
identical. We need to verify the following inequalities:
{\small
\begin{alignat*}{5}
\text{1)\,\,\,} &  \dist_p(A^{n-2}_{R,2},C^{n-2}_{R,2}) && + \dist_p(C^{n-1}_{P,2},D^{n-1}_{P,2})
&& - \dist_p(A^{n-2}_{R,2},C^{n-1}_{P,2}) && - \dist_p(C^{n-2}_{R,2},D^{n-1}_{P,2}) \, && > \, 0\\
  \iff \,\,\, & \sqrt[p]{27.3^p+21.06^p} &&+ \sqrt[p]{5.0^p+5.5^p} &&- \sqrt[p]{13.7^p+0.9^p} &&- \sqrt[p]{18.6^p+16.46^p} &&\! > 0,\\                  
  \text{2)\,\,\,} &  \dist_p(B^{n-1}_{P,2},A^{n-1}_{P,2}) && + \dist_p(D^{n-2}_{R,2},B^{n-2}_{R,2})
&& - \dist_p(B^{n-1}_{P,2},D^{n-2}_{R,2}) && - \dist_p(A^{n-1}_{P,2},B^{n-2}_{R,2}) \, && > \, 0\\
  \iff \,\,\, & \sqrt[p]{4.0^p+5.2^p} &&+ \sqrt[p]{60.84^p+24.96^p} &&- \sqrt[p]{60.84^p+23.06^p} &&- \sqrt[p]{4.0^p+3.3^p} &&\! > 0,\\
  \text{3)\,\,\,} &  \dist_p(B^{n-1}_{P,2},D^{n-2}_{R,2}) && + \dist_p(C^{n-1}_{P,1},D^{n-1}_{P,1})
&& - \dist_p(B^{n-1}_{P,2},C^{n-1}_{P,1}) && - \dist_p(D^{n-2}_{R,2},D^{n-1}_{P,1}) \, && > \, 0\\
  \iff \,\,\, & \sqrt[p]{60.84^p+23.06^p} &&+ \sqrt[p]{12.3^p+14.4^p} &&- \sqrt[p]{15.8^p+11.8^p} &&- \sqrt[p]{57.34^p+20.46^p} &&\! > 0,\\
  \text{4)\,\,\,} &  \dist_p(B^{n-1}_{P,1},A^{n-1}_{P,1}) && + \dist_p(C^{n-2}_{R,2},D^{n-1}_{P,2})
&& - \dist_p(B^{n-1}_{P,1},C^{n-2}_{R,2}) && - \dist_p(A^{n-1}_{P,1},D^{n-1}_{P,2}) \, && > \, 0\\
  \iff \,\,\, & \sqrt[p]{2.2^p+4.9^p} &&+ \sqrt[p]{18.6^p+16.46^p} &&- \sqrt[p]{15.4^p+16.26^p} &&- \sqrt[p]{1.0^p+4.7^p} &&\! > 0,\\
\text{5)\,\,\,} &  \dist_p(A^{n-2}_{R,2},C^{n-1}_{P,2}) && + \dist_p(B^{n-2}_{R,2},A^{n-1}_{P,2})
&& - \dist_p(A^{n-2}_{R,2},B^{n-2}_{R,2}) && - \dist_p(C^{n-1}_{P,2},A^{n-1}_{P,2}) \, && > \, 0\\
  \iff \,\,\, & \sqrt[p]{13.7^p+0.9^p} &&+ \sqrt[p]{4.0^p+3.3^p} &&- \sqrt[p]{0.0^p+7.8^p} &&- \sqrt[p]{9.7^p+3.6^p} &&\! > 0,\\
  \text{6)\,\,\,} &  \dist_p(C^{n-1}_{P,1},B^{n-1}_{P,2}) && + \dist_p(A^{n-1}_{P,1},D^{n-1}_{P,2})
&& - \dist_p(C^{n-1}_{P,1},A^{n-1}_{P,1}) && - \dist_p(B^{n-1}_{P,2},D^{n-1}_{P,2}) \, && > \, 0\\
  \iff \,\,\, & \sqrt[p]{15.8^p+11.8^p} &&+ \sqrt[p]{1.0^p+4.7^p} &&- \sqrt[p]{6.1^p+2.2^p} &&- \sqrt[p]{8.7^p+14.3^p} &&\! > 0,\\
  \text{7)\,\,\,} &  \dist_p(C^{n-2}_{R,2},B^{n-1}_{P,1}) && + \dist_p(D^{n-2}_{R,2},D^{n-1}_{P,1})
&& - \dist_p(C^{n-2}_{R,2},D^{n-2}_{R,2}) && - \dist_p(B^{n-1}_{P,1},D^{n-1}_{P,1}) \, && > \, 0\\
\iff \,\,\, & \sqrt[p]{15.4^p+16.26^p} &&+ \sqrt[p]{57.34^p+20.46^p} &&- \sqrt[p]{33.54^p+53.82^p} &&- \sqrt[p]{8.4^p+17.1^p} &&\! > 0.
  \end{alignat*}}
These inequalities can be checked in the same way as
Inequality~\eqref{eqn:LB:Lp1} was checked; see the details in Appendix~\ref{app:Inequalities}.

\section{Expected Number of 2-Changes}
\label{sec:runningTime}

We analyze the expected number of 2-changes on random $d$-dimensional 
Manhattan and Euclidean instances, for an arbitrary constant dimension 
$d\ge2$. One possible approach for this is to analyze the 
improvement made by the smallest improving 2-change: If the smallest 
improvement is not too small, then the number of improvements cannot be 
large. This approach yields polynomial bounds, but in our analysis, we 
consider not only a single step but certain pairs of steps. We show 
that the smallest improvement made by any such pair is typically much 
larger than the improvement made by a single step, which yields 
better bounds. Our approach is not restricted to pairs of steps. One 
could also consider sequences of steps of length $k$ for any small 
enough $k$. In fact, for general $\phi$-perturbed graphs with $m$ 
edges, we consider sequences of length $\sqrt{\log{m}}$ in~\cite{EnglertRV07}.
The reason why 
we can analyze longer sequences for general graphs is that these inputs 
possess more randomness than $\phi$-perturbed Manhattan and Euclidean 
instances because every edge length is a random variable that is 
independent of the other edge lengths. Hence, the analysis for general 
$\phi$-perturbed graphs demonstrates the limits of our approach under 
optimal conditions. For Manhattan and Euclidean instances, the gain of 
considering longer sequences is small due to the dependencies between 
the edge lengths.

\subsection{Manhattan Instances}
\label{subsec:UB:L1}

In this section, we analyze the expected number of 2-changes on 
$\phi$-perturbed Manhattan instances. First we prove a weaker bound 
than the one in Theorem~\ref{theorem:runningTime1} in a slightly different model.
In this model the position of a vertex~$v_i$ is not chosen according to a density
function~$f_i\colon[0,1]^d \to[0,\phi]$, but instead each of its~$d$ coordinates
is chosen independently. To be more precise, for every~$j\in[d]$, there is a density
function~$f_i^j\colon[0,1] \to[0,\phi]$ according to which the $j$th coordinate of~$v_i$ is chosen.

The proof of this 
weaker bound illustrates our approach and reveals the problems one has 
to tackle in order to improve the upper bounds. It is solely based on 
an analysis of the smallest improvement made by any of the possible 
2-Opt steps. If with high probability every 2-Opt step decreases the 
tour length by an inverse polynomial amount, then with high probability 
only polynomially many 2-Opt steps are possible before a local optimum 
is reached. In fact, the probability that there exists a 2-Opt step that decreases
the tour length by less than an inverse polynomial amount is so small that
(as we will see) even the expected number of possible 2-Opt steps can be bounded polynomially.

\begin{theorem}
\label{theorem:L1weak}
Starting with an arbitrary tour, the expected number of steps performed 
by 2-Opt on $\phi$-perturbed Manhattan instances with $n$ vertices is 
$O(n^6\cdot\log{n}\cdot\phi)$ if the coordinates of every vertex are drawn independently.
\end{theorem}
\begin{proof}
We will see below that, 
in order to prove the desired bound on the expected convergence time, 
we only need two simple observations. First, the initial tour can have 
length at most $dn$ as the number of edges is $n$ and every edge has 
length at most $d$. And second, every 2-Opt step decreases the length 
of the tour by an inverse polynomial amount with high probability. The 
latter can be shown by a union bound over all possible 2-Opt steps. 
Consider a fixed 2-Opt step $S$, let $e_1$ and $e_2$ denote the edges 
removed from the tour in step $S$, and let $e_3$ and $e_4$ denote the 
edges added to the tour. Then the improvement $\Delta(S)$ of step $S$ 
can be written as
\begin{equation}
  \label{eqn:improvementStep} 
  \Delta(S) = \dist(e_1)+\dist(e_2)-\dist(e_3)-\dist(e_4). 
\end{equation}
Without loss of generality let $e_1=(v_1,v_2)$ be the edge between the 
vertices $v_1$ and $v_2$, and let $e_2=(v_3,v_4)$, $e_3=(v_1,v_3)$, and 
$e_4=(v_2,v_4)$. Furthermore, for $i\in\{1,\ldots4\}$, let 
$x^i\in\RR^d$ denote the coordinates of vertex $v_i$. Then the 
improvement $\Delta(S)$ of step $S$ can be written as
\[
\Delta(S) = \sum_{i=1}^d\left(|x^1_i-x^2_i|+|x^3_i-x^4_i|-|x^1_i-x^3_i|-|x^2_i-x^4_i|\right).
\]
Depending on the order of the coordinates, $\Delta(S)$ can be written 
as some linear combination of the coordinates. If, \eg, for all $i\in[d]$, 
$x^1_i\ge x^2_i\ge x^3_i\ge x^4_i$, then the improvement $\Delta(S)$ 
can be written as $\sum_{i=1}^d(-2x^2_i+2x^3_i)$. There are $(4!)^d$ 
such orders and each one gives rise to a linear combination of the 
$x_i^j$'s with integer coefficients.

For each of these linear 
combinations, the probability that it takes a value in the interval 
$(0,\varepsilon]$ is bounded from above by $\varepsilon\phi$.
To see this, we distinguish between two cases:
If all coefficients in the linear combination are zero then the probability that
the linear combination takes a value in the interval~$(0,\varepsilon]$
is zero. If at least one coefficient is nonzero then we can apply the principle
of deferred decisions (see, \eg,~\cite{MotwaniR95}). Let~$x_i^j$ be a variable
that has a nonzero coefficient~$\alpha$ and assume that all random variables except for~$x_i^j$
are already drawn. Then, in order for the linear combination to take a value in the
interval~$(0,\varepsilon]$, the random variable~$x_i^j$ has to take a value in a fixed
interval of length~$\varepsilon/|\alpha|$. As the density of~$x_i^j$ is bounded from above by~$\phi$
and~$\alpha$ is a nonzero integer, the probability of this event is at most~$\varepsilon\phi$.
 
Since $\Delta(S)$ 
can only take a value in the interval $(0,\varepsilon]$ if one of the 
linear combinations takes a value in this interval, the probability of 
the event $\Delta(S)\in(0,\varepsilon]$ can be upper bounded by 
$(4!)^d\varepsilon\phi$.

Let $\Delta_{\min}$ denote the improvement of the smallest improving 
2-Opt step $S$, \ie, $\Delta_{\min} = 
\min\{\Delta(S)\mid\Delta(S)>0\}$. We can estimate $\Delta_{\min}$ by a 
union bound, yielding
\[
       \Pr{\Delta_{\min}\le\varepsilon}  
   \le (4!)^d\varepsilon n^4\phi
\]
as there are at most $n^4$ different 2-Opt steps. Let $T$ denote the 
random variable describing the number of 2-Opt steps before a local 
optimum is reached. Observe that $T$ can only exceed a given number 
$t$ if the smallest improvement $\Delta_{\min}$ is less than $dn/t$, 
and hence
\[
   \Pr{T\ge t} \le \Pr{\Delta_{\min}\le\frac{dn}{t}} \le \frac{d(4!)^dn^5\phi}{t}. 
\]
Since there are at most $n!$ different TSP tours and none of these 
tours can appear twice during the local search, $T$ is always bounded 
by $n!$. Altogether, we can bound the expected value of $T$ by
\[
  \Ex{T} = \sum_{t=1}^{n!}\Pr{T\ge t}
  \le \sum_{t=1}^{n!} \frac{d(4!)^dn^5\phi}{t}.
\]
Since we assumed the dimension $d$ to be a constant, bounding the 
$n$-th harmonic number by $\ln(n)+1$ and using $\ln(n!)=O(n\log{n})$ 
yields
\[
  \Ex{T} \le d(4!)^dn^5\phi (\ln(n!)+1) = O(n^6\cdot\log{n}\cdot\phi).\tag*{\qed}
\]
\end{proof}

The bound in Theorem~\ref{theorem:L1weak} is only based on the smallest 
improvement $\Delta_{\min}$ made by any of the 2-Opt steps. 
Intuitively, this is too pessimistic since most of the steps performed 
by 2-Opt yield a larger improvement than $\Delta_{\min}$. In 
particular, two consecutive steps yield an improvement of at least 
$\Delta_{\min}$ plus the improvement $\Delta_{\min}'$ of the second 
smallest step. This observation alone, however, does not suffice to 
improve the bound substantially. Instead, we show in Lemma~\ref{lemma:numberLinkedPairs} 
that we can regroup the 2-changes to 
pairs such that each pair of 2-changes is \emph{linked} by an edge, 
\ie, one edge added to the tour in the first 2-change is removed from 
the tour in the second 2-change. Then we analyze the smallest 
improvement made by any pair of linked 2-Opt steps. Obviously, this 
improvement is at least $\Delta_{\min}+\Delta_{\min}'$ but one can hope 
that it is much larger because it is unlikely that the 2-change that 
yields the smallest improvement and the 2-change that yields the second 
smallest improvement form a pair of linked steps. We show that this is 
indeed the case and use this result to prove the bound on the expected 
length of the longest path in the state graph of 2-Opt on 
$\phi$-perturbed Manhattan instances claimed in 
Theorem~\ref{theorem:runningTime1}.

\subsubsection{Construction of pairs of linked 2-changes}
\label{subsubsec:LinkedPairs}

Consider an arbitrary sequence of length $t$ of consecutive 2-changes. 
The following lemma guarantees that the number of disjoint linked pairs 
of 2-changes in every such sequence increases linearly with the length 
$t$.
\begin{lemma}
\label{lemma:numberLinkedPairs}
In every sequence of $t$ consecutive 2-changes, the number of disjoint 
pairs of 2-changes that are linked by an edge, \ie, pairs such that 
there exists an edge added to the tour in the first 2-change of the 
pair and removed from the tour in the second 2-change of the pair, is 
at least $(2t-n)/7$.
\end{lemma}
\begin{proof}
Let $S_1,\ldots,S_t$ denote an arbitrary sequence of consecutive
2-changes. The sequence is processed step by step and a list ${\cal L}$
of linked pairs of 2-changes is created. The pairs in~${\cal L}$ are not
necessarily disjoint. Hence, after the list has been created, pairs
have to be removed from the list until there are no non-disjoint pairs
left. Assume that the 2-changes $S_1,\ldots,S_{i-1}$ have already been
processed and that now 2-change $S_i$ has to be processed. Assume
further that in step $S_i$ the edges $e_1$ and $e_2$ are exchanged with
the edges $e_3$ and $e_4$. Let $j$ denote the smallest index with $j>i$
such that edge $e_3$ is removed from the tour in step $S_j$ if such a
step exists. In this case, the pair $(S_i,S_j)$ is added to the
list ${\cal L}$. Analogously, let $j'$ denote the smallest
index with $j'>i$ such that edge $e_4$ is removed from the tour in step
$S_{j'}$ if such a step exists. In this case, also the pair
$(S_i,S_{j'})$ is added to the list ${\cal L}$.

After the sequence has been processed completely, each pair in ${\cal
L}$ is linked by an edge but we still have to identify a subset ${\cal
L}'$ of ${\cal L}$ consisting only of pairwise disjoint pairs. This
subset is constructed in a greedy fashion. We process the list ${\cal
L}$ step by step, starting with an empty list ${\cal L}'$. For each
pair in ${\cal L}$, we check whether it is disjoint from all pairs that
have already been inserted into ${\cal L}'$ or not. In the former case,
the current pair is inserted into ${\cal L'}$. This way, we obtain a
list ${\cal L}'$ of disjoint pairs such that each pair is linked by an
edge. The number of pairs in ${\cal L}$ is at least $2t-n$ because
each of the $t$ steps gives rise to $2$ pairs, unless an edge
is added to the tour that is never removed again. The tour~$C$ obtained after the 2-changes~$S_1,\ldots,S_t$ contains exactly~$n$ edges. For every edge~$e\in C$, only the last step in which~$e$ enters the tour (if such a step exists) does not create a pair of linked 2-changes involving~$e$.

Each 2-change occurs in at most~$4$ different pairs in ${\cal L}$. In order to see this, consider a 2-change~$S$ in which the edges~$e_1$ and~$e_2$ are exchanged with the edges~$e_3$ and~$e_4$. Then~${\cal L}$ contains the following pairs involving~$S$ (if they exist): $(S,S')$ where~$S'$ is either the first step after~$S$ in which~$e_3$ gets removed or the first step after~$S$ in which~$e_4$ gets removed, and $(S',S)$ where~$S'$ is either the last step before~$S$ in which~$e_1$ enters the tour or the last step before~$S$ in which~$e_2$ enters the tour. With similar reasoning, one can argue that each pair in
${\cal L}$ is non-disjoint from at most $6$ other pairs in ${\cal L}$.
This implies that ${\cal L}$ contains at most $7$ times as many pairs
as ${\cal L'}$, which concludes the proof.
\end{proof}

Consider a fixed pair of 2-changes linked by an edge. Without loss of 
generality assume that in the first step the edges $\{v_1,v_2\}$ and 
$\{v_3,v_4\}$ are exchanged with the edges $\{v_1,v_3\}$ and 
$\{v_2,v_4\}$, for distinct vertices $v_1,\ldots,v_4$. Also without 
loss of generality assume that in the second step the edges 
$\{v_1,v_3\}$ and $\{v_5,v_6\}$ are exchanged with the edges 
$\{v_1,v_5\}$ and $\{v_3,v_6\}$. However, note that the vertices $v_5$ 
and $v_6$ are not necessarily distinct from the vertices $v_2$ and 
$v_4$. We distinguish between three different types of pairs.

\begin{figure}[H]
\begin{center}
\includegraphics[scale=0.9]{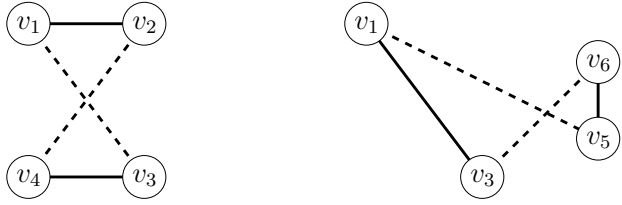}
\caption{A pair of type~0.}         
\label{fig:Type1}
\end{center}
\end{figure}

\begin{figure}[H]
\begin{center}
\includegraphics[scale=0.9]{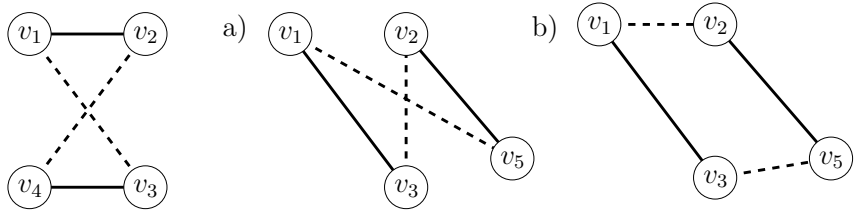}
\caption{Pairs of type~1.}
\label{fig:Type2}
\end{center}
\end{figure}

\begin{itemize}
  \setlength{\itemsep}{0em}
  \item pairs of type 0: $\left|\{v_2,v_4\}\cap\{v_5,v_6\}\right|=0$. This case is
  illustrated in Figure~\ref{fig:Type1}.
  \item pairs of type 1: $\left|\{v_2,v_4\}\cap\{v_5,v_6\}\right|=1$. We can assume 
  \wlg that $v_2\in\{v_5,v_6\}$. We have to distinguish 
  between two subcases: a) The edges $\{v_1,v_5\}$ and $\{v_2,v_3\}$ 
  are added to the tour in the second step.\ b) The edges $\{v_1,v_2\}$ 
  and $\{v_3,v_5\}$ are added to the tour in the second step. These 
  cases are illustrated in Figure~\ref{fig:Type2}.
  \item pairs of type 2: $\left|\{v_2,v_4\}\cap\{v_5,v_6\}\right|=2$. The case 
  $v_2=v_5$ and $v_4=v_6$ cannot appear as it would imply that in the first step
  the edges $\{v_1,v_2\}$ and $\{v_3,v_4\}$ are exchanged with the edges $\{v_1,v_3\}$ and 
  $\{v_2,v_4\}$, and that in the second step the edges $\{v_1,v_3\}$ and 
  $\{v_2,v_4\}$ are again exchanged with the edges $\{v_1,v_2\}$ and $\{v_3,v_4\}$.
  Hence, one of these 2-changes cannot be improving,
  and for pairs of this type we must have $v_2=v_6$ and $v_4=v_5$.
\end{itemize}

When distances are measured according to the Euclidean metric, pairs of 
type 2 result in vast dependencies and hence the probability that there 
exists a pair of this type in which both steps are improvements by at 
most $\varepsilon$ with respect to the Euclidean metric cannot easily be bounded. In order to reduce the number of cases we have to 
consider and in order to prepare for the analysis of $\phi$-perturbed 
Euclidean instances, we exclude pairs of type 2 from our probabilistic 
analysis by leaving out all pairs of type 2 when constructing the list 
${\cal L}$ in the proof of Lemma~\ref{lemma:numberLinkedPairs}. The following lemma shows that 
there are always enough pairs of type~0 or 1. 

\begin{lemma}
\label{lemma:numberLinkedPairs2}
In every sequence of $t$ consecutive 2-changes the number of disjoint 
pairs of 2-changes of type~0 or 1 is at least $t/7-3n/28$.
\end{lemma}
\begin{proof}
We follow the construction in the proof of Lemma~\ref{lemma:numberLinkedPairs} with the only difference that we do not add pairs of type~2 to the list ${\cal L}$. Assume that in step $S_i$ the edges $\{v_1,v_2\}$ and $\{v_3,v_4\}$ 
are replaced by the edges $\{v_1,v_3\}$ and $\{v_2,v_4\}$, and that in 
step $S_j$ these edges are replaced by the edges $\{v_1,v_4\}$ and 
$\{v_2,v_3\}$.
Now consider the next step $S_l$ with $l>j$ in which the 
edge $\{v_1,v_4\}$ is removed from the tour, if such a step exists, and 
the next step $S_{l'}$ with $l'>j$ in which the edge $\{v_2,v_3\}$ is 
removed from the tour if such a step exists. Observe that neither 
$(S_j,S_l)$ nor $(S_j,S_{l'})$ can be a pair of type~2 because 
otherwise the improvement of one of the steps $S_i$, $S_j$, and $S_l$, 
or of one of the steps $S_i$, $S_j$, and $S_{l'}$, respectively, must be negative. In particular, we must have $l\neq l'$. 

Hence, for every pair of type~2 in ${\cal L}$ in which either $l$ or $l'$ is defined, we can identify a pair of type~0 or~1 in ${\cal L}$. Let $x$ denote the number of pairs of type~2 that are encountered in the construction of ${\cal L}$ in the proof of Lemma~\ref{lemma:numberLinkedPairs}. There can be at most $n/2$ pairs of type~2 for which neither $l$ or $l'$ is defined. Hence, the total number~$y$ of type~0 or~1 pairs in ${\cal L}$ must be at least $x-n/2$. This implies $x \le y+n/2$.

Let $z$ denote the number of pairs that are added to the list ${\cal L}$ in the proof of Lemma~\ref{lemma:numberLinkedPairs}. By definition, $z=x+y$. Furthermore, we argued that $z\ge 2t-n$. Altogether this implies $2t-n\le x+y\le 2y+n/2$, which in turn implies $y\ge t-3n/4$. Hence, if we do not add pairs of type~2 to the list ${\cal L}$, we still have at least $t-3n/4$ pairs in ${\cal L}$. By the same arguments as in the proof of Lemma~\ref{lemma:numberLinkedPairs}, the list ${\cal L'}$ contains at least $(t-3n/4)/7=t/7-3n/28$ pairs, which are all of type~0 or~1 and pairwise disjoint.
\end{proof}

\subsubsection{Analysis of pairs of linked 2-changes}

The following lemma gives a bound on the probability that there exists 
a pair of type~0 or~1 in which both steps are small improvements.
\begin{lemma}
\label{lemma:L1Improvements}
In a $\phi$-perturbed Manhattan instance with $n$ vertices, the 
probability that there exists a pair of type~0 or type~1 in which both 
2-changes are improvements by at most $\varepsilon$ is
$O(n^6\cdot\varepsilon^{2}\cdot\phi^2)$.
\end{lemma}
\begin{proof}
First, we consider pairs of type~0. We assume that in the first step 
the edges $\{v_1,v_2\}$ and $\{v_3,v_4\}$ are replaced by the edges 
$\{v_1,v_3\}$ and $\{v_2,v_4\}$ and that in the second step the edges 
$\{v_1,v_3\}$ and $\{v_5,v_6\}$ are replaced by the edges $\{v_1,v_5\}$ 
and $\{v_3,v_6\}$. For $j\in[6]$, let $x^j_i\in\RR^d$, $i=1,2,\ldots,d$, denote the 
$d$ coordinates of vertex $v_j$. Furthermore, let $\Delta_1$ denote the 
(possibly negative) improvement of the first step and let $\Delta_2$ 
denote the (possibly negative) improvement of the second step. 
The random variables $\Delta_1$ and $\Delta_2$ can be written as
\begin{align*}
   \Delta_1 & = \sum_{i=1}^d
   (|x^1_i-x^2_i|+|x^3_i-x^4_i|-|x^1_i-x^3_i|-|x^2_i-x^4_i|)
\intertext{and}
   \Delta_2 & = \sum_{i=1}^d
   (|x^1_i-x^3_i|+|x^5_i-x^6_i|-|x^1_i-x^5_i|-|x^3_i-x^6_i|).
\end{align*}

For any fixed order of the coordinates, $\Delta_1$ and $\Delta_2$ can 
be expressed as linear combinations of the coordinates with integer 
coefficients. For $i\in[d]$, let $\sigma_i$ denote an order of the 
coordinates $x_i^1,\ldots,x_i^6$, let 
$\sigma=(\sigma_1,\ldots,\sigma_d)$, and let $\Delta_1^{\sigma}$ and 
$\Delta_2^{\sigma}$ denote the corresponding linear combinations. We 
denote by ${\cal A}$ the event that both $\Delta_1$ and $\Delta_2$ take 
values in the interval $(0,\varepsilon]$, and we denote by ${\cal 
A}^{\sigma}$ the event that both linear combinations 
$\Delta_1^{\sigma}$ and $\Delta_2^{\sigma}$ take values in the interval 
$(0,\varepsilon]$. Obviously ${\cal A}$ can only occur if for at least 
one $\sigma$, the event ${\cal A}^{\sigma}$ occurs. Hence, we obtain 
\[
   \Pr{{\cal A}} \le \sum_{\sigma}\Pr{{\cal A}^{\sigma}}.
\]
Since there are $(6!)^d$ different orders $\sigma$, which is constant 
for constant dimension $d$, it suffices to show that for every tuple of 
orders $\sigma$, the probability of the event ${\cal A}^{\sigma}$ is 
bounded from above by $O(\varepsilon^2\phi^2)$. Then a union bound over 
all possible pairs of linked 2-changes of type~0 (there are fewer than~$n^6$ of them) and all possible orders~$\sigma$ 
(there is a constant number of them) yields the lemma for pairs of type~0.

We divide the set of possible pairs of linear combinations 
$(\Delta_1^{\sigma},\Delta_2^{\sigma})$ into three classes. We say that 
a pair of linear combinations belongs to class A if at least one of the 
linear combinations equals $0$, we say that it belongs to class B if 
$\Delta_1^{\sigma}=-\Delta_2^{\sigma}$, and we say that it belongs to 
class C if $\Delta_1^{\sigma}$ and $\Delta_2^{\sigma}$ are linearly 
independent. For tuples of orders $\sigma$ that yield pairs from class A,
the event ${\cal A}^{\sigma}$ cannot occur because the 
value of at least one linear combination is $0$. For tuples $\sigma$ that yield pairs from
class $B$, the event cannot occur either because either $\Delta_1^{\sigma}$
or $\Delta_2^{\sigma}=-\Delta_1^{\sigma}$ is at most $0$.
For tuples $\sigma$ that yield pairs from class C, we can apply 
Lemma~\ref{lemma:ProbLinComb} from Appendix~\ref{appendix:probTheory}, 
which shows that the probability of the event ${\cal A}^{\sigma}$ is 
bounded from above by $(\varepsilon\phi)^2$. Hence, we only need to 
show that every pair $(\Delta_1^{\sigma},\Delta_2^{\sigma})$ of linear 
combinations belongs either to class A, B, or C.

Consider a fixed tuple $\sigma=(\sigma_1,\ldots,\sigma_d)$ of orders. 
We split $\Delta_1^{\sigma}$ and $\Delta_2^{\sigma}$ into $d$ parts 
that correspond to the $d$ dimensions. To be precise, for $j\in[2]$, we 
write $\Delta_j^{\sigma}=\sum_{i\in[d]}X^{\sigma_i,i}_j$, where 
$X^{\sigma_i,i}_j$ is a linear combination of the variables 
$x^1_i,\ldots,x^6_i$. As an example let us consider the case $d=2$,
let the first order $\sigma_1$ be $x^1_1\le x^2_1\le x^3_1\le x^4_1\le x^5_1\le x^6_1$,
and let the second order $\sigma_2$ be $x^6_2\le x^5_2\le x^4_2\le x^3_2\le x^2_2\le x^1_2$.
Then we get
\begin{align*}
   \Delta_1^{\sigma} & = \sum_{i=1}^2 (|x^1_i-x^2_i|+|x^3_i-x^4_i|-|x^1_i-x^3_i|-|x^2_i-x^4_i|) \\
                     & = \overbrace{((x^2_1-x^1_1)+(x^4_1-x^3_1)-(x^3_1-x^1_1)-(x^4_1-x^2_1))}^{X_1^{\sigma_1,1}}\\
                     & \quad + \overbrace{((x^1_2-x^2_2)+(x^3_2-x^4_2)-(x^1_2-x^3_2)-(x^2_2-x^4_2))}^{X_1^{\sigma_2,2}}
\end{align*}
and
\begin{align*}
   \Delta_2^{\sigma} & = \sum_{i=1}^2 (|x^1_i-x^3_i|+|x^5_i-x^6_i|-|x^1_i-x^5_i|-|x^3_i-x^6_i|) \\
                     & = \overbrace{((x^3_1-x^1_1)+(x^6_1-x^5_1)-(x^5_1-x^1_1)-(x^6_1-x^3_1))}^{X_2^{\sigma_1,1}}\\
                     & \quad + \overbrace{((x^1_2-x^3_2)+(x^5_2-x^6_2)-(x^1_2-x^5_2)-(x^3_2-x^6_2))}^{X_2^{\sigma_2,2}}.
\end{align*}

If, for one $i\in[d]$, the pair~$(X_1^{\sigma_i,i},X_2^{\sigma_i,i})$ of linear combinations belongs to class C,
then also the pair $(\Delta_1^{\sigma},\Delta_2^{\sigma})$ belongs to class C because the sets of variables occurring
in $X_j^{\sigma_i,i}$ and $X_j^{\sigma_{i'},i'}$ are disjoint for $i\neq i'$. If for all $i\in[d]$ the pair of linear
combinations $(X_1^{\sigma_i,i},X_2^{\sigma_i,i})$ belongs to class A or B, then also the pair $(\Delta_1^{\sigma},\Delta_2^{\sigma})$
belongs either to class A or B. Hence, the following lemma directly implies that $(\Delta_1^{\sigma},\Delta_2^{\sigma})$
belongs to one of the classes A, B, or C.
\begin{lemma}\label{lem:Classes0}
For pairs of type~$0$ and for $i\in[d]$, the pair of linear 
combinations $(X_1^{\sigma_i,i},X_2^{\sigma_i,i})$ belongs either to 
class A, B, or C.
\end{lemma} 

\begin{proof}
Assume that the pair~$(X_1^{\sigma_i,i},X_2^{\sigma_i,i})$ of linear combinations 
is linearly dependent for a 
fixed order $\sigma_i$. Observe that this can only happen 
if the sets of variables occurring in $X_1^{\sigma_i,i}$
and $X_2^{\sigma_i,i}$ are the same. Hence, it can only happen
if the following two conditions occur.
\begin{itemize}
\item $X_1^{\sigma_i,i}$ does not contain $x^2_i$ or $x_i^4$. 
If $x^3_i\ge x_i^4$, it must be true that $x^2_i\ge x_i^4$ in order
for $x_i^4$ to cancel out. Then, in order for $x_i^2$ to cancel out, it must be true that $x^2_i\ge x_i^1$. 
If $x^3_i\le x_i^4$, it must be true that $x^2_i\le x_i^4$ in order for $x_i^4$ to cancel out.
Then, in order for $x_i^2$ to cancel out, it must be true that $x^2_i\le x_i^1$.

Hence, either $x^3_i\ge x_i^4$, $x^2_i\ge x_i^4$, and $x^2_i\ge x_i^1$, or
$x^3_i\le x_i^4$, $x^2_i\le x_i^4$, and $x^2_i\le x_i^1$.

\item $X_2^{\sigma_i,i}$ does not contain $x^5_i$ or $x^6_i$.
If $x^5_i\ge x^6_i$, it must be true that $x^3_i\ge x^6_i$ in order for $x^6_i$ to cancel out,
and it must be true that $x^5_i\ge x_i^1$ in order for $x^5_i$ to cancel out.
If $x^5_i\le x^6_i$, it must be true that $x^3_i\le x^6_i$ in order for $x^6_i$ to cancel out,
and it must be true that $x^5_i\le x_i^1$ in order for $x^5_i$ to cancel out.
 
Hence, either $x^5_i\ge x^6_i$, $x^3_i\ge x^6_i$, and $x^5_i\ge x_i^1$, or 
$x^5_i\le x^6_i$, $x^3_i\le x^6_i$, and $x^5_i\le x_i^1$.
\end{itemize}

Now we choose an order such that $x^2_i$, $x_i^4$, $x^5_i$, and 
$x^6_i$ cancel out. We distinguish between the cases $x_i^1\ge x_i^3$
and $x_i^3\ge x_i^1$.
\begin{itemize}
  \item[$x_i^1\ge x_i^3$:] In this case, we can write $X_1^{\sigma_i,i}$ as
  \begin{align*}
     X_1^{\sigma_i,i} & = |x^1_i-x^2_i|+|x^3_i-x^4_i|-|x^1_i-x^3_i|-|x^2_i-x^4_i| \\
             & = |x^1_i-x^2_i|+|x^3_i-x^4_i|-(x^1_i-x^3_i)-|x^2_i-x^4_i|.             
  \end{align*}
  Since we have argued above that either $x^3_i\ge x_i^4$, $x^2_i\ge x_i^4$, and $x^2_i\ge x_i^1$, or
 $x^3_i\le x_i^4$, $x^2_i\le x_i^4$, and $x^2_i\le x_i^1$, we obtain that either
 \[
     X_1^{\sigma_i,i} = (x^2_i-x^1_i)+(x^3_i-x^4_i)-(x^1_i-x^3_i)-(x^2_i-x^4_i) = -2x^1_i+2x^3_i
 \]
 or
 \[
    X_1^{\sigma_i,i} = (x^1_i-x^2_i)+(x^4_i-x^3_i)-(x^1_i-x^3_i)-(x^4_i-x^2_i) = 0.
 \]
 We can write $X_2^{\sigma_i,i}$ as
 \begin{align*}
     X_2^{\sigma_i,i} & = |x^1_i-x^3_i|+|x^5_i-x^6_i|-|x^1_i-x^5_i|-|x^3_i-x^6_i| \\
                      & = (x^1_i-x^3_i)+|x^5_i-x^6_i|-|x^1_i-x^5_i|-|x^3_i-x^6_i|.             
  \end{align*}
 Since we have argued above that either $x^5_i\ge x^6_i$, $x^3_i\ge x^6_i$, and $x^5_i\ge x_i^1$, or 
 $x^5_i\le x^6_i$, $x^3_i\le x^6_i$, and $x^5_i\le x_i^1$, we obtain that either
 \[
     X_2^{\sigma_i,i} = (x^1_i-x^3_i)+(x^5_i-x^6_i)-(x^5_i-x^1_i)-(x^3_i-x^6_i) = 2x^1_i-2x^3_i 
 \]
 or
 \[
     X_2^{\sigma_i,i} = (x^1_i-x^3_i)+(x^6_i-x^5_i)-(x^1_i-x^5_i)-(x^6_i-x^3_i) = 0. 
 \]
 In summary, the case analysis shows that $X_1^{\sigma_i,i} \in \{0,-2x_i^1+2x^3_i\}$ and 
$X_2^{\sigma_i,i} \in \{0,2x_i^1-2x^3_i\}$. Hence, in this case the 
resulting pair of linear combinations belongs either to class A or B. 

\item[$x_i^3\ge x_i^1$:] In this case, we can write $X_1^{\sigma_i,i}$ as
  \begin{align*}
     X_1^{\sigma_i,i} & = |x^1_i-x^2_i|+|x^3_i-x^4_i|-|x^1_i-x^3_i|-|x^2_i-x^4_i| \\
             & = |x^1_i-x^2_i|+|x^3_i-x^4_i|-(x^3_i-x^1_i)-|x^2_i-x^4_i|.             
  \end{align*}
  Since we have argued above that either $x^3_i\ge x_i^4$, $x^2_i\ge x_i^4$, and $x^2_i\ge x_i^1$, or
 $x^3_i\le x_i^4$, $x^2_i\le x_i^4$, and $x^2_i\le x_i^1$, we obtain that either
 \[
     X_1^{\sigma_i,i} = (x^2_i-x^1_i)+(x^3_i-x^4_i)-(x^3_i-x^1_i)-(x^2_i-x^4_i) = 0
 \]
 or
 \[
    X_1^{\sigma_i,i} = (x^1_i-x^2_i)+(x^4_i-x^3_i)-(x^3_i-x^1_i)-(x^4_i-x^2_i) = 2x^1_i-2x^3_i.
 \]
 We can write $X_2^{\sigma_i,i}$ as
 \begin{align*}
     X_2^{\sigma_i,i} & = |x^1_i-x^3_i|+|x^5_i-x^6_i|-|x^1_i-x^5_i|-|x^3_i-x^6_i| \\
                      & = (x^3_i-x^1_i)+|x^5_i-x^6_i|-|x^1_i-x^5_i|-|x^3_i-x^6_i|.             
  \end{align*}
 Since we have argued above that either $x^5_i\ge x^6_i$, $x^3_i\ge x^6_i$, and $x^5_i\ge x_i^1$, or 
 $x^5_i\le x^6_i$, $x^3_i\le x^6_i$, and $x^5_i\le x_i^1$, we obtain that either
 \[
     X_2^{\sigma_i,i} = (x^3_i-x^1_i)+(x^5_i-x^6_i)-(x^5_i-x^1_i)-(x^3_i-x^6_i) = 0 
 \]
 or
 \[
     X_2^{\sigma_i,i} = (x^3_i-x^1_i)+(x^6_i-x^5_i)-(x^1_i-x^5_i)-(x^6_i-x^3_i) = -2x^1_i+2x^3_i. 
 \]
 In summary, the case analysis shows that $X_1^{\sigma_i,i} \in \{0,2x^1_i-2x^3_i\}$ and 
$X_2^{\sigma_i,i} \in \{0,-2x^1_i+2x^3_i\}$. Hence, also in this case the 
resulting pair of linear combinations belongs either to class A or B.\qed
\end{itemize}
\end{proof}

Now we consider pairs of type~1~a).
Using the same notation as for 
pairs of type~0, we can write the improvement $\Delta_2$ as
\[
   \Delta_2 = \sum_{i\in[d]} (|x^1_i-x^3_i|+|x^2_i-x^5_i|-|x^1_i-x^5_i|-|x^2_i-x^3_i|).
\] 
Again we write, for $j\in[2]$, $\Delta_j^{\sigma}=\sum_{i\in[d]}X^{\sigma_i,i}_j$, where 
$X^{\sigma_i,i}_j$ is a linear combination of the variables 
$x^1_i,\ldots,x^6_i$. Compared to pairs of type~0, only the terms $X^{\sigma_i,i}_2$ are different,
whereas the terms $X^{\sigma_i,i}_1$ do not change.

\begin{lemma}\label{lem:Classes1}
For pairs of type~$1$~a) and for $i\in[d]$, the pair~$(X_1^{\sigma_i,i},X_2^{\sigma_i,i})$ of linear 
combinations belongs either to class A, B, or C.
\end{lemma}
\begin{proof}
Assume that the pair $(X_1^{\sigma_i,i},X_2^{\sigma_i,i})$ is linearly 
dependent for a fixed order $\sigma_i$.
Observe that this can only happen if the sets of variables occurring in $X_1^{\sigma_i,i}$
and $X_2^{\sigma_i,i}$ are the same. Hence, it can only happen
if the following two conditions occur.
\begin{itemize}
\item $X_1^{\sigma_i,i}$ does not contain $x_i^4$.
If $x^3_i\ge x_i^4$, it must be true that $x^2_i\ge x_i^4$ in order for $x_i^4$ to cancel out. 
If $x^3_i\le x_i^4$, it must be true that $x^2_i\le x_i^4$ in order for $x_i^4$ to cancel out.

Hence, either $x^3_i\ge x^4_i$ and $x^2_i\ge x^4_i$, or $x^3_i\le x^4_i$ and $x^2_i\le x^4_i$.

\item $X_2^{\sigma_i,i}$ does not contain $x^5_i$.
If $x^2_i\ge x^5_i$, it must be true that $x^1_i\ge x_i^5$ in order for $x^5_i$ to cancel out.
If $x^2_i\le x^5_i$, it must be true that $x^1_i\le x_i^5$ in order for $x^5_i$ to cancel out.
 
Hence, either $x^2_i\ge x^5_i$ and $x_i^1\ge x^5_i$, or $x^2_i\le x^5_i$ and $x_i^1\le x^5_i$.
\end{itemize}

Now we choose an order such that $x_i^4$ and $x^5_i$ cancel out.
We distinguish between the following cases.
\begin{itemize}
  \item[$x_i^1\ge x_i^3$:] In this case, we can write $X_1^{\sigma_i,i}$ as
  \begin{align*}
     X_1^{\sigma_i,i} & = |x^1_i-x^2_i|+|x^3_i-x^4_i|-|x^1_i-x^3_i|-|x^2_i-x^4_i| \\
             & = |x^1_i-x^2_i|+|x^3_i-x^4_i|-(x^1_i-x^3_i)-|x^2_i-x^4_i|.             
  \end{align*}
  Since we have argued above that either $x^3_i\ge x^4_i$ and $x^2_i\ge x^4_i$, or $x^3_i\le x^4_i$ and $x^2_i\le x^4_i$,
  we obtain that either
 \begin{align*}
     X_1^{\sigma_i,i} & = |x^1_i-x^2_i|+(x^3_i-x^4_i)-(x^1_i-x^3_i)-(x^2_i-x^4_i) \\
                 & = |x^1_i-x^2_i|+2x^3_i-x^1_i-x^2_i \in \{2x_i^3-2x_i^2,2x_i^3-2x_i^1\}.
 \end{align*}
 or
 \begin{align*}
    X_1^{\sigma_i,i} & = |x^1_i-x^2_i|+(x^4_i-x^3_i)-(x^1_i-x^3_i)-(x^4_i-x^2_i) \\
                 & = |x^1_i-x^2_i|-x^1_i+x^2_i \in \{0,-2x_i^1+2x_i^2\}.
  \end{align*}
 We can write $X_2^{\sigma_i,i}$ as
 \begin{align*}
     X_2^{\sigma_i,i} & = |x^1_i-x^3_i|+|x^2_i-x^5_i|-|x^1_i-x^5_i|-|x^2_i-x^3_i| \\
                      & = (x^1_i-x^3_i)+|x^2_i-x^5_i|-|x^1_i-x^5_i|-|x^2_i-x^3_i|.             
  \end{align*}
 Since we have argued above that either $x^2_i\ge x^5_i$ and $x_i^1\ge x^5_i$, or if $x^2_i\le x^5_i$ and $x_i^1\le x^5_i$,
 we obtain that either
 \begin{align*}
     X_2^{\sigma_i,i} & = (x^1_i-x^3_i)+(x^2_i-x^5_i)-(x^1_i-x^5_i)-|x^2_i-x^3_i| \\
                      & =  x^2_i-x^3_i-|x^2_i-x^3_i| \in \{0,2x^2_i-2x^3_i\}
 \end{align*}
 or
 \begin{align*}
     X_2^{\sigma_i,i} & = (x^1_i-x^3_i)+(x^5_i-x^2_i)-(x^5_i-x^1_i)-|x^2_i-x^3_i|\\
                      & = 2x^1_i-x^2_i-x^3_i-|x^2_i-x^3_i| \in \{2x^1_i-2x^2_i,2x^1_i-2x^3_i\}. 
 \end{align*}
 In summary, the case analysis shows that $X_1^{\sigma_i,i} \in \{0,-2x_i^1+2x_i^2,-2x_i^1+2x_i^3,-2x_i^2+2x_i^3\}$ and 
$X_2^{\sigma_i,i} \in \{0,2x^1_i-2x^2_i,2x^1_i-2x^3_i,2x^2_i-2x^3_i\}$. Hence, in this case the 
resulting pair of linear combinations belongs either to class A, B, or C. 
  \item[$x_i^1\le x_i^3$:] In this case, we can write $X_1^{\sigma_i,i}$ as
  \begin{align*}
     X_1^{\sigma_i,i} & = |x^1_i-x^2_i|+|x^3_i-x^4_i|-|x^1_i-x^3_i|-|x^2_i-x^4_i| \\
             & = |x^1_i-x^2_i|+|x^3_i-x^4_i|-(x^3_i-x^1_i)-|x^2_i-x^4_i|.             
  \end{align*}
  Since we have argued above that either $x^3_i\ge x^4_i$ and $x^2_i\ge x^4_i$, or $x^3_i\le x^4_i$ and $x^2_i\le x^4_i$,
  we obtain that either
 \begin{align*}
     X_1^{\sigma_i,i} & = |x^1_i-x^2_i|+(x^3_i-x^4_i)-(x^3_i-x^1_i)-(x^2_i-x^4_i) \\
                 & = |x^1_i-x^2_i|+x^1_i-x^2_i \in \{0,2x^1_i-2x^2_i\}.
 \end{align*}
 or
 \begin{align*}
    X_1^{\sigma_i,i} & = |x^1_i-x^2_i|+(x^4_i-x^3_i)-(x^3_i-x^1_i)-(x^4_i-x^2_i) \\
                 & = |x^1_i-x^2_i|+x^1_i+x^2_i-2x^3_i \in \{2x^1_i-2x^3_i,2x_i^2-2x_i^3\}.
  \end{align*}
 We can write $X_2^{\sigma_i,i}$ as
 \begin{align*}
     X_2^{\sigma_i,i} & = |x^1_i-x^3_i|+|x^2_i-x^5_i|-|x^1_i-x^5_i|-|x^2_i-x^3_i| \\
                      & = (x^3_i-x^1_i)+|x^2_i-x^5_i|-|x^1_i-x^5_i|-|x^2_i-x^3_i|.             
  \end{align*}
 Since we have argued above that either $x^2_i\ge x^5_i$ and $x_i^1\ge x^5_i$, or $x^2_i\le x^5_i$ and $x_i^1\le x^5_i$,
 we obtain that either
 \begin{align*}
     X_2^{\sigma_i,i} & = (x^3_i-x^1_i)+(x^2_i-x^5_i)-(x^1_i-x^5_i)-|x^2_i-x^3_i| \\
                      & =  -2x^1_i+x^2_i+x^3_i-|x^2_i-x^3_i| \in \{-2x^1_i+2x^3_i,-2x^1_i+2x^2_i\}
 \end{align*}
 or
 \begin{align*}
     X_2^{\sigma_i,i} & = (x^3_i-x^1_i)+(x^5_i-x^2_i)-(x^5_i-x^1_i)-|x^2_i-x^3_i|\\
                      & = -x^2_i+x^3_i-|x^2_i-x^3_i| \in \{0,-2x^2_i+2x^3_i\}. 
 \end{align*}
 In summary, the case analysis shows that $X_1^{\sigma_i,i} \in \{0,2x^1_i-2x^2_i,2x^1_i-2x^3_i,2x_i^2-2x_i^3\}$ and 
$X_2^{\sigma_i,i} \in \{0,-2x^1_i+2x^2_i,-2x^1_i+2x^3_i,-2x^2_i+2x^3_i\}$. Hence, in this case the 
resulting pair of linear combinations belongs either to class A, B, or C.\qed
\end{itemize}
\end{proof}

Finally we consider pairs of type~1~b).
Using the same notation as before, we can write the improvement $\Delta_2$ as
\[
   \Delta_2 = \sum_{i=1}^d \big(|x^1_i-x^3_i|+|x^2_i-x^5_i|-|x^1_i-x^2_i|-|x^3_i-x^5_i|\big).
\]
Again we write, for $j\in[2]$, $\Delta_j^{\sigma}=\sum_{i\in[d]}X^{\sigma_i,i}_j$, where 
$X^{\sigma_i,i}_j$ is a linear combination of the variables 
$x^1_i,\ldots,x^6_i$. And again only the terms $X^{\sigma_i,i}_2$ are different from before.

\begin{lemma}\label{lem:Classes1b}
For pairs of type~$1$~b) and for $i\in[d]$, the pair of linear 
combinations $(X_1^{\sigma_i,i},X_2^{\sigma_i,i})$ belongs either to 
class A, B, or C.
\end{lemma}
\begin{proof}
Using the same notation as for 
pairs of type~0, we can write the improvement $\Delta_2$ as
\[
   \Delta_2 = \sum_{i=1}^d \big(|x^1_i-x^3_i|+|x^2_i-x^5_i|-|x^1_i-x^2_i|-|x^3_i-x^5_i|\big).
\]

Assume that the pair $(X_1^{\sigma_i,i},X_2^{\sigma_i,i})$ is linearly 
dependent for a fixed order $\sigma_i$.
Observe that this can only happen if the sets of variables occurring in $X_1^{\sigma_i,i}$
and $X_2^{\sigma_i,i}$ are the same. Hence, it can only happen
if the following two conditions occur.

\begin{itemize}
\item $X_1^{\sigma_i,i}$ does not contain $x_i^4$. We have considered this condition already
for pairs of type~1~a) and showed that
either $x^3_i\ge x^4_i$ and $x^2_i\ge x^4_i$, or $x^3_i\le x^4_i$ and $x^2_i\le x^4_i$.

\item $X_2^{\sigma_i,i}$ does not contain $x^5_i$.
If $x^2_i\ge x^5_i$, it must be true that $x^3_i\ge x_i^5$ in order for $x^5_i$ to cancel out.
If $x^2_i\le x^5_i$, it must be true that $x^3_i\le x_i^5$ in order for $x^5_i$ to cancel out.
 
Hence, either $x^2_i\ge x^5_i$ and $x_i^3\ge x^5_i$, or $x^2_i\le x^5_i$ and $x_i^3\le x^5_i$.
\end{itemize}

Now we choose an order such that $x_i^4$ and $x^5_i$ cancel out.
We distinguish between the following cases.
\begin{itemize}
  \item[$x_i^1\ge x_i^3$:] We have argued already for pairs of type~1~a) that in this case
  $X_1^{\sigma_i,i} \in \{0,-2x_i^1+2x_i^2,-2x_i^1+2x_i^3,-2x_i^2+2x_i^3\}$. 
  
  We can write $X_2^{\sigma_i,i}$ as
 \begin{align*}
     X_2^{\sigma_i,i} & = |x^1_i-x^3_i|+|x^2_i-x^5_i|-|x^1_i-x^2_i|-|x^3_i-x^5_i| \\
                      & = (x^1_i-x^3_i)+|x^2_i-x^5_i|-|x^1_i-x^2_i|-|x^3_i-x^5_i|.             
  \end{align*}
 Since we have argued above that either $x^2_i\ge x^5_i$ and $x_i^3\ge x^5_i$, or $x^2_i\le x^5_i$ and $x_i^3\le x^5_i$,
 we obtain that either
 \begin{align*}
     X_2^{\sigma_i,i} & = (x^1_i-x^3_i)+(x^2_i-x^5_i)-|x^1_i-x^2_i|-(x^3_i-x^5_i) \\
                      & =  x^1_i+x^2_i-2x^3_i-|x^1_i-x^2_i| \in \{2x^2_i-2x^3_i,2x^1_i-2x^3_i\}
 \end{align*}
 or
 \begin{align*}
     X_2^{\sigma_i,i} & = (x^1_i-x^3_i)+(x^5_i-x^2_i)-|x^1_i-x^2_i|-(x^5_i-x^3_i) \\
                      & = x^1_i-x^2_i-|x^1_i-x^2_i| \in \{0,2x^1_i-2x^2_i\}. 
 \end{align*}
 In summary, the case analysis shows that $X_1^{\sigma_i,i} \in \{0,-2x_i^1+2x_i^2,-2x_i^1+2x_i^3,-2x_i^2+2x_i^3\}$ and 
$X_2^{\sigma_i,i} \in \{0,2x_i^1-2x_i^2,2x_i^1-2x_i^3,2x_i^2-2x_i^3\}$. Hence, in this case the 
resulting pair of linear combinations belongs either to class A, B, or C.

  \item[$x_i^1\le x_i^3$:] We have argued already for pairs of type~1~a) that in this case
  $X_1^{\sigma_i,i} \in \{0,2x^1_i-2x^2_i,2x^1_i-2x^3_i,2x_i^2-2x_i^3\}$. 
  
  We can write $X_2^{\sigma_i,i}$ as
 \begin{align*}
     X_2^{\sigma_i,i} & = |x^1_i-x^3_i|+|x^2_i-x^5_i|-|x^1_i-x^2_i|-|x^3_i-x^5_i| \\
                      & = (x^3_i-x^1_i)+|x^2_i-x^5_i|-|x^1_i-x^2_i|-|x^3_i-x^5_i|.             
  \end{align*}
 Since we have argued above that either $x^2_i\ge x^5_i$ and $x_i^3\ge x^5_i$, or $x^2_i\le x^5_i$ and $x_i^3\le x^5_i$,
 we obtain that either
 \begin{align*}
     X_2^{\sigma_i,i} & = (x^3_i-x^1_i)+(x^2_i-x^5_i)-|x^1_i-x^2_i|-(x^3_i-x^5_i) \\
                      & =  -x^1_i+x^2_i-|x^1_i-x^2_i| \in \{0,-2x^1_i+2x^2_i\}
 \end{align*}
 or
 \begin{align*}
     X_2^{\sigma_i,i} & = (x^3_i-x^1_i)+(x^5_i-x^2_i)-|x^1_i-x^2_i|-(x^5_i-x^3_i) \\
                      & = -x^1_i-x^2_i+2x^3_i-|x^1_i-x^2_i| \in \{-2x^1_i+2x^3_i,-2x^2_i+2x^3_i\}. 
 \end{align*}
 In summary, the case analysis shows that $X_1^{\sigma_i,i} \in \{0,2x^1_i-2x^2_i,2x^1_i-2x^3_i,2x_i^2-2x_i^3\}$ and 
$X_2^{\sigma_i,i} \in \{0,-2x^1_i+2x^2_i,-2x^1_i+2x^3_i,-2x^2_i+2x^3_i\}$. Hence, in this case the 
resulting pair of linear combinations belongs either to class A, B, or C. \qed
\end{itemize}
\end{proof}
We have argued above that for tuples $\sigma$ of orders that yield pairs from class A or B,
the event ${\cal A}^{\sigma}$ cannot occur.
For tuples $\sigma$ that yield pairs from class C, we can apply 
Lemma~\ref{lemma:ProbLinComb} from Appendix~\ref{appendix:probTheory}, 
which shows that the probability of the event ${\cal A}^{\sigma}$ is 
bounded from above by $(\varepsilon\phi)^2$. As we have shown that every tuple yields a pair from
class A, B, or C, we can conclude the proof of Lemma~\ref{lemma:L1Improvements} by a union bound over
all pairs of linked 2-changes of type~0 and~1 and all tuples $\sigma$. As these are $O(n^6)$, the lemma follows. 
\end{proof}

\subsubsection{Expected number of 2-changes}

Based on Lemmas~\ref{lemma:numberLinkedPairs2} 
and~\ref{lemma:L1Improvements}, we are now able to prove part a) of 
Theorem~\ref{theorem:runningTime1}.
\begin{proof}[Theorem~\ref{theorem:runningTime1} a)]
Let $T$ denote the random variable that describes the length of the 
longest path in the state graph. If $T\ge t$, then there must exist a 
sequence $S_1,\ldots,S_t$ of $t$ consecutive 2-changes in the state 
graph. We start by identifying a set of linked pairs of type~0 and 1 in 
this sequence. Due to Lemma~\ref{lemma:numberLinkedPairs2}, we know 
that we can find at least $z=t/7-3n/28$ such pairs. Let 
$\Delta_{\min}^*$ denote the smallest improvement made by any pair of 
improving 2-Opt steps of type 0 or 1. If~$T\ge t$, then~$\Delta_{\min}^*\le \frac{dn}{z}$
as the initial tour has length at most~$dn$ and every linked pair of type~0 or~1
decreases the length of the tour by at least~$\Delta_{\min}^*$. 
For $t>n$, we have 
$z=t/7-3n/28>t/28$ and hence due to 
Lemma~\ref{lemma:L1Improvements},
\[
  \Pr{T\ge t}\le \Pr{\Delta_{\min}^*\le \frac{dn}{z}} \le 
  \Pr{\Delta_{\min}^*\le \frac{28dn}{t}} 
  = O\left(\frac{n^8\phi^2}{t^2}\right).
\]
Using the fact that probabilities are bounded from above by one, we obtain
\[
  \Pr{T\ge t} = O\left(\min\left\{\frac{n^8\phi^2}{t^2},1\right\}\right).
\]

Since $T$ cannot exceed $n!$, this implies the following bound on the 
expected number of 2-changes:
\begin{align*}
  \Ex{T} & \le n+\sum_{t=n+1}^{n!}O\left(\min\left\{\frac{n^8\phi^2}{t^2},1\right\}\right) \\
         & = n+O\left(\sum_{t=n+1}^{n^4\phi}1\right)+O\left(\sum_{t=n^4\phi+1}^{n!}\frac{n^8\phi^2}{t^2}\right)
         = O(n^4\cdot\phi).
\end{align*}

This concludes the proof of part a) of the theorem.
\end{proof}

Chandra, Karloff, and Tovey~\cite{ChandraKT99} show that for every 
metric that is induced by a norm on $\RR^d$, and for any set of $n$ 
points in the unit hypercube $[0,1]^d$, the optimal tour visiting all 
$n$ points has length $O(n^{(d-1)/d})$. Furthermore, every insertion 
heuristic finds an $O(\log{n})$-approximation~\cite{RosenkrantzSL77}. 
Hence, if one starts with a solution computed by an insertion 
heuristic, the initial tour has length $O(n^{(d-1)/d}\cdot\log{n})$. 
Using this observation yields part a) of 
Theorem~\ref{theorem:runningTime2}:

\begin{proof}[Theorem~\ref{theorem:runningTime2} a)]
Since the initial tour has length $O(n^{(d-1)/d}\cdot\log{n})$, we obtain
for an appropriate constant $c$ and $t>n$,
\begin{alignat*}{1}
  \Pr{T\ge t} & \le \Pr{\Delta_{\min}^*\le \frac{c\cdot 
  n^{(d-1)/d}\cdot\log{n}}{t}}\\  
  & =  O\left(\min\left\{\frac{n^{8-2/d}\cdot\log^2{n}\cdot\phi^2}{t^2},1\right\}\right),
\end{alignat*}
where the equality follows from Lemma~\ref{lemma:L1Improvements}.
This yields
\[
  \Ex{T} \le
  n+\sum_{t=n+1}^{n!}O\left(\min\left\{\frac{n^{8-2/d}\cdot\log^2{n}\cdot\phi^2}{t^2},1\right\}\right)
          = O(n^{4-1/d}\cdot\log{n}\cdot\phi).\tag*{\qed}
\]
\end{proof}

\subsection{Euclidean Instances}
\label{subsec:UB:L2}

In this section, we analyze the expected number of 2-changes on 
$\phi$-perturbed Euclidean instances. The analysis is similar to the 
analysis of Manhattan instances in the previous section; only 
Lemma~\ref{lemma:L1Improvements} needs to be replaced by the following
equivalent version for the $L_2$ metric, which will be proved later
in this section.

\begin{lemma}
\label{lemma:L2Improvements}
For $\phi$-perturbed $L_2$ instances, the probability that there exists 
a pair of type~0 or type~1 in which both 2-changes are improvements by 
at most $\varepsilon\le1/2$ is bounded by 
$O(n^6\cdot\phi^5\cdot\varepsilon^{2}\cdot\log^2(1/\varepsilon))
+O(n^5\cdot\phi^4\cdot\varepsilon^{3/2}\cdot\log(1/\varepsilon))$.
\end{lemma}

The bound that this lemma provides is slightly weaker than its $L_1$ 
counterpart, and hence also the bound on the expected running time is 
slightly worse for $L_2$ instances. The crucial step to proving 
Lemma~\ref{lemma:L2Improvements} is to gain a better understanding of 
the random variable that describes the improvement of a \emph{single} fixed 
2-change. In the next section, we analyze this random variable under 
several conditions, \eg, under the condition that the length of one of 
the involved edges is fixed. With the help of these results, pairs of 
linked 2-changes can easily be analyzed. Let us mention that our 
analysis of a single 2-change yields a bound of 
$O(n^7\cdot\log^2(n)\cdot\phi^3)$ for the expected number of 2-changes. 
For Euclidean instances in which all points are distributed uniformly 
at random over the unit square, this bound already improves the best 
previously known bound of $O(n^{10}\cdot\log{n})$.

\subsubsection{Analysis of a single 2-change}

We analyze a 2-change in which the edges $\{O,Q_1\}$ and $\{P,Q_2\}$ 
are exchanged with the edges $\{O,Q_2\}$ and $\{P,Q_1\}$ for some 
vertices $O$, $P$, $Q_1$, and $Q_2$. In the input model we consider, 
each of these points has a probability distribution over the unit 
hypercube according to which it is chosen. In this section, we consider 
a simplified random experiment in which $O$ is chosen to be the origin 
and $P$, $Q_1$, and $Q_2$ are chosen independently and uniformly at 
random from a $d$-dimensional hyperball with radius $\sqrt{d}$ centered 
at the origin. In the next section, we argue that the analysis of this 
simplified random experiment helps to analyze the actual random 
experiment that occurs in the probabilistic input model.

Due to the rotational symmetry of the simplified model, we assume 
without loss of generality that $P$ lies at position $(0^{d-1},T)$ for some 
$T\ge0$. For $i\in[2]$, Let $Z_i$ denote the difference 
$\dist(O,Q_i)-\dist(P,Q_i)$. Then the improvement $\Delta$ of the 
2-change can be expressed as $Z_1-Z_2$. The random variables $Z_1$ and 
$Z_2$ are identically distributed, and they are independent if $T$ is 
fixed.
We denote by $f_{Z_1|T=\tau,R=r}$ the density of $Z_1$ 
conditioning on the fact that $\dist(O,Q_1)=r$ and $T=\tau$.
Similarly, we denote by $f_{Z_2|T=\tau,R=r}$ the density of $Z_2$ 
conditioning on the fact that $\dist(O,Q_2)=r$ and $T=\tau$.
As $Z_1$ and $Z_2$ are identically distributed, the conditional
densities $f_{Z_1|T=\tau,R=r}$ and $f_{Z_2|T=\tau,R=r}$ are identical
as well. Hence, we can drop the index in the following and write
$f_{Z|T=\tau,R=r}$.
\begin{lemma}\label{lemma:densityZ}
For $\tau,r\in(0,\sqrt{d}]$, and $z\in(-\tau,\min\{\tau,2r-\tau\})$,
\[ 
f_{Z|T=\tau,R=r}(z)  \le \begin{cases}
\sqrt{\frac{2}{\tau^2-z^2}} & \mbox{if } r\ge\tau, \\
\sqrt{\frac{2}{(\tau+z)(2r-\tau-z)}} & \mbox {if } r<\tau.
\end{cases}
\] For 
$z\notin[-\tau,\min\{\tau,2r-\tau\}]$, the density 
$f_{Z|T=\tau,R=r}(z)$ is $0$.
\end{lemma}
\begin{proof}
We denote by $Z$ the random variable $\dist(O,Q)-\dist(P,Q)$, where $Q$ 
is a point chosen uniformly at random from a $d$-dimensional hyperball 
with radius $\sqrt{d}$ centered at the origin. In the following, we 
assume that the plane spanned by the points $O$, $P$, and $Q$ is fixed 
arbitrarily, and we consider the random experiment conditioned on the 
event that $Q$ lies in this plane. To make the calculations simpler, we 
use polar coordinates to describe the location of $Q$. Since the radius 
$\dist(O,Q)=r$ is given, the point $Q$ is completely determined by the 
angle $\alpha$ between the $y$-axis and the line between $O$ and $Q$ 
(see Figure~\ref{fig:UB:Triangle}). Hence, the random variable $Z$ can 
be written as
\[
   Z = r - \sqrt{r^2+\tau^2-2r\tau\cdot\cos{\alpha}}.
\]
It is easy to see that $Z$ can only take values in the interval 
$[-\tau,\min\{\tau,2r-\tau\}]$, and hence the density
$f_{Z|T=\tau,R=r}(z)$ is $0$ outside this interval.
\begin{figure}[H]
\begin{center}
\includegraphics{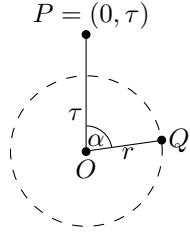}
\caption{The random variable $Z$ is defined as $r-\dist(P,Q)$.}
\label{fig:UB:Triangle}
\end{center}
\end{figure}

Since $Q$ is chosen uniformly at random from a hyperball centered at
the origin, rotational symmetry implies that
the angle $\alpha$ is chosen 
uniformly at random from the interval $[0,2\pi)$. For symmetry reasons, 
we can assume that $\alpha$ is chosen uniformly from the interval 
$[0,\pi)$. When $\alpha$ is restricted to the interval $[0,\pi)$, 
there exists a unique inverse function mapping $Z$ to $\alpha$, namely
\[
 \alpha(z) = \arccos\left(\frac{\tau^2+2zr-z^2}{2r\tau}\right).
\]
For $|x|<1$, the derivative of the arc cosine is
\[
   (\arccos(x))' = -\frac{1}{\sqrt{1-x^2}} \le 0.
\]
Hence, the density $f_{Z|T=\tau,R=r}$ can be expressed as
\[
   f_{Z|T=\tau,R=r}(z) = f_{\alpha}(\alpha(z))\cdot \left|\frac{d}{dz}\alpha(z)\right|
          = -\frac{1}{\pi}\cdot \frac{d}{dz}\alpha(z)
          \le -\frac{d}{dz}\alpha(z),
\]
where $f_{\alpha}$ denotes the density of $\alpha$, \ie, the density of the uniform 
distribution over $[0,\pi)$.
Using the chain rule, we obtain that the derivative of $\alpha(z)$ equals
\begin{alignat*}{1}
  & \frac{r-z}{r \tau}\cdot\frac{-1}{\sqrt{1-\frac{(\tau^2+2zr-z^2)^2}{4r^2\tau^2}}}\\
  = & \frac{2(z-r)}{\sqrt{4r^2\tau^2-4r^2z^2-4r\tau^2z+4rz^3-\tau^4+2\tau^2z^2-z^4}}.
\end{alignat*}
In order to prove the lemma, we distinguish between the cases $r\ge\tau$
and $r< \tau$.

{\bf First case: $r\ge \tau$.} \\
In this case, it suffices to show that
\begin{equation}\label{eqn:LemmaZ}
  4r^2\tau^2-4r^2z^2-4r\tau^2z+4rz^3-\tau^4+2\tau^2z^2-z^4 \ge 2(z-r)^2(\tau^2-z^2),
\end{equation}
which is implied by
\begin{alignat*}{1}
   &    4r^2\tau^2-4r^2z^2-4r\tau^2z+4rz^3-\tau^4+2\tau^2z^2-z^4 - 2(z-r)^2(\tau^2-z^2) \\
   & =  2r^2(\tau^2-z^2)-\tau^4+z^4  
   \ge 2\tau^2(\tau^2-z^2)-\tau^4+z^4 
    = (\tau^2-z^2)^2 \ge 0.
\end{alignat*}
This proves the lemma for $r\ge \tau$ because
\begin{align*}
   -\frac{d}{dz}\alpha(z) & =
   -\frac{2(z-r)}{\sqrt{4r^2\tau^2-4r^2z^2-4r\tau^2z+4rz^3-\tau^4+2\tau^2z^2-z^4}}\\
   & \le -\frac{2(z-r)}{\sqrt{2(z-r)^2(\tau^2-z^2)}}
   = -\frac{2(z-r)}{|z-r|\sqrt{2(\tau^2-z^2)}} = \sqrt{\frac{2}{\tau^2-z^2}},
\end{align*}
where we have used~\eqref{eqn:LemmaZ} for the inequality.

{\bf Second case: $r< \tau$.} \\
In this case, it suffices to show that
\[
  4r^2\tau^2-4r^2z^2-4r\tau^2z+4rz^3-\tau^4+2\tau^2z^2-z^4 \ge 2(z-r)^2(\tau+z)(2r-\tau-z),
\]
which is implied by
\begin{alignat}{1}
   & 4r^2\tau^2\!-\!4r^2z^2\!-\!4r\tau^2z\!+\!4rz^3\!-\!\tau^4\!+\!2\tau^2z^2\!-\!z^4\!-\!2(z\!-\!r)^2(\tau\!+\!z)(2r\!-\!\tau-z) \ge 0 \notag \\
   &\iff (-2r+z+\tau)(\tau+z)(z^2+2\tau z-2rz+2r^2-\tau^2-2\tau r) \ge 0 \notag \\
   &\iff z^2+2\tau z-2rz+2r^2-\tau^2-2\tau r \le 0, \label{eqn:equivalenceA}
\end{alignat}
where the first equivalence follows because the left hand sides of the first and second inequality are identical and
where the last equivalence follows because $(-2r+z+\tau)<0$ and $(\tau+z)>0$.
Both these inequalities are true because $z\in(-\tau,\min\{\tau,2r-\tau\})$.
Inequality~\eqref{eqn:equivalenceA} follows from
\begin{equation*}
\begin{split}
    &  z^2+2\tau z-2rz+2r^2-\tau^2-2\tau r \\
   =  & z^2+2z(\tau-r)+2r^2-\tau^2-2\tau r \\
  \le &(2r-\tau)^2+2(2r-\tau)(\tau-r)+2r^2-\tau^2-2\tau r \\
   = & 2(r^2-\tau^2) \le 0,
\end{split}
\end{equation*}
where the first inequality follows because $z\le 2r-\tau$.
\end{proof}

Based on Lemma~\ref{lemma:densityZ}, the density of the random variable 
$\Delta=Z_1-Z_2$ under the conditions $R_1:=\dist(O,Q_1)=r_1$, 
$R_2:=\dist(O,Q_2)=r_2$, and $T:=\dist(O,P)=\tau$ can be computed as
the convolution of the densities of the random variables $Z_1$ and $-Z_2$. 
The former density equals $f_{Z|T=\tau,R=r}$ and the latter density can easily
be obtained from $f_{Z|T=\tau,R=r}$.
\begin{lemma}
\label{lemma:densityDelta}
Let $\tau,r_1,r_2\in(0,\sqrt{d}]$, and let $Z_1$ and $Z_2$ be 
independent random variables drawn according to the densities 
$f_{Z|T=\tau,R=r_1}$ and $f_{Z|T=\tau,R=r_2}$, respectively. For
$\delta\in(0,1/2]$ and a sufficiently large constant $\kappa$, the density
$f_{\Delta|T=\tau,R_1=r_1,R_2=r_2}(\delta)$ of the random variable
$\Delta=Z_1-Z_2$ is bounded from above by
\[
  \begin{cases}
   \frac{\kappa}{\tau}\cdot\ln\left(\delta^{-1}\right)
      & \mbox{if } \tau\le r_1, \tau\le r_2, \\
   \frac{\kappa}{\sqrt{r_1r_2}}\cdot\left(\ln\left(\delta^{-1}\right)+\ln|2(r_1-r_2)-\delta|^{-1}\right)
      & \mbox{if } r_1\le \tau, r_2\le\tau, \delta\neq 2(r_1-r_2), \\
   \frac{\kappa}{\sqrt{\tau r_1}}\cdot\ln\left(\delta^{-1}\right) &\mbox{if } r_1\le \tau \le r_2, \\
   \frac{\kappa}{\sqrt{\tau r_2}}\cdot\left(\ln\left(\delta^{-1}\right)+\ln|2(\tau-r_2)-\delta|^{-1}\right) & \mbox{if } r_2\le \tau \le r_1, \delta\neq 2(\tau-r_2).
 \end{cases}
\]
\end{lemma}

The simple but somewhat tedious calculation that yields 
Lemma~\ref{lemma:densityDelta} is deferred to 
Appendix~\ref{appendix:proof:densityDelta}. In order to prove 
Lemma~\ref{lemma:L2Improvements}, we need bounds on the densities of 
the random variables $\Delta$, $Z_1$, and $Z_2$ under certain 
conditions. We summarize these bounds in the following lemma.
\begin{lemma}
\label{lemma:UB:DensitiesDeltaZ}
Let $\tau,r\in(0,\sqrt{d}]$, $\delta\in(0,1/2]$, and
let $\kappa$ denote a sufficiently large constant.
\begin{enumerate}
\renewcommand{\labelenumi}{\alph{enumi})}
  \item For $i\in[2]$, the density of $\Delta$ under the condition
  $R_i=r$ is bounded by
   \[
  f_{\Delta|R_i=r}(\delta) \le \frac{\kappa}{\sqrt{r}}\cdot\ln\left(\delta^{-1}\right).
   \]
\item The density of $\Delta$ under the condition $T=\tau$ is bounded by
   \[
  f_{\Delta|T=\tau}(\delta) \le \frac{\kappa}{\tau}\cdot\ln\left(\delta^{-1}\right).
\]
\item The density of $\Delta$ is bounded by
\[
   f_{\Delta}(\delta) \le   
   \kappa \cdot\ln\left(\delta^{-1}\right).
\]
\item For $i\in[2]$, the density of $Z_i$ under the condition
$T=\tau$ is bounded by
 \[
  f_{Z_i|T=\tau}(z) \le \frac{\kappa}{\sqrt{\tau^2-z^2}}
\]
if $|z|< \tau$. Since $Z_i$ takes only values in 
the interval $[-\tau,\tau]$, the conditional density 
$f_{Z_i|T=\tau}(z)$ is $0$ for $z\notin[-\tau,\tau]$.
\end{enumerate}
\end{lemma}
Lemma~\ref{lemma:UB:DensitiesDeltaZ} follows from 
Lemmas~\ref{lemma:densityZ} and~\ref{lemma:densityDelta} by integrating 
over all values of the unconditioned distances. The proof can be found 
in Appendix~\ref{appendix:proof:densities}.

\subsubsection{Simplified random experiments}
\label{subsubsec:Simplified}

In the previous section we did not analyze the random experiment that 
really takes place. Instead of choosing the points according to the 
given density functions, we simplified their distributions by placing 
point $O$ in the origin and by giving the other points $P$, $Q_1$, and 
$Q_2$ uniform distributions centered around the origin. In our input 
model, however, each of these points is described by a density function 
over the unit hypercube. We consider the probability of the event 
$\Delta\in[0,\varepsilon]$ in the original input model as well as 
in the simplified random experiment. In the following, we denote this 
event by ${\cal E}$. We claim that the simplified random experiment 
that we analyze is only slightly dominated by the original random 
experiment, in the sense that the probability of the event ${\cal E}$ 
in the simplified random experiment is smaller by at most some factor 
depending on $\phi$.

In order to compare the probabilities in the original and in the 
simplified random experiment, consider the original experiment and 
assume that the point $O$ lies at position $x\in[0,1]^d$. Then one 
can identify a region ${\cal R}_{x}\subseteq\RR^{3d}$ with the 
property that the event ${\cal E}$ occurs if and only if the random 
vector $(P,Q_1,Q_2)$ lies in ${\cal R}_{x}$. No matter how the 
position $x$ of $O$ is chosen, this region always has the same 
shape, only its position is shifted. That is, 
${\cal R}_{x}=\{(x,x,x)+{\cal R}_{0^d}\}$.
Let ${\cal V}=\sup_{x\in[0,1]^d}\Vol({\cal R}_{x}\cap[0,1]^{3d})$. 
Then the probability of ${\cal E}$ can be bounded from above by 
$\phi^3\cdot{\cal V}$ in the original random experiment because the density
of the random vector $(P,Q_1,Q_2)$ is bounded from above by $\phi^3$
as $P$, $Q_1$, and $Q_2$ are independent vectors whose densities are bounded
by $\phi$. Since $\Delta$ is invariant under translating $O$, $P$, $Q_1$, and $Q_2$ by
the same vector, we obtain
\begin{alignat*}{1}
    \Vol\left({\cal R}_{x}\cap[0,1]^{3d}\right)
  &= \Vol\left({\cal R}_{0^d}\cap([-x_1,1-x_1]\times\cdots\times[-x_d,1-x_d])^3\right)\\
  &\le \Vol\left({\cal R}_{0^d}\cap[-1,1]^{3d}\right),
\end{alignat*}
where the equality follows from shifting ${\cal R}_{x}\cap[0,1]^{3d}$
by $(-x,-x,-x)$.
Hence, ${\cal V}\le {\cal V}':=\Vol({\cal R}_{0^d}\cap[-1,1]^{3d})$. 
In the simplified random experiment, $P$, $Q_1$, and $Q_2$ are chosen
uniformly from the hyperball centered at the origin with radius $\sqrt{d}$.
This hyperball contains the hypercube $[-1,1]^d$ completely. 
Hence, the region on which the vector $(P,Q_1,Q_2)$ is uniformly distributed   
contains the region ${\cal R}_{0^d}\cap[-1,1]^{3d}$ completely.
As the vector $(P,Q_1,Q_2)$ is uniformly distributed on a region of volume
$V_d(\sqrt{d})^3$, where $V_d(\sqrt{d})$ denotes the 
volume of a $d$-dimensional hyperball with radius $\sqrt{d}$,
this implies that the probability of ${\cal E}$ in the simplified random experiment can be bounded from below by 
${\cal V}'/V_d(\sqrt{d})^3$.
Since a $d$-dimensional hyperball with radius $\sqrt{d}$ is contained in a hypercube with side length 
$2\sqrt{d}$, its volume can be bounded from above by $(2\sqrt{d})^d=(4d)^{d/2}$. 
Hence, the probability of ${\cal E}$ in the simplified random 
experiment is at least ${\cal V}'/(4d)^{3d/2}$, and we have argued above
that the probability of ${\cal E}$ in the original random experiment
is at most $\phi^3\cdot{\cal V}\le \phi^3\cdot{\cal V}'$. Hence, 
the probability of ${\cal E}$ in the simplified random 
experiment is smaller by at most a factor of $((4d)^{d/2}\phi)^3$ 
compared to the original random experiment.

Taking into account this factor and using 
Lemma~\ref{lemma:UB:DensitiesDeltaZ}~c) and a union bound over all 
possible 2-changes yields the following lemma about the improvement of 
a single 2-change.
\begin{lemma}\label{lemma:ImprovementSingle}
The probability that there exists an improving 2-change whose 
improvement is at most $\varepsilon\le1/2$ is bounded from above by 
$O(n^4\cdot\phi^3\cdot\varepsilon\cdot\log(1/\varepsilon))$.
\end{lemma}
\begin{proof}
As in the proof of Theorem~\ref{theorem:L1weak}, we first 
consider a fixed 2-change $S$, whose improvement we denote 
by $\Delta(S)$.
For the simplified random experiment, Lemma~\ref{lemma:UB:DensitiesDeltaZ}~c)
yields the following bound on the probability that the improvement $\Delta(S)$
lies in $(0,\varepsilon]$:
\begin{alignat*}{1}
   \Pr{\Delta(S)\in(0,\varepsilon]} & = \kappa \int_{0}^{\varepsilon}\ln\left(\delta^{-1}\right)\,d\delta
   = \left[\delta\ln\left(\delta^{-1}\right)+\delta\right]_{0}^{\varepsilon}\\
   & = \varepsilon\ln\varepsilon^{-1}+\varepsilon \le 3\varepsilon\ln\varepsilon^{-1},
\end{alignat*} 
where we used $\varepsilon\le1/2$ for the last inequality.

We have argued that the probability of the event $\Delta(S)\in(0,\varepsilon]$ in the
simplified random experiment is smaller by at most a factor of $((4d)^{d/2}\phi)^3$
compared to the original random experiment. Together with the factor of at most $n^4$ coming from 
a union bound over all possible 2-changes~$S$, we obtain for the original random experiment
\[
   \Pr{\exists S: \Delta(S)\in(0,\varepsilon]} \le 3\varepsilon\ln\varepsilon^{-1} \cdot ((4d)^{d/2}\phi)^3 \cdot n^4,
\] 
which proves the lemma because $d$ is regarded as a constant.
\end{proof}

Using similar arguments as in the proof of 
Theorem~\ref{theorem:L1weak} yields the following upper bound on the 
expected number of 2-changes.
\begin{theorem}
\label{theorem:L2weak}
Starting with an arbitrary tour, the expected number of steps 
performed by 2-Opt on $\phi$-perturbed Euclidean instances is 
$O(n^7\cdot\log^2{(n)}\cdot\phi^3)$.
\end{theorem}
\begin{proof}
As in the proof of Theorem~\ref{theorem:L1weak}, let $T$ denote the longest path in the state graph.
Let $\Delta_{\min}$ denote the smallest improvement made by any of the 2-changes.
Then, as in the proof of Theorem~\ref{theorem:L1weak}, we know that $T\ge t$ implies that
$\Delta_{\min}\le (\sqrt{d}n)/t$ because each of the $n$ edges in the initial tour has length at most $\sqrt{d}$.
As $T$ cannot exceed $n!$, we obtain with Lemma~\ref{lemma:ImprovementSingle}
\begin{align*}
   \Ex{T} & = \sum_{t=1}^{n!} \Pr{T \ge t} \le \sum_{t=1}^{n!} \Pr{\Delta_{\min}\le \frac{\sqrt{d}n}{t}} \\
    & = O\left(\sum_{t=1}^{n!} \frac{n^5\cdot\phi^3\cdot\sqrt{d}}{t}\cdot\log\left(\frac{t}{\sqrt{d}n}\right)\right) \\
    & = O\left(\sum_{t=1}^{n!} \frac{n^5\cdot\phi^3\cdot\sqrt{d}}{t}\cdot\ln{t}\right) \\
    & = O\left(n^5\cdot\phi^3\cdot\sqrt{d}\cdot \int_{t=1}^{n!}\frac{\ln{t}}{t}\,dt\right) \\
    & = O\left(n^5\cdot\phi^3\cdot\sqrt{d}\cdot \left[\frac12\ln^2{t}\right]_{t=1}^{n!}\right) \\
    & = O\left(n^7\cdot\phi^3\cdot\sqrt{d}\cdot \ln^2{n}\right),
\end{align*}
which proves the lemma because $d$ is regarded as a constant.
\end{proof}

\paragraph{Pairs of type~0.}
In order to improve upon Theorem~\ref{theorem:L2weak}, we consider 
pairs of linked 2-changes as in the analysis of $\phi$-perturbed 
Manhattan instances. Since our analysis of pairs of linked 2-changes is 
based on the analysis of a single 2-change that we presented in the 
previous section, we also have to consider simplified random 
experiments when analyzing pairs of 2-changes. For a fixed pair of 
type~0, we assume that point $v_3$ is chosen to be the origin and the 
other points $v_1$, $v_2$, $v_4$, $v_5$, and $v_6$ are chosen uniformly 
at random from a hyperball with radius $\sqrt{d}$ centered at $v_3$. 
Let ${\cal E}$ denote the event that both $\Delta_1$ and $\Delta_2$ lie 
in the interval $[0,\varepsilon]$, for some given $\varepsilon$. With 
the same arguments as above, one can see that the probability of ${\cal 
E}$ in the simplified random experiment is smaller compared to the 
original experiment by at most a factor of $((4d)^{d/2}\phi)^5$.
The exponent $5$ is due to the fact that we have now five other points
instead of only three.

\paragraph{Pairs of type~1.}
For a fixed pair of type~1, we consider the simplified random 
experiment in which $v_2$ is placed in the origin and the other points 
$v_1$, $v_3$, $v_4$, and $v_5$ are chosen uniformly at random from a 
hyperball with radius $\sqrt{d}$ centered at $v_2$. In this case, the 
probability in the simplified random experiment is smaller by at most a 
factor of $((4d)^{d/2}\phi)^4$. The exponent $4$ is due to the fact that
we have now four other points.

\subsubsection{Analysis of pairs of linked 2-changes}

Finally, we can prove Lemma~\ref{lemma:L2Improvements}.
\begin{proof}[Lemma~\ref{lemma:L2Improvements}]
We start by considering pairs of type~0. We consider the simplified 
random experiment in which $v_3$ is chosen to be the origin and the 
other points are drawn uniformly at random from a hyperball with radius 
$\sqrt{d}$ centered at $v_3$. If the position of the point $v_1$ is 
fixed, then the events $\Delta_1\in[0,\varepsilon]$ and 
$\Delta_2\in[0,\varepsilon]$ are independent as only the vertices $v_1$ 
and $v_3$ appear in both the first and the second step. In fact, 
because the densities of the points $v_2$, $v_4$, $v_5$, and $v_6$ are 
rotationally symmetric, the concrete position of $v_1$ is not important 
in our simplified random experiment anymore; only the distance $R$ 
between $v_1$ and $v_3$ is of interest.

For $i\in[2]$, we determine the conditional probability of the 
event $\Delta_i\in[0,\varepsilon]$ under the condition that the 
distance $\dist(v_1,v_3)$ is fixed with the help of 
Lemma~\ref{lemma:UB:DensitiesDeltaZ}~a), and obtain
\begin{align}\label{eqn:ProbDeltaI}
        & \Pr{\Delta_i\in[0,\varepsilon]\,|\,\dist(v_1,v_3)=r}\notag
    =   \int_{0}^{\varepsilon}f_{\Delta|R_i=r}(\delta)\,d\delta
    \le \int_{0}^{\varepsilon}\frac{\kappa}{\sqrt{r}}\ln\left(\delta^{-1}\right)\,d\delta\notag\\
    = & \frac{\kappa}{\sqrt{r}}\cdot\left[\delta\left(1+\ln\left(\delta^{-1}\right)\right)\right]_{0}^{\varepsilon}
    = \frac{\kappa}{\sqrt{r}}\cdot\varepsilon\cdot(1+\ln(1/\varepsilon))
    \le \frac{3\kappa}{\sqrt{r}}\cdot\varepsilon\cdot\ln(1/\varepsilon),
\end{align}
where the last inequality follows because, as $\varepsilon\le1/2$,
$1\le 2\ln(1/\varepsilon)$. 
Since for fixed distance $\dist(v_1,v_3)$ the random variables $\Delta_1$ 
and $\Delta_2$ are independent, we obtain
\begin{equation}
\label{eqn:eventE}
  \Pr{\Delta_1,\Delta_2\in[0,\varepsilon]\,|\,\dist(v_1,v_3)=r} \le
  \frac{9\kappa^2}{r}\cdot\varepsilon^2\cdot\ln^2(1/\varepsilon).  
\end{equation}
For $r\in[0,\sqrt{d}]$, the density $f_{\dist(v_1,v_3)}$ of the random 
variable $\dist(v_1,v_3)$ in the simplified random experiment is 
$r^{d-1}/d^{d/2-1}$. In order to see this, remember that $v_3$ is chosen
to be the origin and $v_1$ is chosen uniformly at random from a hyperball
with radius $\sqrt{d}$ centered at the origin. The volume $V_d(r)$ 
of a $d$-dimensional hyperball with radius $r$ is $C_d\cdot r^d$ for some 
constant $C_d$ depending on $d$. Now the density $f_{\dist(v_1,v_3)}$
can be written as 
\[
   f_{\dist(v_1,v_3)}(r) = \frac{\frac{d}{dr}V_d(r)}{V_d(\sqrt{d})}
   = \frac{C_d\cdot d\cdot r^{d-1}}{C_d\cdot d^{d/2}}
   = \frac{r^{d-1}}{d^{d/2-1}}.
\]   
Combining this observation with the bound given in~(\ref{eqn:eventE}) yields
\begin{alignat*}{1}
   \Pr{\Delta_1,\Delta_2\in[0,\varepsilon]} & 
   \le \int_{0}^{\sqrt{d}}\left(\frac{9\kappa^2}{r}\varepsilon^2\ln^2(1/\varepsilon)\right)
   \left(\frac{r^{d-1}}{d^{d/2-1}}\right)\,dr\\
   & = O\left(\varepsilon^2\cdot\ln^2(1/\varepsilon)\right),
\end{alignat*}
where the last equation follows because $d$ is assumed to be a constant.
There are $O(n^6)$ different pairs of type~0; hence a union bound over 
all of them concludes the proof of the first term in the sum in
Lemma~\ref{lemma:L2Improvements} when taking into 
account the factor $((4d)^{d/2}\phi)^5$ that results from considering 
the simplified random experiment (see Section~\ref{subsubsec:Simplified}).

It remains to consider pairs of type~1. We consider the simplified 
random experiment in which $v_2$ is chosen to be the origin and the 
other points are drawn uniformly at random from a hyperball with radius 
$\sqrt{d}$ centered at $v_2$. In contrast to pairs of type $0$, pairs 
of type $1$ exhibit larger dependencies as only $5$ different vertices 
are involved in these pairs. Fix one pair of type $1$. The two 
2-changes share the whole triangle consisting of $v_1$, $v_2$, and 
$v_3$. In the second step, there is only one new vertex, namely $v_5$. 
Hence, there is not enough randomness contained in a pair of type $1$ 
such that $\Delta_1$ and $\Delta_2$ are nearly independent as for pairs 
of type $0$.

We start by considering pairs of type~1~a) as defined in Section~\ref{subsubsec:LinkedPairs}.
First, we analyze the 
probability that $\Delta_1$ lies in the interval $[0,\varepsilon]$. 
After that, we analyze the probability that $\Delta_2$ lies in the 
interval $[0,\varepsilon]$ under the condition that the points $v_1$, 
$v_2$, $v_3$, and $v_4$ have already been chosen. In the analysis of the 
second step we cannot make use of the fact that the distances 
$\dist(v_1,v_3)$ and $\dist(v_2,v_3)$ are random variables anymore 
since we exploited their randomness already in the analysis of the 
first step. The only distances whose randomness we can exploit are the 
distances $\dist(v_1,v_5)$ and $\dist(v_2,v_5)$. We pessimistically 
assume that the distances $\dist(v_1,v_3)$ and $\dist(v_2,v_3)$ have 
been chosen by an adversary. This means the adversary can determine an 
interval of length $\varepsilon$ in which the random variable 
$\dist(v_2,v_5)-\dist(v_1,v_5)$ must lie in order for $\Delta_2$ to lie 
in $[0,\varepsilon]$.

Analogously to~(\ref{eqn:ProbDeltaI}), the probability of the event 
$\Delta_1\in[0,\varepsilon]$ under the condition $\dist(v_1,v_2)=r$ can 
be bounded by
\begin{equation}
  \Pr{\Delta_1\in[0,\varepsilon]\,|\,\dist(v_1,v_2)=r} 
  \le \frac{3\kappa}{\sqrt{r}}\cdot\varepsilon\cdot\ln(1/\varepsilon).
  \label{eqn:CondProbDelta1}
\end{equation}
Due to Lemma~\ref{lemma:UB:DensitiesDeltaZ}~d), the conditional density 
of the random variable $Z=\dist(v_2,v_5)-\dist(v_1,v_5)$ under the 
condition $\dist(v_1,v_2)=r$ can be bounded by
\[
  f_{Z|\dist(v_1,v_2)=r}(z) \le \frac{\kappa}{\sqrt{r^2-z^2}}
\]
for $|z| < r$. Note that Lemma~\ref{lemma:UB:DensitiesDeltaZ}~d) applies
if we set $O=v_2$, $P=v_1$, and $Q_i=v_5$. Then~$T=\dist(O,P)=\dist(v_1,v_2)$. 

This upper bound on the density function $f_{Z|\dist(v_1,v_2)=r}(z)$ is symmetric around zero, it is monotonically increasing
for $z\in[0,r)$, and it is monotonically decreasing in $(-r,0)$. This implies that
the intervals the adversary can specify that have the highest upper bound on the
probability of $Z$ falling into them are $[-r,-r+\varepsilon]$ and 
$[r-\varepsilon,r]$.
Hence, the conditional probability of the event 
$\Delta_2\in[0,\varepsilon]$ under the condition $\dist(v_1,v_2)=r$ and 
for fixed points $v_3$ and $v_4$ is bounded from above by
\[
   \int_{\max\{r-\varepsilon,-r\}}^{r} \frac{\kappa}{\sqrt{r^2-z^2}}\, dz,
\]
where the lower bound in the integral follows because $Z$ can only take values in $[-r,r]$.
This can be rewritten as
\[
   \kappa\cdot \int_{\max\{r-\varepsilon,-r\}}^{r} \frac{1}{\sqrt{r+|z|}}\cdot\frac{1}{\sqrt{r-|z|}}\, dz
   \le \frac{\kappa}{\sqrt{r}}\cdot \int_{\max\{r-\varepsilon,-r\}}^{r} \frac{1}{\sqrt{r-|z|}}\, dz.
\]
For $\varepsilon\le r$, we have $r-\varepsilon\ge 0\ge -r$ and hence,
\[
   \frac{\kappa}{\sqrt{r}}\cdot \int_{\max\{r-\varepsilon,-r\}}^{r} \frac{1}{\sqrt{r-|z|}}\, dz
   = \frac{\kappa}{\sqrt{r}}\cdot \int_{r-\varepsilon}^{r} \frac{1}{\sqrt{r-z}}\, dz
   = \frac{2\kappa \sqrt{\varepsilon}}{\sqrt{r}}
   \le \frac{4\kappa \sqrt{\varepsilon}}{\sqrt{r}}.
\]
For $\varepsilon\in(r,2r]$, we have $0\ge r-\varepsilon\ge -r$ and hence,
\begin{align*}
   \frac{\kappa}{\sqrt{r}}\cdot \int_{\max\{r-\varepsilon,-r\}}^{r} \frac{1}{\sqrt{r-|z|}}\, dz
   & = \frac{\kappa}{\sqrt{r}}\cdot \left(\int_{0}^{r} \frac{1}{\sqrt{r-z}}\, dz+\int_{r-\varepsilon}^{0} \frac{1}{\sqrt{r+z}}\, dz\right)\\
   & \le \frac{\kappa}{\sqrt{r}}\cdot \left(2\sqrt{r}+\int_{-r}^{0} \frac{1}{\sqrt{r+z}}\, dz\right)\\
   & = \frac{4\kappa \sqrt{r}}{\sqrt{r}} \le \frac{4\kappa \sqrt{\varepsilon}}{\sqrt{r}},
\end{align*}
where we used $\varepsilon > r$ for the last inequality.
For $\varepsilon> 2r$, we have $r-\varepsilon \le -r$ and hence,
\[
   \frac{\kappa}{\sqrt{r}}\cdot \int_{\max\{r-\varepsilon,-r\}}^{r} \frac{1}{\sqrt{r-|z|}}\, dz
   = \frac{\kappa}{\sqrt{r}}\cdot \int_{-r}^{r} \frac{1}{\sqrt{r-|z|}}\, dz
   = 2\kappa 
   \le \frac{2\kappa \sqrt{\varepsilon}}{\sqrt{r}}
   \le \frac{4\kappa \sqrt{\varepsilon}}{\sqrt{r}},
\]
where we used $\varepsilon > r$ for the penultimate inequality.
Altogether this argument shows that
\begin{equation}\label{eqn:CondProbDelta2}
  \Pr{\Delta_2\in[0,\varepsilon]\,|\,\text{$v_1,v_2,v_3,v_4$ fixed arbitrarily with $\dist(v_1,v_2)=r$}}  \le \frac{4\kappa \sqrt{\varepsilon}}{\sqrt{r}}.
\end{equation}
Since~\eqref{eqn:CondProbDelta2} uses only the randomness of~$v_5$ which is independent of~$\Delta_1$,
we can multiply the upper bounds from~(\ref{eqn:CondProbDelta1}) and~\eqref{eqn:CondProbDelta2} to obtain
\[
  \Pr{\Delta_1,\Delta_2\in[0,\varepsilon]\,|\,\dist(v_1,v_2)=r} \le \frac{12\kappa^2}{r}\varepsilon^{3/2}\cdot\ln(1/\varepsilon).
\]
In order to get rid of the condition $\dist(v_1,v_2)=r$, we integrate over 
all possible values the random variable $\dist(v_1,v_2)$ can take, yielding
\begin{align*}
  \Pr{\Delta_1,\Delta_2\in[0,\varepsilon]} 
  & = \int_{0}^{\sqrt{d}} \frac{r^{d-1}}{d^{d/2-1}}\cdot \Pr{\Delta_1,\Delta_2\in[0,\varepsilon]\,|\,\dist(v_1,v_2)=r}\,dr\\
  & \le \int_{0}^{\sqrt{d}}\frac{12\kappa^2\cdot r^{d-2}}{d^{d/2-1}}\cdot
  \varepsilon^{3/2}\cdot\ln(1/\varepsilon)\, dr 
  = O\left(\varepsilon^{3/2}\cdot\ln(1/\varepsilon)\right),
\end{align*}
where the last equation follows because $d$ is assumed to be constant.
Applying a union bound over all $O(n^5)$ possible pairs of type~1~a) 
concludes the proof when one takes into account the factor 
$((4d)^{d/2}\phi)^4$ due to considering the simplified random 
experiment (see Section~\ref{subsubsec:Simplified}).

For pairs of type~1~b), the situation looks somewhat similar. We 
analyze the first step and in the second step, we can only exploit the 
randomness of the distances $\dist(v_2,v_5)$ and $\dist(v_3,v_5)$. Due 
to Lemma~\ref{lemma:UB:DensitiesDeltaZ}~b) and similarly
to~(\ref{eqn:ProbDeltaI}), the probability of the event 
$\Delta_1\in[0,\varepsilon]$ under the condition $\dist(v_2,v_3)=\tau$ 
can be bounded by
\begin{equation}
  \Pr{\Delta_1\in[0,\varepsilon]\,|\,\dist(v_2,v_3)=\tau} 
  \le \frac{3\kappa}{\tau}\cdot\varepsilon\cdot\ln(1/\varepsilon).
  \label{eqn:CondProbDelta1_2}
\end{equation}
The remaining analysis of pairs of type~1~b) can be carried out 
completely analogously to the analysis of pairs of type~1~a).
\end{proof}

\subsubsection{Expected number of 2-changes}

Based on Lemmas~\ref{lemma:numberLinkedPairs2} 
and~\ref{lemma:L2Improvements}, we are now able to prove part b) of 
Theorem~\ref{theorem:runningTime1}, which states that
the expected length of the longest path in the 2-Opt state graph
is $O(n^{4+1/3}\cdot\log( n\phi)\cdot\phi^{8/3})$ for 
$\phi$-perturbed Euclidean instances with $n$ points.

\begin{proof}[Theorem~\ref{theorem:runningTime1} b)]
We use the same notation as in the proof of part a) of the theorem. 
For $t>n$, we have $t/7-3n/28>t/28$ and 
hence using Lemma~\ref{lemma:L2Improvements} with $\varepsilon=\frac{48\sqrt{d}n}{t}$
yields
\begin{equation*}
\begin{split}
  & \Pr{T\ge t} 
  \le \Pr{\Delta_{\min}^*\le \frac{28\sqrt{d}n}{t}} \\
  & =  O\left(\min\left\{\frac{n^8\cdot\log^2(t)\cdot\phi^5}{t^2},1\right\}\right) + O\left(\min\left\{\frac{n^{13/2}\cdot\log(t)\cdot\phi^4}{t^{3/2}},1\right\}\right).
\end{split}
\end{equation*}
This implies that the expected length 
of the longest path in the state graph is bounded from above by
\begin{equation}\label{eqn:ExpectedNum2Changes}
\begin{split}
  n+\sum_{t=n+1}^{n!}\left(O\left(\min\left\{\frac{n^8\cdot\log^2(t)\cdot\phi^5}{t^2},1\right\}\right)\right.\\ + 
  \left.O\left(\min\left\{\frac{n^{13/2}\cdot\log(t)\cdot\phi^4}{t^{3/2}},1\right\}\right)\right).
  \end{split}
\end{equation}

In the following, we use the fact that, for $a>0$,
\[
   \int_a^{\infty} \frac{\ln^2(x)}{x^2}\,dx = \left[-\frac{\ln^2(x)+2\ln(x)+2}{x}\right]_a^{\infty}
   = O\left(\frac{\ln^2(a)}{a}\right).
\]
For $t_A=n^4\cdot\log(n\phi)\cdot\phi^{5/2}$, the first sum in~\eqref{eqn:ExpectedNum2Changes} can be bounded as follows:
\begin{align*}
  & \,\sum_{t=n+1}^{n!} O\left(\min\left\{\frac{n^8\cdot\log^2(t)\cdot\phi^5}{t^2},1\right\}\right)\\
  \le & \, t_A+O\left(\int_{t=t_A}^{\infty} \frac{n^8\cdot\log^2(t)\cdot\phi^5}{t^2}\,dt\right)
  = t_A+O\left(\left[-\frac{n^8\cdot\log^2(t)\cdot\phi^5}{t}\right]_{t=t_A}^{\infty}\right)\\
  = & \, t_A+O\left(\frac{n^8\cdot\log^2(t_A)\cdot\phi^5}{t_A}\right)
  = t_A+O\left(\frac{n^8\cdot\log^2(n\phi)\cdot\phi^5}{t_A}\right)
  = O(t_A).
\end{align*}
In the following, we use the fact that, for $a>0$,
\[
   \int_a^{\infty} \frac{\ln(x)}{x^{3/2}}\,dx = \left[-\frac{2\ln(x)+4}{\sqrt{x}}\right]_a^{\infty}
   = O\left(\frac{\ln(a)}{\sqrt{a}}\right).
\]
For $t_B=n^{13/3}\cdot\log^{2/3}(n\phi)\cdot \phi^{8/3}$, the second sum in~\eqref{eqn:ExpectedNum2Changes} can be bounded
as follows: 
\begin{align*}
  & \,\sum_{t=n+1}^{n!} O\left(\min\left\{\frac{n^{13/2}\cdot\log(t)\cdot\phi^4}{t^{3/2}},1\right\}\right) \\
  \le & \, t_B+O\left(\int_{t=t_B}^{\infty} \frac{n^{13/2}\cdot\log(t)\cdot\phi^4}{t^{3/2}}\,dt\right)\\
  = &\, t_B+O\left(\left[-\frac{n^{13/2}\cdot\log(t)\cdot\phi^4}{\sqrt{t}}\right]_{t=t_B}^{\infty}\right)\\
  = &\,t_B+O\left(\frac{n^{13/2}\cdot\log(t_B)\cdot\phi^4}{\sqrt{t_B}}\right)
  = t_B+O\left(\frac{n^{13/2}\cdot\log(n\phi)\cdot\phi^4}{\sqrt{t_B}}\right)
  = O(t_B).
\end{align*}

Together this yields
\[
  \Ex{T} = O\left(n^4\cdot\log(n\phi)\cdot\phi^{5/2}\right)+ O\left(n^{13/3}\cdot\log^{2/3}(n\phi)\cdot \phi^{8/3}\right),
\]
which concludes the proof of part b) of the theorem.
\end{proof}

Using the same observations as in the proof of 
Theorem~\ref{theorem:runningTime2}~a) also yields part b):
\begin{proof}[Theorem~\ref{theorem:runningTime2} b)]
Estimating the length of the initial tour by 
$O(n^{(d-1)/d}\cdot\log{n})$ instead of $O(n)$ improves the upper bound 
on the expected number of 2-changes by a factor of 
$\Theta(n^{1/d}/\log{n})$ compared to 
Theorem~\ref{theorem:runningTime1}~b). This observation yields the 
bound claimed in Theorem~\ref{theorem:runningTime2}~b).
\end{proof}

\section{Expected Approximation Ratio}

\label{sec:approximation}

In this section, we consider the expected approximation ratio of the solution
found by 2-Opt on $\phi$-perturbed $L_p$ instances. Chandra, Karloff, and
Tovey~\cite{ChandraKT99} show that if one has a set of $n$ points in the unit
hypercube $[0,1]^d$ and the distances are measured according to a metric that is
induced by a norm, then every locally optimal solution has length at most $c\cdot
n^{(d-1)/d}$ for an appropriate constant $c$ depending on the dimension $d$ and
the metric. Hence, it follows for every $L_p$ metric that 2-Opt yields a tour of
length $O(n^{(d-1)/d})$ on $\phi$-perturbed $L_p$ instances.
This implies that the approximation ratio of 2-Opt on these instances can be bounded
from above by $O(n^{(d-1)/d})/\OPT$, where $\OPT$ denotes the length
of the shortest tour. We will show a lower bound on  
$\OPT$ that holds with high probability in $\phi$-perturbed $L_p$ instances.
Based on this, we prove Theorem~\ref{theorem:approximation}.

\begin{proof}[Theorem~\ref{theorem:approximation}]
Let $v_1,\ldots,v_n\in\RR^d$ denote the points of the $\phi$-perturbed instance.
We denote by $k$ the largest integer $k\le n\phi$ that can be written
as $k=\ell^d$ for some $\ell\in\NN$. We partition the unit hypercube into $k$
smaller hypercubes with volume $1/k$ each and analyze how many of these smaller
hypercubes contain at least one of the points. Assume that $X>3^d$ of these
hypercubes contain a point; then the optimal tour must have length at least
\begin{equation}\label{eqn:LengthOptimal}
 \left\lceil \frac{X}{3^d} \right\rceil\cdot \frac{1}{\sqrt[d]{k}}.
\end{equation}
In order to see this, we construct a set
$P\subseteq\{v_1,\ldots,v_n\}$ of points as follows: Consider the points
$v_1,\ldots,v_n$ one after another, and insert a point $v_i$ into $P$ if $P$ does
not contain a point in the same hypercube as $v_i$ or in one of its $3^d-1$
neighboring hypercubes yet. Due to the triangle inequality, the optimal tour on
$P$ is at most as long as the optimal tour on $v_1,\ldots,v_n$. Furthermore, $P$
contains at least $\left\lceil X/3^d \right\rceil\ge 2$ points and every edge between two points from $P$ has
length at least $1/\sqrt[d]{k}$ since $P$ does not contain two points in the same
or in two neighboring hypercubes.
Hence, it remains to analyze the random variable $X$. For each 
hypercube $i$ with $1\le i\le k$, we define a random variable $X_i$ 
which takes value $0$ if hypercube $i$ is empty and value $1$ if 
hypercube $i$ contains at least one point. The density functions that 
specify the locations of the points induce for each pair of hypercube 
$i$ and point $j$ a probability $p_i^j$ such that point $j$ falls into 
hypercube $i$ with probability $p_i^j$. Hence, one can think of 
throwing $n$ balls into $k$ bins in a setting where each ball has its 
own probability distribution over the bins. Due to the bounded density, 
we have $p_i^j\le \phi/k$. For each hypercube $i$, let $M_i$ denote the 
probability mass associated with hypercube~$i$, that is
\[
   M_i = \sum_{j=1}^np_i^j \le \frac{n\phi}{k}.
\]
We can write the expected value of the random variable $X_i$ as 
\[
   \Ex{X_i} = \Pr{X_i=1} = 1-\prod_{j=1}^n(1-p_i^j)
   \ge 1-\left(1-\frac{M_i}{n}\right)^n
\]
as, under the constraint $\sum_j(1-p_i^j)=n-M_i$, the term
$\prod_j(1-p_i^j)$ is maximized if all $p_i^j$ are equal.
Due to linearity of expectation, the expected value of $X$ is
\[
  \Ex{X} \ge \sum_{i=1}^k \left(1-\left(1-\frac{M_i}{n}\right)^n\right)
         = k - \sum_{i=1}^k\left(1-\frac{M_i}{n}\right)^n.
\]
Observe that $\sum_i M_i=n$ and hence, also the sum 
$\sum_{i}\left(1-M_i/n\right)=k-1$ is fixed. As the function $f(x)=x^n$
is convex for $n\ge 1$, the sum $\sum_i(1-M_i/n)^n$
becomes maximal if the $M_i$'s are chosen as unbalanced as
possible. Hence, we assume that $\lceil k/\phi\rceil$ of the $M_i$'s
take their maximal value of $n\phi/k$ and the other $M_i$'s are zero.
This yields, for sufficiently large $n$,
\begin{align*}
  \Ex{X} \ge & k -\left(
  \left\lceil\frac{k}{\phi}\right\rceil\left(1-\frac{\phi}{k}\right)^n 
  + \left(k-\left\lceil\frac{k}{\phi}\right\rceil\right)\right) \\
  = & \left\lceil\frac{k}{\phi}\right\rceil - \left\lceil\frac{k}{\phi}\right\rceil \cdot \left(1-\frac{\phi}{k}\right)^n\\
  \ge & \frac{k}{\phi}-\frac{2k}{\phi}\left(1-\frac{\phi}{k}\right)^n\\
  \ge & \frac{k}{\phi}\left(1-2\left(1-\frac{1}{n}\right)^n\right)
  \ge \frac{k}{\phi}\left(1-\frac{2}{e}\right)
  \ge \frac{k}{4\phi}. 
\end{align*}
For the second inequality we have used that~$\frac{k}{\phi}\ge 1$ for sufficiently large~$n$ and hence~$\left\lceil\frac{k}{\phi}\right\rceil\le \frac{2k}{\phi}$.
For the third inequality we have used that~$n\ge \frac{k}{\phi}$, which follows from
the definition of~$k$ as the largest integer $k\le n\phi$ that can be written
as $k=\ell^d$ for some $\ell\in\NN$. This definition also implies 
\[
  n\phi < (\ell+1)^d =(\sqrt[d]{k}+1)^d \le (2\sqrt[d]{k})^d = 2^dk 
\]
and hence, $\Ex{X}\ge n/2^{d+2}$.

Next we show that $X$ is sharply concentrated 
around its mean value. The random variable $X$ is the sum of $k$ 
0-1-random variables $X_i$. If these random variables were independent, we 
could simply use a Chernoff bound to bound the probability that $X$ 
takes a value that is much smaller than its mean value. Intuitively,
whenever we already know that 
some of the $X_i$'s are zero, then the probability of the event that 
another $X_i$ also takes the value zero becomes smaller. Hence, 
intuitively, the dependencies can only help to bound the probability 
that $X$ takes a value smaller than its mean value.

To formalize this intuition, we use the framework of \emph{negatively associated} random variables,
introduced by Dubhashi and Ranjan~\cite{DubhashiR98}. In Appendix~\ref{appendix:NegativelyAssociated},
we repeat the formal definition and we show that the $X_i$ are negatively associated. 
Dubhashi and Ranjan show (Proposition 7 of~\cite{DubhashiR98}) that in the case of 
negatively associated random variables, one can still apply a Chernoff 
bound.
The Chernoff bound from~\cite{MotwaniR95} implies that, for any~$\delta\in(0,1)$,
\[
   \Pr{X \le (1-\delta)\cdot\Ex{X}} \le \exp\left(-\frac{\Ex{X}\cdot\delta^2}{2}\right).
\]
This yields
\begin{equation}\label{eqn:ApplicationChernoffBound}
  \Pr{X\le \frac{n}{2^{d+3}}} \le 
  \Pr{X\le \frac{\Ex{X}}{2}} \le
  \exp\left(-\frac{\Ex{X}}{8}\right)
  \le
  \exp\left(-\frac{n}{2^{d+5}}\right),
\end{equation}
where we used $\Ex{X}\ge n/2^{d+2}$ for the first and last inequality.

In order to bound the expected approximation ratio of any locally optimal
solution, we distinguish between two cases:
\begin{itemize} 
  \item If $X\ge \frac{n}{2^{d+3}}$, then, assuming that $n$ is large enough,
        we have that $X>3^d$ and hence, \eqref{eqn:LengthOptimal} implies that
        \[
           \OPT \ge \left\lceil \frac{X}{3^d} \right\rceil\cdot \frac{1}{\sqrt[d]{k}} \ge
           \frac{X}{3^d\sqrt[d]{k}} \ge \frac{n}{2^{d+3}3^d\sqrt[d]{k}}
           = \Theta\left(\frac{n^{(d-1)/d}}{\sqrt[d]{\phi}}\right),
        \]
        where we used that~$k=\Theta(n\phi)$ for the last equation.
        Combining this with Chandra, Karloff, and Tovey's~\cite{ChandraKT99} result that
		every locally optimal solution has length at most $O(n^{(d-1)/d})$ yields an
		approximation ratio of
		\[
		   \frac{O(n^{(d-1)/d})}{\Theta\left(\frac{n^{(d-1)/d}}{\sqrt[d]{\phi}}\right)}
		   = O(\sqrt[d]{\phi}).
		\]
\item If $X < \frac{n}{2^{d+3}}$, then we use $n$ as an upper bound on the approximation
      ratio of any locally optimal solution. This bound holds in fact for any possible tour,
      as the following argument shows:
      The length of every tour is bounded from above by $n$ times the length $\alpha$ of the longest edge.
      Let $u$ and $v$ be the vertices that this edge connects. Then every tour has to contain
      a path between $u$ and $v$. Due to the triangle inequality, this path must have
      length at least $\alpha$.
      
      We have seen in~\eqref{eqn:ApplicationChernoffBound} that the event $X < \frac{n}{2^{d+3}}$ occurs only with exponentially
      small probability. This implies that it adds at most
      \[
         \exp\left(-\frac{n}{2^{d+5}}\right)\cdot n = o(1)
      \]    
      to the expected approximation ratio.
\end{itemize}
This concludes the proof as the contribution of both cases to the expected approximation
ratio is $O(\sqrt[d]{\phi})$.
\end{proof}

\section{Smoothed Analysis}
\label{sec:SmoothedAnalysis}

Smoothed Analysis was introduced by Spielman and 
Teng~\cite{SpielmanT04} as a hybrid of worst case and average case 
analysis. The semi-random input model in a smoothed analysis is 
designed to capture the behavior of algorithms on typical inputs better 
than a worst case or average case analysis alone as it allows an 
adversary to specify an arbitrary input which is randomly perturbed 
afterwards. In Spielman and Teng's analysis of the Simplex algorithm 
the adversary specifies an arbitrary linear program which is perturbed 
by adding independent Gaussian random variables to each number in the 
linear program. Our probabilistic analysis of Manhattan and Euclidean 
instances can also be seen as a smoothed analysis in which an adversary 
can choose the distributions for the points over the unit hypercube. 
The adversary is restricted to distributions that can be represented by 
densities that are bounded by $\phi$. Our model cannot handle Gaussian 
perturbations directly because the support of Gaussian random variables 
is not bounded.

Assume that every point $v_1,\ldots,v_n$ is described by a density 
whose support is restricted to the hypercube $[-\alpha,1+\alpha]^d$, for 
some $\alpha\ge1$. Then after appropriate scaling and translating, we 
can assume that all supports are restricted to the unit hypercube 
$[0,1]^d$. Thereby, the maximal density $\phi$ increases by at most a 
factor of $(2\alpha+1)^d$. Hence, after appropriate scaling and 
translating, Theorems~\ref{theorem:runningTime1}, 
\ref{theorem:runningTime2}, and~\ref{theorem:approximation} can still 
be applied if one takes into account the increased densities.

One possibility to cope with Gaussian perturbations is to consider 
\emph{truncated Gaussian perturbations}. In such a perturbation model, 
the coordinates of each point are initially chosen from $[0,1]^d$ and 
then perturbed by adding Gaussian random variables with mean $0$ and with some standard 
deviation $\sigma$ to them that are conditioned to lie in 
$[-\alpha,\alpha]$ for some $\alpha\ge1$. The maximal density of such 
truncated Gaussian random variables for $\sigma\le1$ is bounded from 
above by
\begin{equation}\label{eqn:TruncatedGaussian}
   \frac{1/(\sigma\sqrt{2\pi})}{1-\sigma\cdot\exp(-\alpha^2/(2\sigma^2))}.
\end{equation}
This is shown by the following calculation in which we denote 
by $X$ a Gaussian random variable with mean $0$ and standard deviation $\sigma$,
by $f(z)=\exp(-z^2/(2\sigma^2))/(\sigma\sqrt{2\pi})$ its density function and by
$f_{X|X\in[-\alpha,\alpha]}$ the density of $X$ conditioned on the fact that $X\in[-\alpha,\alpha]$:  
\begin{align*}
  f_{X|X\in[-\alpha,\alpha]}(z) & \le \frac{f(z)}{\Pr{X\in [-\alpha,\alpha]}} = 
  \frac{\exp(-z^2/(2\sigma^2))}{\sigma \sqrt{2\pi}\cdot \Pr{X\in [-\alpha,\alpha]}}\\
  & \le \frac{1/(\sigma\sqrt{2\pi})}{\Pr{X\in [-\alpha,\alpha]}}
  = \frac{1/(\sigma\sqrt{2\pi})}{1-\Pr{X\notin [-\alpha,\alpha]}}\\
  & \le \frac{1/(\sigma\sqrt{2\pi})}{1-\sigma\cdot\exp(-a^2/(2\sigma^2))},  
\end{align*}
where we used the following bound on the probability that $X$ does not lie in $[-\alpha,\alpha]$:
\begin{align}
  \Pr{X\notin [-\alpha,\alpha]} & = \int_{\alpha}^{\infty} f(z)\,dz + \int_{-\infty}^{-\alpha} f(z)\,dz\notag\\
  & = 2\int_{\alpha}^{\infty} f(z)\,dz = \frac{\sqrt{2}}{\sigma\sqrt{\pi}}\int_{\alpha}^{\infty} \exp(-z^2/(2\sigma^2))\,dz\notag\\
  & \le \frac{\sqrt{2}}{\sigma\sqrt{\pi}}\int_{\alpha}^{\infty} z\cdot \exp(-z^2/(2\sigma^2))\,dz\notag\\
  & = \frac{\sqrt{2}}{\sigma\sqrt{\pi}}\left[-\sigma^2\exp(-z^2/(2\sigma^2))\right]_{\alpha}^{\infty}\notag\\
  & = \frac{\sigma\sqrt{2}}{\sqrt{\pi}}\exp(-\alpha^2/(2\sigma^2))
  \le \sigma\cdot\exp(-\alpha^2/(2\sigma^2)),\label{eqn:estimateGaussian}
\end{align}
where the inequality follows from $\alpha\ge 1$.

After such a truncated perturbation, all points lie in the hypercube 
$[-\alpha,1+\alpha]^d$. Hence, one can apply 
Theorems~\ref{theorem:runningTime1},~\ref{theorem:runningTime2}, 
and~\ref{theorem:approximation} with
\begin{equation*}
  \phi =
  \frac{(2\alpha+1)^d}{(\sigma\sqrt{2\pi}-\sigma^2\sqrt{2\pi}\exp(-\alpha^2/(2\sigma^2)))^d}  
  = O\left(\frac{\alpha^d}{\sigma^d}\right),
\end{equation*}
where the first equality follows from~\eqref{eqn:TruncatedGaussian} and the observation
that shifting and scaling the hypercube $[-\alpha,1+\alpha]^d$ to $[0,1]^d$ leads to
densities that are larger than the original densities by at most a factor of $(2\alpha+1)^d$.
The second equality follows because the term $\sigma^2\sqrt{2\pi}\exp(-\alpha^2/(2\sigma^2))$
is in $o(\sigma)$ if $\sigma$ goes to $0$.

It is not necessary to truncate the Gaussian random variables if the 
standard deviation is small enough. For $\sigma\le 
\min\{\alpha/\sqrt{2(n+1)\ln{n}+2\ln{d}},1\}$, the probability that one 
of the Gaussian random variables has an absolute value larger than 
$\alpha\ge 1$ is bounded from above by $n^{-n}$. This follows from a union
bound over all $dn$ Gaussian variables and~\eqref{eqn:estimateGaussian}:
\begin{align*}
  & dn \cdot  \Pr{X\notin[-\alpha,\alpha]}
  \le \exp(\ln(dn)) \left(\sigma\cdot\exp(-\alpha^2/(2\sigma^2))\right)\\
  \le & \exp(\ln(dn)-\alpha^2/(2\sigma^2))
  \le \exp(\ln(dn)-(n+1)\ln{n}-\ln{d}) = n^{-n}.
\end{align*}
We have used~$\sigma\le 1$ for the second inequality.
In this case, even if 
one does not truncate the random variables, 
Theorems~\ref{theorem:runningTime1},~\ref{theorem:runningTime2}, 
and~\ref{theorem:approximation} can be applied with 
$\phi=O(\alpha^d/\sigma^d)$. To see this, it suffices to observe that 
the worst-case bound for the number of 2-changes is $n!$ and the 
worst-case approximation ratio is $O(\log{n})$~\cite{ChandraKT99}. 
Multiplying these values with the failure probability of $n^{-n}$ 
adds less than 1 to the expected values. In particular, this 
implies that the expected length of the longest path in the state graph 
is bounded by $O(\mathrm{poly}(n,1/\sigma))$.

\section{Conclusions and Open Problems}

We have shown several new results on the running time and the 
approximation ratio of the 2-Opt heuristic. However, there are still a 
variety of open problems regarding this algorithm. Our lower bounds 
only show that there exist families of instances on which 2-Opt takes 
an exponential number of steps if it uses a particular pivot rule. It 
would be interesting to analyze the diameter of the state graph and to 
either present instances on which every pivot rule needs an exponential 
number of steps or to prove that there is always an improvement 
sequence of polynomial length to a locally optimal solution. Also the 
worst number of local improvements for some natural pivot rules like, 
\eg, the one that always makes the largest possible improvement or the 
one that always chooses a random improving 2-change, is not known yet. 
Furthermore, the complexity of computing locally optimal solutions is 
open. The only result in this regard is due to Krentel~\cite{Krentel89} 
who shows that it is {\sf PLS}-complete to compute a local optimum for 
the metric TSP for $k$-Opt for some constant $k$. It is not known 
whether his construction can be embedded into the Euclidean metric and 
whether it is {\sf PLS}-complete to compute locally optimal solutions 
for $2$-Opt. Fischer and Torenvliet~\cite{FischerT95} show, however, 
that for the general TSP, it is {\sf PSPACE}-hard to compute a local 
optimum for $2$-Opt that is reachable from a given initial tour.

The obvious open question concerning the probabilistic analysis is 
how the gap between experiments and theory can be narrowed further. 
In order to tackle this question, new methods seem to be necessary. Our 
approach, which is solely based on analyzing the smallest improvement 
made by a sequence of linked 2-changes, seems to yield too pessimistic 
bounds. Another interesting area to explore is the expected 
approximation ratio of 2-Opt. In experiments, approximation ratios 
close to $1$ are observed. For instances that are chosen uniformly at 
random, the bound on the expected approximation ratio is a constant 
but unfortunately a large one. It seems to be a very challenging 
problem to improve this constant to a value that matches the 
experimental results.

Besides 2-Opt, there are also other local search algorithms that are
successful for the traveling salesperson problem. In particular, the
Lin-Kernighan heuristic~\cite{LinK73} is one of the most successful local
search algorithm for the symmetric TSP. It is a variant of $k$-Opt in
which $k$ is not fixed and it can roughly be described as follows: Each
local modification starts by removing one edge $\{a,b\}$ from the current
tour, which results in a Hamiltonian path with the two endpoints $a$ and
$b$. Then an edge $\{b,c\}$ is added, which forms a cycle; there is a
unique edge $\{c,d\}$ incident to $c$ whose removal breaks the cycle,
producing a new Hamiltonian path with endpoints $a$ and $d$. This
operation is called a rotation. Now either a new Hamiltonian cycle can be
obtained by adding the edge $\{a,d\}$ to the tour or another rotation can
be performed. There are a lot of different variants and heuristic
improvements of this basic scheme, but little is known theoretically.
Papadimitriou~\cite{Papadimitriou92} shows for a variant of the
Lin-Kernighan heuristic that computing a local optimum is {\sf
PLS}-complete, which is a sharp contrast to the experimental results.
Since the Lin-Kernighan heuristic is widely used in practice, a
theoretical explanation for its good behavior in practice is of great
interest. Our analysis of 2-Opt relies crucially on the fact that there
are only a polynomial number of different 2-changes. For the Lin-Kernighan
heuristic, however, the number of different local improvements is
exponential. Hence, it is an interesting question as to whether nonetheless the
smallest possible improvement is polynomially large or whether different
methods yield a polynomial upper bound on the expected running time of
the Lin-Kernighan heuristic.

\begin{appendix}

\section{Inequalities from Section~\ref{subsubsec:LB:LpPoints}}
\label{app:Inequalities}

Inequalities corresponding to the improvements made by the 2-changes 
in the sequence in which $G_{n-1}^P$ changes its state from $(S,L)$ to $(S,S)$
while resetting $G_{n-1}^R$:

\noindent{\bf Inequality 1:}
\begin{alignat*}{5}
\,\,\, & \sqrt[p]{9.7^p+3.6^p} &&+ \sqrt[p]{4.3^p+6.9^p} &&- \sqrt[p]{0.3^p+1.7^p} &&- \sqrt[p]{14.3^p+1.6^p} &&\, > 0\\
\end{alignat*}
For $p\ge 3$, we obtain
\begin{equation*}\textstyle
  \sqrt[p]{0.3^p+1.7^p}
   = 1.7\cdot\sqrt[p]{1+\left(\frac{0.3}{1.7}\right)^p} 
  \le 1.7\cdot\sqrt[3]{1+\left(\frac{0.3}{1.7}\right)^{3}}
  < 1.71
\end{equation*}
and
\begin{equation*}\textstyle
  \sqrt[p]{14.3^p+1.6^p}
   = 14.3\cdot\sqrt[p]{1+\left(\frac{1.6}{14.3}\right)^p}
  \le 14.3\cdot\sqrt[3]{1+\left(\frac{1.6}{14.3}\right)^{3}}
  < 14.31.
\end{equation*}
Hence, for $p\ge 3$,
\[
\sqrt[p]{9.7^p+3.6^p} + \sqrt[p]{4.3^p+6.9^p} - \sqrt[p]{0.3^p+1.7^p} - \sqrt[p]{14.3^p+1.6^p}
\ge 9.7+6.9-1.71-14.31 > 0. 
\]

\noindent{\bf Inequality 2:}
\begin{alignat*}{5}
\,\,\, & \sqrt[p]{0.0^p+1.0^p} &&+ \sqrt[p]{8.7^p+14.3^p} &&- \sqrt[p]{1.5^p+7.1^p} &&- \sqrt[p]{7.2^p+6.2^p} &&\, > 0\\
\end{alignat*}
For $p\ge 4$, we obtain
\begin{equation*}\textstyle
  \sqrt[p]{1.5^p+7.1^p}
   = 7.1\cdot\sqrt[p]{1+\left(\frac{1.5}{7.1}\right)^p} 
  \le 7.1\cdot\sqrt[4]{1+\left(\frac{1.5}{7.1}\right)^{4}}
  < 7.11
\end{equation*}
and
\begin{equation*}\textstyle
  \sqrt[p]{7.2^p+6.2^p}
   = 7.2\cdot\sqrt[p]{1+\left(\frac{6.2}{7.2}\right)^p}
  \le 7.2\cdot\sqrt[4]{1+\left(\frac{6.2}{7.2}\right)^{4}}
  < 8.04.
\end{equation*}
Hence, for $p\ge 4$,
\[
\sqrt[p]{0.0^p+1.0^p} + \sqrt[p]{8.7^p+14.3^p} - \sqrt[p]{1.5^p+7.1^p} - \sqrt[p]{7.2^p+6.2^p}
\ge 1.0+14.3-7.11-8.04 > 0. 
\]
For the remaining case $p=3$, the inequality can simply be checked by plugging in the appropriate values.

\noindent{\bf Inequality 3:}
\begin{alignat*}{5}
\,\,\, & \sqrt[p]{1.5^p+7.1^p} &&+ \sqrt[p]{4.3^p+6.9^p} &&- \sqrt[p]{3.5^p+3.7^p} &&- \sqrt[p]{9.3^p+3.9^p} &&\, > 0\\
\end{alignat*}
For $p\ge 4$, we obtain
\begin{equation*}\textstyle
  \sqrt[p]{3.5^p+3.7^p}
   = 3.7\cdot\sqrt[p]{1+\left(\frac{3.5}{3.7}\right)^p} 
  \le 3.7\cdot\sqrt[4]{1+\left(\frac{3.5}{3.7}\right)^{4}}
  < 4.29
\end{equation*}
and
\begin{equation*}\textstyle
  \sqrt[p]{9.3^p+3.9^p}
   = 9.3\cdot\sqrt[p]{1+\left(\frac{3.9}{9.3}\right)^p}
  \le 9.3\cdot\sqrt[4]{1+\left(\frac{3.9}{9.3}\right)^{4}}
  < 9.38.
\end{equation*}
Hence, for $p\ge 4$,
\[
\sqrt[p]{1.5^p+7.1^p} + \sqrt[p]{4.3^p+6.9^p} - \sqrt[p]{3.5^p+3.7^p} - \sqrt[p]{9.3^p+3.9^p}
\ge 7.1+6.9-4.29-9.38 > 0. 
\]
For the remaining case $p=3$, the inequality can simply be checked by plugging in the appropriate values.

\noindent{\bf Inequality 4:}
\begin{alignat*}{5}
\,\,\, & \sqrt[p]{0.0^p+1.0^p} &&+ \sqrt[p]{14.3^p+1.6^p} &&- \sqrt[p]{6.5^p+1.6^p} &&- \sqrt[p]{7.8^p+4.2^p} &&\, > 0\\
\end{alignat*}
For $p\ge 3$, we obtain
\begin{equation*}\textstyle
  \sqrt[p]{6.5^p+1.6^p}
   = 6.5\cdot\sqrt[p]{1+\left(\frac{1.6}{6.5}\right)^p} 
  \le 6.5\cdot\sqrt[3]{1+\left(\frac{1.6}{6.5}\right)^{3}}
  < 6.54
\end{equation*}
and
\begin{equation*}\textstyle
  \sqrt[p]{7.8^p+4.2^p}
   = 7.8\cdot\sqrt[p]{1+\left(\frac{4.2}{7.8}\right)^p}
  \le 7.8\cdot\sqrt[3]{1+\left(\frac{4.2}{7.8}\right)^{3}}
  < 8.19.
\end{equation*}
Hence, for $p\ge 3$,
\[
\sqrt[p]{0.0^p+1.0^p} + \sqrt[p]{14.3^p+1.6^p} - \sqrt[p]{6.5^p+1.6^p} - \sqrt[p]{7.8^p+4.2^p}
\ge 1.0+14.3-6.54-8.19 > 0. 
\]

\noindent{\bf Inequality 5:}
\begin{alignat*}{5}
\,\,\, & \sqrt[p]{0.3^p+1.7^p} &&+ \sqrt[p]{7.2^p+6.2^p} &&- \sqrt[p]{4.0^p+5.2^p} &&- \sqrt[p]{3.5^p+2.7^p} &&\, > 0\\
\end{alignat*}
For $p\ge 7$, we obtain
\begin{equation*}\textstyle
  \sqrt[p]{4.0^p+5.2^p}
   = 5.2\cdot\sqrt[p]{1+\left(\frac{4.0}{5.2}\right)^p} 
  \le 5.2\cdot\sqrt[7]{1+\left(\frac{4.0}{5.2}\right)^{7}}
  < 5.32
\end{equation*}
and
\begin{equation*}\textstyle
  \sqrt[p]{3.5^p+2.7^p}
   = 3.5\cdot\sqrt[p]{1+\left(\frac{2.7}{3.5}\right)^p}
  \le 3.5\cdot\sqrt[7]{1+\left(\frac{2.7}{3.5}\right)^{7}}
  < 3.58.
\end{equation*}
Hence, for $p\ge 7$,
\[
\sqrt[p]{0.3^p+1.7^p} + \sqrt[p]{7.2^p+6.2^p} - \sqrt[p]{4.0^p+5.2^p} - \sqrt[p]{3.5^p+2.7^p}
> 1.7+7.2-5.32-3.58 = 0. 
\]
For the remaining cases $p\in\{3,4,5,6\}$, the inequality can simply be checked by plugging in the appropriate values.

\noindent{\bf Inequality 6:}
\begin{alignat*}{5}
\,\,\, & \sqrt[p]{3.5^p+3.7^p} &&+ \sqrt[p]{7.8^p+4.2^p} &&- \sqrt[p]{3.5^p+2.7^p} &&- \sqrt[p]{7.8^p+3.2^p} &&\, > 0\\
\end{alignat*}
For $p\ge 5$, we obtain
\begin{equation*}\textstyle
  \sqrt[p]{3.5^p+2.7^p}
   = 3.5\cdot\sqrt[p]{1+\left(\frac{2.7}{3.5}\right)^p} 
  \le 3.5\cdot\sqrt[5]{1+\left(\frac{2.7}{3.5}\right)^{5}}
  < 3.68
\end{equation*}
and
\begin{equation*}\textstyle
  \sqrt[p]{7.8^p+3.2^p}
   = 7.8\cdot\sqrt[p]{1+\left(\frac{3.2}{7.8}\right)^p}
  \le 7.8\cdot\sqrt[5]{1+\left(\frac{3.2}{7.8}\right)^{5}}
  < 7.82.
\end{equation*}
Hence, for $p\ge 5$,
\[
\sqrt[p]{3.5^p+3.7^p} + \sqrt[p]{7.8^p+4.2^p} - \sqrt[p]{3.5^p+2.7^p} - \sqrt[p]{7.8^p+3.2^p}
> 3.7+7.8-3.68-7.82 = 0. 
\]
For the remaining cases $p\in\{3,4\}$, the inequality can simply be checked by plugging in the appropriate values.

\noindent{\bf Inequality 7:}
\begin{alignat*}{5}
\,\,\, & \sqrt[p]{6.5^p+1.6^p} &&+ \sqrt[p]{9.3^p+3.9^p} &&- \sqrt[p]{5.0^p+5.5^p} &&- \sqrt[p]{7.8^p+3.2^p} &&\, > 0\\
\end{alignat*}
For $p\ge 3$, we obtain
\begin{equation*}\textstyle
  \sqrt[p]{5.0^p+5.5^p}
   = 5.5\cdot\sqrt[p]{1+\left(\frac{5.0}{5.5}\right)^p} 
  \le 5.5\cdot\sqrt[3]{1+\left(\frac{5.0}{5.5}\right)^{3}}
  < 6.63
\end{equation*}
and
\begin{equation*}\textstyle
  \sqrt[p]{7.8^p+3.2^p}
   = 7.8\cdot\sqrt[p]{1+\left(\frac{3.2}{7.8}\right)^p}
  \le 7.8\cdot\sqrt[3]{1+\left(\frac{3.2}{7.8}\right)^{3}}
  < 7.98.
\end{equation*}
Hence, for $p\ge 3$,
\[
\sqrt[p]{6.5^p+1.6^p} + \sqrt[p]{9.3^p+3.9^p} - \sqrt[p]{5.0^p+5.5^p} - \sqrt[p]{7.8^p+3.2^p}
\ge 6.5+9.3-6.63-7.98 > 0. 
\]

Inequalities corresponding to the improvements made by the 2-changes 
in the sequence in which gadget $G^R_{n-2}$ resets gadget $G^P_{n-1}$ from 
$(S,S)$ to $(L,L)$:

\noindent{\bf Inequality 1:}
\begin{alignat*}{5}
\,\,\, & \sqrt[p]{27.3^p+21.06^p} &&+ \sqrt[p]{5.0^p+5.5^p} &&- \sqrt[p]{13.7^p+0.9^p} &&- \sqrt[p]{18.6^p+16.46^p} &&\, > 0\\
\end{alignat*}
For $p\ge 10$, we obtain
\begin{equation*}\textstyle
  \sqrt[p]{13.7^p+0.9^p}
   = 13.7\cdot\sqrt[p]{1+\left(\frac{0.9}{13.7}\right)^p} 
  \le 13.7\cdot\sqrt[10]{1+\left(\frac{0.9}{13.7}\right)^{10}}
  < 13.71
\end{equation*}
and
\begin{equation*}\textstyle
  \sqrt[p]{18.6^p+16.46^p}
   = 18.6\cdot\sqrt[p]{1+\left(\frac{16.46}{18.6}\right)^p}
  \le 18.6\cdot\sqrt[10]{1+\left(\frac{16.46}{18.6}\right)^{10}}
  < 19.09.
\end{equation*}
Hence, for $p\ge 10$,
\begin{alignat*}{1}
& \sqrt[p]{27.3^p+21.06^p} + \sqrt[p]{5.0^p+5.5^p} - \sqrt[p]{13.7^p+0.9^p} - \sqrt[p]{18.6^p+16.46^p}\\
 > &\, 27.3+5.5-13.71-19.09 = 0. 
\end{alignat*}
For the remaining cases $p\in\{3,4,5,6,7,8,9\}$, the inequality can simply be checked by plugging in the appropriate values.

\noindent{\bf Inequality 2:}
\begin{alignat*}{5}
\,\,\, & \sqrt[p]{4.0^p+5.2^p} &&+ \sqrt[p]{60.84^p+24.96^p} &&- \sqrt[p]{60.84^p+23.06^p} &&- \sqrt[p]{4.0^p+3.3^p} &&\, > 0\\
\end{alignat*}
For $p\ge 4$, we obtain
\begin{equation*}\textstyle
  \sqrt[p]{60.84^p+23.06^p}
   = 60.84\cdot\sqrt[p]{1+\left(\frac{23.06}{60.84}\right)^p} 
  \le 60.84\cdot\sqrt[4]{1+\left(\frac{23.06}{60.84}\right)^{4}}
  < 61.16
\end{equation*}
and
\begin{equation*}\textstyle
  \sqrt[p]{4.0^p+3.3^p}
   = 4.0\cdot\sqrt[p]{1+\left(\frac{3.3}{4.0}\right)^p}
  \le 4.0\cdot\sqrt[4]{1+\left(\frac{3.3}{4.0}\right)^{4}}
  < 4.4.
\end{equation*}
Hence, for $p\ge 4$,
\begin{alignat*}{1}
&\sqrt[p]{4.0^p+5.2^p} + \sqrt[p]{60.84^p+24.96^p} - \sqrt[p]{60.84^p+23.06^p} - \sqrt[p]{4.0^p+3.3^p}\\
\ge &\, 5.2+60.84-61.16-4.4 > 0. 
\end{alignat*}
For the remaining case $p=3$, the inequality can simply be checked by plugging in the appropriate values.

\noindent{\bf Inequality 3:}
\begin{alignat*}{5}
\,\,\, & \sqrt[p]{60.84^p+23.06^p} &&+ \sqrt[p]{12.3^p+14.4^p} &&- \sqrt[p]{15.8^p+11.8^p} &&- \sqrt[p]{57.34^p+20.46^p} &&\, > 0\\
\end{alignat*}
For $p\ge 4$, we obtain
\begin{equation*}\textstyle
  \sqrt[p]{15.8^p+11.8^p}
   = 15.8\cdot\sqrt[p]{1+\left(\frac{11.8}{15.8}\right)^p} 
  \le 15.8\cdot\sqrt[4]{1+\left(\frac{11.8}{15.8}\right)^{4}}
  < 16.91
\end{equation*}
and
\begin{equation*}\textstyle
  \sqrt[p]{57.34^p+20.46^p}
   = 57.34\cdot\sqrt[p]{1+\left(\frac{20.46}{57.34}\right)^p}
  \le 57.34\cdot\sqrt[4]{1+\left(\frac{20.46}{57.34}\right)^{4}}
  < 57.58.
\end{equation*}
Hence, for $p\ge 4$,
\begin{alignat*}{1}
& \sqrt[p]{60.84^p+23.06^p} + \sqrt[p]{12.3^p+14.4^p} - \sqrt[p]{15.8^p+11.8^p} - \sqrt[p]{57.34^p+20.46^p}\\
\ge & \,60.84+14.4-16.91-57.58 > 0. 
\end{alignat*}
For the remaining case $p=3$, the inequality can simply be checked by plugging in the appropriate values.

\noindent{\bf Inequality 4:}
\begin{alignat*}{5}
\,\,\, & \sqrt[p]{2.2^p+4.9^p} &&+ \sqrt[p]{18.6^p+16.46^p} &&- \sqrt[p]{15.4^p+16.26^p} &&- \sqrt[p]{1.0^p+4.7^p} &&\, > 0\\
\end{alignat*}
For $p\ge 5$, we obtain
\begin{equation*}\textstyle
  \sqrt[p]{15.4^p+16.26^p}
   = 16.26\cdot\sqrt[p]{1+\left(\frac{15.4}{16.26}\right)^p} 
  \le 16.26\cdot\sqrt[5]{1+\left(\frac{15.4}{16.26}\right)^{5}}
  < 18.22
\end{equation*}
and
\begin{equation*}\textstyle
  \sqrt[p]{1.0^p+4.7^p}
   = 4.7\cdot\sqrt[p]{1+\left(\frac{1.0}{4.7}\right)^p}
  \le 4.7\cdot\sqrt[5]{1+\left(\frac{1.0}{4.7}\right)^{5}}
  < 4.71.
\end{equation*}
Hence, for $p\ge 5$,
\begin{alignat*}{1}
& \sqrt[p]{2.2^p+4.9^p} + \sqrt[p]{18.6^p+16.46^p} - \sqrt[p]{15.4^p+16.26^p} - \sqrt[p]{1.0^p+4.7^p}\\
\ge &\, 4.9+18.6-18.22-4.71 > 0. 
\end{alignat*}
For the remaining cases $p\in\{3,4\}$, the inequality can simply be checked by plugging in the appropriate values.

\noindent{\bf Inequality 5:}
\begin{alignat*}{5}
\,\,\, & \sqrt[p]{13.7^p+0.9^p} &&+ \sqrt[p]{4.0^p+3.3^p} &&- \sqrt[p]{0.0^p+7.8^p} &&- \sqrt[p]{9.7^p+3.6^p} &&\, > 0\\
\end{alignat*}
For $p\ge 3$, we obtain
\begin{equation*}\textstyle
  \sqrt[p]{9.7^p+3.6^p}
   = 9.7\cdot\sqrt[p]{1+\left(\frac{3.6}{9.7}\right)^p}
  \le 9.7\cdot\sqrt[3]{1+\left(\frac{3.6}{9.7}\right)^{3}}
  < 9.87.
\end{equation*}
Hence, for $p\ge 3$,
\[
\sqrt[p]{13.7^p+0.9^p} + \sqrt[p]{4.0^p+3.3^p} - \sqrt[p]{0.0^p+7.8^p} - \sqrt[p]{9.7^p+3.6^p}
\ge 13.7+4.0-7.8-9.87 > 0. 
\]

\noindent{\bf Inequality 6:}
\begin{alignat*}{5}
\,\,\, & \sqrt[p]{15.8^p+11.8^p} &&+ \sqrt[p]{1.0^p+4.7^p} &&- \sqrt[p]{6.1^p+2.2^p} &&- \sqrt[p]{8.7^p+14.3^p} &&\, > 0\\
\end{alignat*}
For $p\ge 7$, we obtain
\begin{equation*}\textstyle
  \sqrt[p]{6.1^p+2.2^p}
   = 6.1\cdot\sqrt[p]{1+\left(\frac{2.2}{6.1}\right)^p} 
  \le 6.1\cdot\sqrt[7]{1+\left(\frac{2.2}{6.1}\right)^{7}}
  < 6.11
\end{equation*}
and
\begin{equation*}\textstyle
  \sqrt[p]{8.7^p+14.3^p}
   = 14.3\cdot\sqrt[p]{1+\left(\frac{8.7}{14.3}\right)^p}
  \le 14.3\cdot\sqrt[7]{1+\left(\frac{8.7}{14.3}\right)^{7}}
  < 14.37.
\end{equation*}
Hence, for $p\ge 7$,
\begin{alignat*}{1}
& \sqrt[p]{15.8^p+11.8^p} + \sqrt[p]{1.0^p+4.7^p} - \sqrt[p]{6.1^p+2.2^p} - \sqrt[p]{8.7^p+14.3^p}\\
\ge &\, 15.8+4.7-6.11-14.37 > 0. 
\end{alignat*}
For the remaining cases $p\in\{3,4,5,6\}$, the inequality can simply be checked by plugging in the appropriate values.

\noindent{\bf Inequality 7:}
\begin{alignat*}{5}
\,\,\, & \sqrt[p]{15.4^p+16.26^p} &&+ \sqrt[p]{57.34^p+20.46^p} &&- \sqrt[p]{33.54^p+53.82^p} &&- \sqrt[p]{8.4^p+17.1^p} &&\, > 0\\
\end{alignat*}
For $p\ge 4$, we obtain
\begin{equation*}\textstyle
  \sqrt[p]{33.54^p+53.82^p}
   = 53.82\cdot\sqrt[p]{1+\left(\frac{33.54}{53.82}\right)^p} 
  \le 53.82\cdot\sqrt[4]{1+\left(\frac{33.54}{53.82}\right)^{4}}
  < 55.75
\end{equation*}
and
\begin{equation*}\textstyle
  \sqrt[p]{8.4^p+17.1^p}
   = 17.1\cdot\sqrt[p]{1+\left(\frac{8.4}{17.1}\right)^p}
  \le 17.1\cdot\sqrt[4]{1+\left(\frac{8.4}{17.1}\right)^{4}}
  < 17.35.
\end{equation*}
Hence, for $p\ge 4$,
\begin{alignat*}{1}
&\sqrt[p]{15.4^p+16.26^p} + \sqrt[p]{57.34^p+20.46^p} - \sqrt[p]{33.54^p+53.82^p} - \sqrt[p]{8.4^p+17.1^p}\\
\ge &\,16.26+57.34-55.75-17.35 > 0. 
\end{alignat*}{1}
For the remaining case $p=3$, the inequality can simply be checked by plugging in the appropriate values.

\section{Some Probability Theory}
\label{appendix:probTheory}

\begin{lemma}
\label{lemma:ProbLinComb}
Let $X^1,\ldots,X^n\in [0,1]^d$ be stochastically independent $d$-dimensional random row vectors, 
and, for $i\in[n]$ and some $\phi\ge 1$, let 
$f_i\colon[0,1]^d\to[0,\phi]$ denote the joint probability density of the entries 
of $X^i$. Furthermore, let $\lambda^1,\ldots,\lambda^k\in\ZZ^{dn}$ be fixed
linearly independent row vectors. For $i\in[n]$ and a fixed 
$\varepsilon\ge 0$, we denote by ${\cal A}_i$ the event that 
$\lambda^{i}\cdot X$ takes a value in the interval $[0,\varepsilon]$, 
where $X$ denotes the vector $X=(X^1,\ldots,X^n)\tra$. Under these 
assumptions,
\[ 
\Pr{\bigcap_{i=1}^{k}{\cal A}_i}\le (\varepsilon\phi)^{k}. 
\]
\end{lemma}

\begin{proof}
The main tool for proving the lemma is a change of variables. Instead 
of using the canonical basis of the $dn$-dimensional vector space 
$\RR^{dn}$, we use the given linear combinations as basis vectors. To 
be more precise, the basis ${\cal B}$ that we use consists of two 
parts: it contains the vectors $\lambda^{1},\ldots,\lambda^{k}$ and it 
is completed by some vectors from the canonical basis 
$\{e^{1},\ldots,e^{dn}\}$, where $e^{i}$ denotes the $i$-th canonical 
row vector, \ie, $e^i_i=1$ and $e^i_j=0$ for $j\neq i$. That is, the 
basis ${\cal B}$ can be written as 
$\{\lambda^{1},\ldots,\lambda^{k},e^{\pi(1)},\ldots,e^{\pi(dn-k)}\}$, 
for some injective function $\pi\colon[dn-k]\to[dn]$.

Let $\Phi\colon\RR^{dn}\to\RR^{dn}$ be defined by $\Phi(x)=Ax$, where 
$A$ denotes the $(dn)\times(dn)$-matrix 
\[
   \begin{pmatrix}
     \lambda^{1}\\
     \vdots\\
     \lambda^{k}\\
     e^{\pi(1)}\\
     \vdots\\
     e^{\pi(dn-k)}
   \end{pmatrix}.
\]
Since ${\cal B}$ is a basis of $\mathbb{R}^{dn}$, the function $\Phi$ is a 
bijection. We define $Y=(Y_1,\ldots,Y_{dn})\tra$ as $Y=\Phi(X)$, and 
for $i\in[n]$, we denote by $Y^i$ the vector 
$(Y_{d(i-1)+1},\ldots,Y_{di})$. Let $f\colon\RR^{dn}\to\RR$ denote the 
joint density of the entries of the random vectors $X^1,\ldots,X^{n}$, 
and let $g\colon\RR^{dn}\to\RR$ denote the joint density of the entries 
of the random vectors $Y^1,\ldots,Y^{n}$. Due to the independence of 
the random vectors $X^1,\ldots,X^{n}$, we have 
$f(x_1,\ldots,x_{dn})=f_1(x_1,\ldots,x_d)\cdot\cdots\cdot 
f_{n}(x_{d(n-1)+1},\ldots,x_{dn})$. We can express the joint density 
$g$ as
\[
   g(y_1,\ldots,y_{dn}) = |\det_{\partial} \Phi^{-1}(y_1,\ldots,y_{dn})\,|\cdot f(\Phi^{-1}(y_1,\ldots,y_{dn})),
\]
where $\det_{\partial}$ denotes the determinant of the Jacobian matrix 
of $\Phi^{-1}$ (see, e.g.,~\cite{SpielmanT04}).

The matrix $A$ is invertible as ${\cal B}$ is a basis of 
$\mathbb{R}^{dn}$. Hence, for $y\in\RR^{dn}$, $\Phi^{-1}(y)=A^{-1}y$ 
and the Jacobian matrix of $\Phi^{-1}$ equals $A^{-1}$. Thus, 
$\det_{\partial}\Phi^{-1}=\det A^{-1}=(\det A)^{-1}$. Since all entries 
of $A$ are integers, also its determinant must be an integer, and since 
it has rank $dn$, we know that $\det A\neq0$. Hence, $|\det A\,|\ge 1$ 
and $|\det A^{-1}\,|\le 1$. For $y\in\RR^{dn}$, we decompose 
$\Phi^{-1}(y)\in\RR^{dn}$ into $n$ subvectors with $d$ entries each, 
\ie, $\Phi^{-1}(y)=(\Phi^{-1}_1(y),\ldots,\Phi^{-1}_n(y))$ with 
$\Phi^{-1}_i(y)\in\RR^d$ for $i\in[n]$. This yields
\begin{equation*}
 g(y) = |\det A^{-1}\,|\cdot f(\Phi^{-1}(y)) \le
 f_1(\Phi^{-1}_1(y))\cdots f_{n}(\Phi^{-1}_{n}(y)),
\end{equation*}
where we used that $|\det A^{-1}\,|\le 1$ and that the vectors
$X^1,\ldots,X^n$ are stochastically independent.

The probability we want to estimate can be written as
\begin{equation}\label{eqn:appendix:probA}
  \Pr{\bigcap_{i=1}^{k}{\cal A}_i} = \int_{y_1=0}^{\varepsilon}\cdots
  \int_{y_{k}=0}^{\varepsilon}\int_{y_{k+1}=-\infty}^{\infty}\cdots\int_{y_{dn}=-\infty}^{\infty}g(y_1,\ldots,y_{dn})\,dy_{dn}\cdots
  dy_1.
\end{equation}
Since all entries of the vectors $X^1,\ldots,X^n$ take only values in
the interval $[0,1]$ and since for $i\in\{k+1,\ldots,dn\}$, the random
variable $Y_i$ coincides with one of these entries,
(\ref{eqn:appendix:probA}) simplifies to
\begin{equation}
  \label{eqn:probEstimate}
  \Pr{\bigcap_{i=1}^{k}{\cal A}_i} = \int_{y_1=0}^{\varepsilon}\cdots
  \int_{y_{k}=0}^{\varepsilon}\int_{y_{k+1}=0}^{1}\cdots\int_{y_{dn}=0}^{1}g(y_1,\ldots,y_{dn})\,dy_{dn}\cdots
  dy_1.
\end{equation}

By the definition of~$\pi$, the basis ${\cal B}$ consists of the vectors $\lambda^{1},\ldots,\lambda^{k}$
and the canonical vectors $e^i$ for $i\in\Pi=\{\ell\mid\exists j\in[dn-k]\colon \pi(j)=\ell\}$.
We divide the vectors $e^{1},\ldots,e^{dn}$ into $n$ groups of $d$ vectors each, \ie, the first group
consists of the vectors $e^{1},\ldots,e^{d}$, the second group consists of the vectors $e^{d+1},\ldots,e^{2d}$,
and so on. The set of vectors $e^{i}$ with $i\notin\Pi$, \ie, the vectors from the canonical basis that
are replaced by the vectors $\lambda^{1},\ldots,\lambda^{k}$ in basis ${\cal B}$, can intersect at most $k$
of these groups. In order to simplify the notation, we reorder and rename the groups such that only vectors from the
first $k$ groups are replaced by the vectors $\lambda^{1},\ldots,\lambda^{k}$. As every group consists of $d$
vectors, we can assume that, after renaming, $[dn]\setminus\Pi\subseteq [dk]$, \ie,
only vectors $e^{i}$ from the canonical basis with $i\le dk$ are replaced by the vectors 
$\lambda^{1},\ldots,\lambda^{k}$ in the basis $\cal B$. After that, we can reorder and rename the groups $k+1,\ldots,n$ 
such that $\pi(i)=i$, for $i>dk$. This implies, in particular, that for $i>k$ we have $\Phi^{-1}_i(y)=(y_{di+1},\ldots,y_{d(i+1)})$. 
Under these assumptions, the density $g$ can be upper bounded as follows:
\begin{equation}
  \label{eqn:jointDensity}
  g(y_1,\ldots,y_{dn}) \le \phi^{k}\cdot f_{k+1}(y_{dk+1},\ldots,y_{d(k+1)})\cdots f_n(y_{d(n-1)+1},\ldots,y_{dn}),
\end{equation}
where we bounded each of the densities $f_1,\ldots,f_k$ from above by $\phi$ and used that
$\Phi^{-1}_i(y)=(y_{di+1},\ldots,y_{d(i+1)})$ for $i>k$. 

Putting together~(\ref{eqn:probEstimate}) and~(\ref{eqn:jointDensity}) 
yields
\begin{equation*}
\begin{split}
\Pr{\bigcap_{i=1}^k{\cal A}_i}&
\le (\varepsilon\phi)^{k}\cdot\left(
\int_{y_{dk+1}=0}^{1}\cdots\int_{y_{d(k+1)}=0}^{1}f_{k+1}(y_{dk+1},\ldots,y_{d(k+1)})\right.\\
& \left.\quad\ldots
\int_{y_{d(n-1)+1}=0}^{1}\int_{y_{dn}=0}^{1}f_{n}(y_{d(n-1)+1},\ldots,y_{dn})
\,dy_{dn}\cdots dy_{dk+1}\right)\\
& = (\varepsilon\phi)^{k},
\end{split}
\end{equation*}
where the last equation follows because $f_{k+1},\ldots,f_n$ are density 
functions. The occurrence of $\varepsilon^k$ is due to the first $k$ integrals
in~(\ref{eqn:probEstimate}) because each of the variables $y_1,\ldots,y_k$ is integrated
over an interval of length $\varepsilon$ and none of them appears in the integrand coming
from~(\ref{eqn:jointDensity}).
\end{proof}

\section{Proofs of some Lemmas from Section~\ref{subsec:UB:L2}}

\subsection{Proof of Lemma~\ref{lemma:densityDelta}}
\label{appendix:proof:densityDelta}

Let $a,c\in(0,C]$ for some $C>0$. In the following proof, we use the 
following two identities (see~\cite{bronshtein2007handbook}):
\begin{align*}
  \int_{0}^c\frac{1}{\sqrt{z(c-z)}}\,dz & = 
  \left[ \arctan\left(\frac{z-c/2}{\sqrt{z(c-z)}}\right) \right]_{0}^c \\
  & = \left(\lim_{x\to\infty}\arctan(x)\right) - \left(\lim_{x\to-\infty}\arctan(x)\right) 
  = \frac{\pi}{2} - (-\frac{\pi}{2}) = \pi 
\end{align*}
and
\begin{alignat*}{1}
  \int_{0}^{a}\frac{1}{\sqrt{z(z+c)}}\,dz & =
  \left[\ln\left(\frac{c}{2}+z+\sqrt{z(z+c)}\right)\right]_{0}^{a}\\  
  & = \ln\left(\frac{c}{2}+a+\sqrt{a(a+c)}\right) - \ln\left(\frac{c}{2}\right)\\
  & \le \ln\left(\frac{c}{2}+a+\sqrt{(a+c)(a+c)}\right) + \ln\left(\frac{2}{c}\right)\\
  & = \ln\left(\frac{3}{2}c+2a\right)+\ln\left(\frac{2}{c}\right)
  \le \ln\left(4C\right)+\ln\left(\frac{2}{c}\right).
\end{alignat*}

Since in both identities the integrands are non-negative, the following inequalities are true
for any~$[\alpha_1,\alpha_2]\subseteq[0,c]$ and~$[\beta_1,\beta_2]\subseteq [0,a]$:
\begin{equation}\label{eqn:Integral1}
     \int_{\alpha_1}^{\alpha_2}\frac{1}{\sqrt{z(c-z)}}\,dz \le \pi
\end{equation}
and
\begin{equation}\label{eqn:Integral2}
   \int_{\beta_1}^{\beta_2}\frac{1}{\sqrt{z(z+c)}}\,dz \le \ln\left(4C\right)+\ln\left(\frac{2}{c}\right).
\end{equation}
We will frequently use these inequalities in the following.

\begin{proof}[Lemma~\ref{lemma:densityDelta}]
The conditional density of $\Delta$ can be calculated as convolution of
the conditional densities of $Z_1$ and $Z_2$ as follows:
\[
   f_{\Delta|T=\tau,R_1=r_1,R_2=r_2}(\delta) 
     =  \int_{-\infty}^{\infty}f_{Z|T=\tau,R=r_1}(z)\cdot f_{Z|T=\tau,R=r_2}(z-\delta)\, dz.       
\]
In order to estimate this integral, we distinguish between several 
cases. In the following, let $\kappa$ denote a
sufficiently large constant.

{\bf First case: $\tau\le r_1$ and $\tau\le r_2$.} \\
Since $Z_i$ takes only values in the interval $[-\tau,\tau]$, we can 
assume $0<\delta\le\min\{1/2,2\tau\}$ and
\[
 f_{\Delta|T=\tau,R_1=r_1,R_2=r_2}(\delta)
 =  \int_{-\tau+\delta}^{\tau}f_{Z|T=\tau,R=r_1}(z)\cdot f_{Z|T=\tau,R=r_2}(z-\delta)\, dz.
\]
Due to Lemma~\ref{lemma:densityZ}, we can estimate the densities of $Z_1$ and $Z_2$ by
\begin{equation}\label{FirstEstimateZ}
  f_{Z|T=\tau,R=r_i}(z) 
     \le  \sqrt{\frac{2}{\tau^2-z^2}}
    \le \sqrt{\frac{2}{\tau(\tau-|z|)}} 
     \le  \sqrt{\frac{2}{\tau}}\left(\frac{1}{\sqrt{\tau-z}}+\frac{1}{\sqrt{\tau+z}}\right).
\end{equation}
For $\delta\in (0,\min\{1/2,2\tau\}]$,
we obtain the following upper bound on the density of $\Delta$:
\begin{align*}
  & f_{\Delta|T=\tau,R_1=r_1,R_2=r_2}(\delta) \\
   \le & \frac{2}{\tau} \int_{-\tau+\delta}^{\tau}\left(\frac{1}{\sqrt{\tau-z}}+\frac{1}{\sqrt{\tau+z}}\right)\left(\frac{1}{\sqrt{\tau-z+\delta}}+\frac{1}{\sqrt{\tau+z-\delta}}\right)\, dz  \\
   = & \frac{2}{\tau} \left(\int_{-\tau+\delta}^{\tau}\frac{1}{\sqrt{(\tau-z)(\tau-z+\delta)}}\,dz 
                           + \int_{-\tau+\delta}^{\tau}\frac{1}{\sqrt{(\tau+z)(\tau-z+\delta)}}\,dz \right.\\
   & \left.               + \int_{-\tau+\delta}^{\tau}\frac{1}{\sqrt{(\tau-z)(\tau+z-\delta)}}\,dz
                           + \int_{-\tau+\delta}^{\tau}\frac{1}{\sqrt{(\tau+z)(\tau+z-\delta)}}\,dz \right) \\
   = & \frac{2}{\tau} \left(\int_{0}^{2\tau-\delta}\frac{1}{\sqrt{z'(z'+\delta)}}\,dz' 
                           + \int_{\delta}^{2\tau}\frac{1}{\sqrt{z'(2\tau+\delta-z')}}\,dz' \right.\\
   & \left.               + \int_{0}^{2\tau-\delta}\frac{1}{\sqrt{z'(2\tau-\delta-z')}}\,dz'
                           + \int_{0}^{2\tau-\delta}\frac{1}{\sqrt{z'(z'+\delta)}}\,dz' \right).
\end{align*}
For the four integrals, we used the substitutions
$z'=\tau-z$, $z'=\tau+z$, $z'=\tau-z$, and $z'=\tau-\delta+z$, respectively.
Using~\eqref{eqn:Integral1} and~\eqref{eqn:Integral2} and the fact that $2\tau-\delta \le 2\sqrt{d}=O(1)$ 
yields that the previous term is bounded from above by
\begin{align*}
  & \frac{2}{\tau}\left(
    \left(\ln(4(2\sqrt{d}))+\ln\left(2\delta^{-1}\right)\right) + \pi + \pi
  + \left(\ln(4(2\sqrt{d}))+\ln\left(2\delta^{-1}\right)\right)\right)\\
  = & \frac{2}{\tau}\left(2\pi+2\ln(8\sqrt{d})+2\ln\left(2\delta^{-1}\right)\right)
  = \frac{O(1) + 4\ln\left(\delta^{-1}\right)}{\tau}.
\end{align*}
Since we assume that $\delta\le 1/2$, the logarithm $\ln(\delta^{-1})$ is bounded from below by 
the constant $\ln(2)$. Using this observation, we can absorb the~$O(1)$ term and bound the previous expression from above by
\[
  \frac{\kappa}{\tau}\cdot\ln\left(\delta^{-1}\right)
\]
if $\kappa$ is a large enough constant.

{\bf Second case: $r_1\le \tau$ and $r_2\le\tau$.} \\
Since $Z_i$ takes only values in the interval $[-\tau,2r_i-\tau]$, we 
can assume $0<\delta\le\min\{1/2,2r_1\}$ and
\[
 f_{\Delta|T=\tau,R_1=r_1,R_2=r_2}(\delta) 
 = \int_{-\tau+\delta}^{\min\{2r_1-\tau,2r_2-\tau+\delta\}}f_{Z|T=\tau,R=r_1}(z)\cdot f_{Z|T=\tau,R=r_2}(z-\delta)\,dz.
\]
The limits of the integral follow because~$f_{Z|T=\tau,R=r_1}(z)$ is only nonzero for~$z\in[-\tau,2r_1-\tau]$
and~$f_{Z|T=\tau,R=r_2}(z-\delta)$ is only nonzero for~$z\in[-\tau+\delta,2r_2-\tau+\delta]$.
The intersection of these two intervals is~$[-\tau+\delta,\min\{2r_1-\tau,2r_2-\tau+\delta\}]$.

Due to Lemma~\ref{lemma:densityZ}, we can estimate the densities of $Z_1$ and $Z_2$ by
\begin{align}
  f_{Z|T=\tau,R_i=r_i}(z) \le &  \sqrt{\frac{2}{(\tau+z)(2r_i-\tau-z)}} 
     \le  \begin{cases}
          \sqrt{\frac{2}{r_i(\tau+z)}} & \mbox{if }  z\le r_i-\tau \\
          \sqrt{\frac{2}{r_i(2r_i-\tau-z)}} & \mbox{if }  z\ge r_i-\tau
      \end{cases}  \notag\\
    \le & \sqrt{\frac{2}{r_i}}\left(\frac{1}{\sqrt{\tau+z}}+\frac{1}{\sqrt{2r_i-\tau-z}}\right).\label{SecondEstimateZ}
\end{align}
{\bf Case 2.1: $\delta\in(\max\{0,2(r_1-r_2)\},2r_1]$.} \\
We obtain the following upper bound on the density of $\Delta$:
\begin{align*}
  & f_{\Delta|T=\tau,R_1=r_1,R_2=r_2}(\delta) \\
   \le & \frac{2}{\sqrt{r_1r_2}} \int_{-\tau+\delta}^{2r_1-\tau}\left(\frac{1}{\sqrt{\tau+z}}+\frac{1}{\sqrt{2r_1-\tau-z}}\right)\left(\frac{1}{\sqrt{\tau+z-\delta}}+\frac{1}{\sqrt{2r_2-\tau-z+\delta}}\right)\, dz  \\
   = & \frac{2}{\sqrt{r_1r_2}} \left(\int_{-\tau+\delta}^{2r_1-\tau}\frac{1}{\sqrt{(\tau+z)(\tau+z-\delta)}}\,dz 
                                                  + \int_{-\tau+\delta}^{2r_1-\tau}\frac{1}{\sqrt{(2r_1-\tau-z)(\tau+z-\delta)}}\,dz \right.\\
   & \left.                                    + \int_{-\tau+\delta}^{2r_1-\tau}\!\frac{1}{\sqrt{(\tau+z)(2r_2-\tau-z+\delta)}}\,dz
                                                  + \int_{-\tau+\delta}^{2r_1-\tau}\!\frac{1}{\sqrt{(2r_1-\tau-z)(2r_2-\tau-z+\delta)}}\,dz \right)\\
   = & \frac{2}{\sqrt{r_1r_2}} \left(\int_{0}^{2r_1-\delta}\frac{1}{\sqrt{(z'+\delta)z'}}\,dz' 
                           + \int_{0}^{2r_1-\delta}\frac{1}{\sqrt{(2r_1-\delta-z')z'}}\,dz' \right.\\
   & \left.               + \int_{\delta}^{2r_1}\frac{1}{\sqrt{z'(2r_2+\delta-z')}}\,dz'
                           + \int_{0}^{2r_1-\delta}\frac{1}{\sqrt{z'(2(r_2-r_1)+\delta+z')}}\,dz' \right).
\end{align*}
For the four integrals, we used the substitutions
$z'=z+\tau-\delta$, $z'=z+\tau-\delta$, $z'=z+\tau$, and $z'=2r_1-\tau-z$, respectively.
Using~\eqref{eqn:Integral1} and~\eqref{eqn:Integral2} and the facts that $2r_1-\delta \le 2\sqrt{d}$
and $2(r_2-r_1)+\delta \le 2r_2 \le 2\sqrt{d}$ yields that the previous term is bounded from above by
\begin{align*}
     & \frac{2}{\sqrt{r_1r_2}}\left(\left(\ln(4(2\sqrt{d}))+\ln\left(2\delta^{-1}\right)\right)+\pi+\pi+\left(\ln(4(2\sqrt{d}))+\ln\left(2(2(r_2-r_1)+\delta)^{-1}\right)\right)\right)\\
     = & 
     \frac{2}{\sqrt{r_1r_2}}\left( 2\pi+2\ln(8\sqrt{d})+\ln\left(2\delta^{-1}\right)+\ln\left(2(2(r_2-r_1)+\delta)^{-1}\right)\right) \\
     \le & \frac{2}{\sqrt{r_1r_2}}\left(\ln\left(\delta^{-1}\right)+\ln\left((2(r_2-r_1)+\delta)^{-1}\right)+O(1)\right)\\
     \le & \frac{\kappa}{\sqrt{r_1r_2}}\left(\ln\left(\delta^{-1}\right)+\ln\left((2(r_2-r_1)+\delta)^{-1}\right)\right),
\end{align*}
where the last inequality assumes that $\kappa$ is a large enough constant.

{\bf Case 2.2: $\delta\in(0,\max\{0,2(r_1-r_2)\})$.} \\
We obtain the following upper bound on the density of $\Delta$:
\begin{align*}
  & f_{\Delta|T=\tau,R_1=r_1,R_2=r_2}(\delta) \\
   \le & \frac{2}{\sqrt{r_1r_2}} \int_{-\tau+\delta}^{2r_2-\tau+\delta}\!\!\left(\frac{1}{\sqrt{\tau+z}}+\frac{1}{\sqrt{2r_1-\tau-z}}\right)\left(\frac{1}{\sqrt{\tau+z-\delta}}+\frac{1}{\sqrt{2r_2-\tau-z+\delta}}\right)dz  \\
   = & \frac{2}{\sqrt{r_1r_2}} \left(\int_{-\tau+\delta}^{2r_2-\tau+\delta}\frac{1}{\sqrt{(\tau+z)(\tau+z-\delta)}}\,dz 
                                                  + \int_{-\tau+\delta}^{2r_2-\tau+\delta}\frac{1}{\sqrt{(2r_1-\tau-z)(\tau+z-\delta)}}\,dz \right.\\
   & \left.                                    +\! \int_{-\tau+\delta}^{2r_2-\tau+\delta}\!\!\!\!\!\!\frac{1}{\sqrt{(\tau+z)(2r_2-\tau-z+\delta)}}\,dz
                                                  \!+\! \int_{-\tau+\delta}^{2r_2-\tau+\delta}\!\!\!\!\!\!\frac{1}{\sqrt{(2r_1-\tau-z)(2r_2-\tau-z+\delta)}}\,dz \right) \\
   = & \frac{2}{\sqrt{r_1r_2}} \left(\int_{0}^{2r_2}\frac{1}{\sqrt{(z'+\delta)z'}}\,dz' 
                           + \int_{0}^{2r_2}\frac{1}{\sqrt{(2r_1-\delta-z')z'}}\,dz' \right.\\
   & \left.               + \int_{0}^{2r_2}\frac{1}{\sqrt{(2r_2+\delta-z')z'}}\,dz'
                           + \int_{0}^{2r_2}\frac{1}{\sqrt{(2(r_1-r_2)-\delta+z')z'}}\,dz' \right) \\
 \le & \frac{2}{\sqrt{r_1r_2}}\left(2\pi+2\ln(8\sqrt{d})+\ln\left(2\delta^{-1}\right)+\ln\left(2(2(r_1-r_2)-\delta)^{-1}\right)\right) \\
 \le & \frac{\kappa}{\sqrt{r_1r_2}}\left(\ln\left(\delta^{-1}\right)+\ln\left((2(r_1-r_2)-\delta)^{-1}\right)\right).
\end{align*}
For the four integrals, we used the substitutions
$z'=z+\tau-\delta$, $z'=z+\tau-\delta$, $z'=2r_2-\tau+\delta-z$, and $z'=2r_2-\tau+\delta-z$, respectively.
Using~\eqref{eqn:Integral1} and~\eqref{eqn:Integral2} and the facts that $\delta\le 2(r_1-r_2) \le 2\sqrt{d}$
and $2(r_1-r_2)-\delta \le 2(r_1-r_2) \le 2\sqrt{d}$ yields the penultimate inequality.
The last inequality follows for the same reasons as in Case 2.1.

{\bf Third case: $r_1\le\tau\le r_2$.} \\
Since $Z_1$ takes only values in the interval $[-\tau,2r_1-\tau]$ and 
$Z_2$ takes only values in the interval $[-\tau,\tau]$, the random variable~$\Delta=Z_1-Z_2$
takes only values in the interval~$[-2\tau,2r_1]$. For~$\delta\notin[-2\tau,2r_1]$, the density of~$\Delta$
is trivially zero. As additionally, by definition,~$\delta\in(0,1/2]$, we can assume $0<\delta\le\min\{1/2,2r_1\}$ and
\[
 f_{\Delta|T=\tau,R_1=r_1,R_2=r_2}(\delta)
 = \int_{-\tau+\delta}^{2r_1-\tau}f_{Z|T=\tau,R=r_1}(z)\cdot f_{Z|T=\tau,R=r_2}(z-\delta)\, dz.
\]
Using~\eqref{FirstEstimateZ} and~\eqref{SecondEstimateZ}, we obtain
the following upper bound on the density of $\Delta$ for $\delta\in(0,\min\{1/2,2r_1\}]$:
\begin{align*}
  & f_{\Delta|T=\tau,R_1=r_1,R_2=r_2}(\delta) \\
   \le & \frac{2}{\sqrt{\tau r_1}} \int_{-\tau+\delta}^{2r_1-\tau}\left(\frac{1}{\sqrt{\tau+z}}+\frac{1}{\sqrt{2r_1-\tau-z}}\right)\left(\frac{1}{\sqrt{\tau-z+\delta}}+\frac{1}{\sqrt{\tau+z-\delta}}\right)\, dz  \\
   = & \frac{2}{\sqrt{\tau r_1}} \left(\int_{-\tau+\delta}^{2r_1-\tau}\frac{1}{\sqrt{(\tau+z)(\tau-z+\delta)}}\,dz 
                                                    + \int_{-\tau+\delta}^{2r_1-\tau}\frac{1}{\sqrt{(2r_1-\tau-z)(\tau-z+\delta)}}\,dz \right.\\
   & \left.                                      + \int_{-\tau+\delta}^{2r_1-\tau}\frac{1}{\sqrt{(\tau+z)(\tau+z-\delta)}}\,dz
                                                    + \int_{-\tau+\delta}^{2r_1-\tau}\frac{1}{\sqrt{(2r_1-\tau-z)(\tau+z-\delta)}}\,dz \right) \\
   = & \frac{2}{\sqrt{\tau r_1}} \left(\int_{\delta}^{2r_1}\frac{1}{\sqrt{z'(2\tau+\delta-z')}}\,dz' 
                           + \int_{0}^{2r_1-\delta}\frac{1}{\sqrt{z'(2(\tau-r_1)+\delta+z')}}\,dz' \right.\\
   & \left.               + \int_{0}^{2r_1-\delta}\frac{1}{\sqrt{(z'+\delta)z'}}\,dz'
                           + \int_{0}^{2r_1-\delta}\frac{1}{\sqrt{(2r_1-\delta-z')z'}}\,dz' \right) \\
 \le & \frac{2}{\sqrt{\tau r_1}}\left(2\pi+2\ln(8\sqrt{d})+\ln\left(2\delta^{-1}\right)+\ln\left(2(2(\tau-r_1)+\delta)^{-1}\right)\right) \\
 \le & \frac{\kappa}{\sqrt{\tau r_1}}\cdot\ln\left(\delta^{-1}\right).
\end{align*}
For the four integrals, we used the substitutions
$z'=z+\tau$, $z'=2r_1-\tau-z$, $z'=z+\tau-\delta$, and $z'=z+\tau-\delta$, respectively.
Using~\eqref{eqn:Integral1} and~\eqref{eqn:Integral2} and the facts that $2(\tau-r_1)+\delta\le 2\tau \le 2\sqrt{d}$
and $\delta \le 2r_1 \le 2\sqrt{d}$ yields the penultimate inequality.
The last inequality follows for the same reasons as in Case 2.1.

{\bf Fourth case: $r_2\le\tau\le r_1$.} \\
Since $Z_1$ takes only values in the interval $[-\tau,\tau]$ and $Z_2$ 
takes only values in the interval $[-\tau,2r_2-\tau]$, the random variable~$\Delta=Z_1-Z_2$
takes only values in the interval~$[-2r_2,2\tau]$. For~$\delta\notin[-2r_2,2\tau]$, the density of~$\Delta$
is trivially zero. As additionally, by definition,~$\delta\in(0,1/2]$, we can assume 
$0<\delta\le \min\{1/2,2\tau\}$ and
\[
 f_{\Delta|T=\tau,R_1=r_1,R_2=r_2}(\delta)
 = \int_{-\tau+\delta}^{\min\{2r_2-\tau+\delta,\tau\}}f_{Z|T=\tau,R=r_1}(z)\cdot f_{Z|T=\tau,R=r_2}(z-\delta)\, dz.
\]
The limits of the integral follow because~$f_{Z|T=\tau,R=r_1}(z)$ is only nonzero for~$z\in[-\tau,\tau]$
and~$f_{Z|T=\tau,R=r_2}(z-\delta)$ is only nonzero for~$z\in[-\tau+\delta,2r_2-\tau+\delta]$.
The intersection of these two intervals is~$[-\tau+\delta,\min\{2r_2-\tau+\delta,\tau\}]$.

{\bf Case 4.1: $\delta\in(0,2(\tau-r_2))$.} \\
Using~\eqref{FirstEstimateZ} and~\eqref{SecondEstimateZ}, we obtain
the following upper bound on the density of $\Delta$:
\begin{align*}
  & f_{\Delta|T=\tau,R_1=r_1,R_2=r_2}(\delta) \\
   \le & \frac{2}{\sqrt{\tau r_2}}
  \int_{-\tau+\delta}^{2r_2-\tau+\delta}\left(\frac{1}{\sqrt{\tau-z}}+\frac{1}{\sqrt{\tau+z}}\right)\left(\frac{1}{\sqrt{\tau+z-\delta}}+\frac{1}{\sqrt{2r_2-\tau-z+\delta}}\right)\, dz  \\
   = & \frac{2}{\sqrt{\tau r_2}} \left(\int_{-\tau+\delta}^{2r_2-\tau+\delta}\frac{1}{\sqrt{(\tau+z)(\tau+z-\delta)}}\,dz 
                                                     + \int_{-\tau+\delta}^{2r_2-\tau+\delta}\frac{1}{\sqrt{(\tau-z)(\tau+z-\delta)}}\,dz \right.\\
   & \left.                                       + \int_{-\tau+\delta}^{2r_2-\tau+\delta}\frac{1}{\sqrt{(\tau+z)(2r_2-\tau-z+\delta)}}\,dz
                                                     + \int_{-\tau+\delta}^{2r_2-\tau+\delta}\frac{1}{\sqrt{(\tau-z)(2r_2-\tau-z+\delta)}}\,dz \right) \\
   = & \frac{2}{\sqrt{\tau r_2}} \left(\int_{0}^{2r_2}\frac{1}{\sqrt{(z'+\delta)z'}}\,dz' 
                           + \int_{0}^{2r_2}\frac{1}{\sqrt{(2\tau-\delta-z')z'}}\,dz' \right.\\
   & \left.               + \int_{0}^{2r_2}\frac{1}{\sqrt{(2r_2+\delta-z')z'}}\,dz'
                           + \int_{0}^{2r_2}\frac{1}{\sqrt{(2(\tau-r_2)-\delta+z')z'}}\,dz' \right) \\
 \le & \frac{2}{\sqrt{\tau r_2}}\left(2\pi+2\ln(8\sqrt{d})+\ln\left(2\delta^{-1}\right)+\ln\left(2(2(\tau-r_2)-\delta)^{-1}\right)\right) \\
 \le & \frac{\kappa}{\sqrt{\tau r_2}}\left(\ln\left(\delta^{-1}\right)+\ln\left((2(\tau-r_2)-\delta)^{-1}\right)\right).
\end{align*}
For the four integrals, we used the substitutions
$z'=z+\tau-\delta$, $z'=z+\tau-\delta$, $z'=2r_2-\tau-z+\delta$, and $z'=2r_2-\tau-z+\delta$, respectively.
Using~\eqref{eqn:Integral1} and~\eqref{eqn:Integral2} and the facts that $\delta \le 2r_2 \le 2\sqrt{d}$
and $2(\tau-r_2)-\delta\le 2\tau \le 2\sqrt{d}$ yields the penultimate inequality.
The last inequality follows for the same reasons as in Case 2.1.

{\noindent\bf Case 4.2: $\delta\in(2(\tau-r_2),2\tau]$.} \\
Using~\eqref{FirstEstimateZ} and~\eqref{SecondEstimateZ}, we obtain
the following upper bound on the density of $\Delta$:
\begin{align*}
  & f_{\Delta|T=\tau,R_1=r_1,R_2=r_2}(\delta) \\
   \le & \frac{2}{\sqrt{\tau r_2}}
  \int_{-\tau+\delta}^{\tau}\left(\frac{1}{\sqrt{\tau-z}}+\frac{1}{\sqrt{\tau+z}}\right)\left(\frac{1}{\sqrt{\tau+z-\delta}}+\frac{1}{\sqrt{2r_2-\tau-z+\delta}}\right)\, dz  \\
   = & \frac{2}{\sqrt{\tau r_2}} \left(\int_{-\tau+\delta}^{\tau}\frac{1}{\sqrt{(\tau-z)(\tau+z-\delta)}}\,dz 
                                                    + \int_{-\tau+\delta}^{\tau}\frac{1}{\sqrt{(\tau+z)(\tau+z-\delta)}}\,dz \right.\\
   & \left.                                      + \int_{-\tau+\delta}^{\tau}\frac{1}{\sqrt{(\tau-z)(2r_2-\tau-z+\delta)}}\,dz
                                                    + \int_{-\tau+\delta}^{\tau}\frac{1}{\sqrt{(\tau+z)(2r_2-\tau-z+\delta)}}\,dz \right) \\
   = & \frac{2}{\sqrt{\tau r_2}} \left(\int_{0}^{2\tau-\delta}\frac{1}{\sqrt{(2\tau-\delta-z')z'}}\,dz' 
                           + \int_{0}^{2\tau-\delta}\frac{1}{\sqrt{(z'+\delta)z'}}\,dz' \right.\\
   & \left.               + \int_{0}^{2\tau-\delta}\frac{1}{\sqrt{z'(2(r_2-\tau)+\delta+z')}}\,dz'
                           + \int_{\delta}^{2\tau}\frac{1}{\sqrt{z'(2r_2+\delta-z')}}\,dz' \right) \\
 \le & \frac{2}{\sqrt{\tau r_2}}\left(2\pi+2\ln(8\sqrt{d})+\ln\left(2\delta^{-1}\right)+\ln\left(2(2(r_2-\tau)+\delta)^{-1}\right)\right) \\
 \le & \frac{\kappa}{\sqrt{\tau r_2}}\left(\ln\left(\delta^{-1}\right)+\ln\left((2(r_2-\tau)+\delta)^{-1}\right)\right).
\end{align*}
For the four integrals, we used the substitutions
$z'=\tau+z-\delta$, $z'=\tau+z-\delta$, $z'=\tau-z$, and $z'=\tau+z$, respectively.
Using~\eqref{eqn:Integral1} and~\eqref{eqn:Integral2} and the facts that $\delta \le 2\tau \le 2\sqrt{d}$
and $2(r_2-\tau)+\delta\le 2r_2 \le 2\sqrt{d}$ yields the penultimate inequality.
The last inequality follows for the same reasons as in Case 2.1.

Altogether, this yields the lemma.
\end{proof}

\subsection{Proof of Lemma~\ref{lemma:UB:DensitiesDeltaZ}}
\label{appendix:proof:densities}

First, we derive the following lemma, which gives bounds on the 
conditional density of the random variable $\Delta$ when only one of 
the radii $R_1$ and $R_2$ is given.
\begin{lemma}
\label{lemma:densityDelta2}
Let $r_1,r_2,\tau\in(0,\sqrt{d})$ and $\delta\in(0,1/2]$. In the 
following, let $\kappa$ denote a sufficiently large constant.
\begin{enumerate}
 \renewcommand{\labelenumi}{\alph{enumi})}
\item The density of $\Delta$ under the conditions $T=\tau$ and $R_1=r_1$ is
bounded by
\[
f_{\Delta|T=\tau,R_1=r_1}(\delta) \le \begin{cases}
\frac{\kappa}{\sqrt{\tau r_1}}\cdot\ln\left(\delta^{-1}\right) & \mbox{if }r_1\le\tau,\\
\frac{\kappa}{\tau}\cdot\ln\left(\delta^{-1}\right) & \mbox{if }r_1\ge\tau.
\end{cases}
\]
\item The density of $\Delta$, under the conditions $T=\tau$ and
$R_2=r_2$, is bounded by
\[
f_{\Delta|T=\tau,R_2=r_2}(\delta) \le \begin{cases}
   \frac{\kappa}{\sqrt{\tau r_2}}\cdot(\ln\left(\delta^{-1}\right)+\ln|2(\tau-r_2)-\delta|^{-1}) & \mbox{if } r_2\le\tau,\\
   \frac{\kappa}{\tau}\cdot\ln\left(\delta^{-1}\right) & \mbox{if } r_2\ge\tau.
   \end{cases}
\]
\end{enumerate}
\end{lemma}
\begin{proof}
a) We can write the density of $\Delta$ under the conditions $T=\tau$
and $R_1=r_1$ as
\begin{equation}\label{eqn:f_DeltaTauR1}
  f_{\Delta|T=\tau,R_1=r_1}(\delta) 
   = \int_{0}^{\sqrt{d}}f_{R_2}(r_2)\cdot
  f_{\Delta|T=\tau,R_1=r_1,R_2=r_2}(\delta)\,dr_2,
\end{equation}
where $f_{R_2}$ denotes the density of the length $R_2=\dist(O,Q_2)$.
The point~$Q_2$ is chosen uniformly at random from a hyperball with radius~$\sqrt{d}$ centered at the point~$O$.
The volume of a $d$-dimensional hyperball of radius~$r\ge 0$ is~$V_d(r)=\alpha r^d$ for~$\alpha=\frac{\pi^{d/2}}{\Gamma(d/2+1)}$
(see~\cite{bronshtein2007handbook}). The probability distribution~$F_{R_2}(r)$ of~$R_2$ is, for~$r\in[0,\sqrt{d}]$, proportional to~$V_d(r)$.
Let~$F_{R_2}(r)=\beta \alpha r^d$ for some~$\beta\ge 0$. Since~$F_{R_2}(\sqrt{d})=1$, it must be true that~$\beta=\frac{1}{\alpha d^{d/2}}$.
This yields, for~$r\in[0,\sqrt{d}]$,
\[
    f_{R_2}(r) = \frac{d}{dr} F_{R_2}(r) = \beta \alpha d r^{d-1} = \frac{r^{d-1}}{d^{d/2-1}}.
\]
Together with~\eqref{eqn:f_DeltaTauR1} this implies
\[
  f_{\Delta|T=\tau,R_1=r_1}(\delta) 
  = \int_{0}^{\sqrt{d}}\frac{r_2^{d-1}}{d^{d/2-1}}\cdot 
  f_{\Delta|T=\tau,R_1=r_1,R_2=r_2}(\delta)\,dr_2.
\]

We use Lemma~\ref{lemma:densityDelta} to bound this integral. For 
$r_1\le\tau$, we obtain
\begin{alignat*}{1}
   & f_{\Delta|T=\tau,R_1=r_1}(\delta)\\
   \le & \int_{0}^{\tau}\frac{r_2^{d-1}}{d^{d/2-1}}\cdot  \frac{\kappa}{\sqrt{r_1r_2}}\left(\ln\left(\delta^{-1}\right)+\ln|2(r_1-r_2)-\delta|^{-1}\right)\,dr_2\\
  & + \int_{\tau}^{\sqrt{d}}\frac{r_2^{d-1}}{d^{d/2-1}}\cdot \frac{\kappa}{\sqrt{\tau r_1}}\cdot\ln\left(\delta^{-1}\right)\,dr_2\\
   = & \frac{\kappa\ln\left(\delta^{-1}\right)}{d^{d/2-1}\sqrt{r_1}}\int_{0}^{\tau}r_2^{d-3/2}\,dr_2
       + \frac{\kappa}{d^{d/2-1}\sqrt{r_1}}\int_{0}^{\tau}r_2^{d-3/2}\ln|2(r_1-r_2)-\delta|^{-1}\,dr_2\\
   & + \frac{\kappa \ln\left(\delta^{-1}\right)}{d^{d/2-1}\sqrt{\tau r_1}} \int_{\tau}^{\sqrt{d}}r_2^{d-1} \,dr_2.
\end{alignat*}
The integral in the second line corresponds to the case~$r_1\le \tau$ and~$r_2\le \tau$ of Lemma~\ref{lemma:densityDelta}
and the integral in the third line corresponds to the case~$r_1\le\tau\le r_2$. 
Using the fact that~$\tau\le\sqrt{d}=O(1)$ and $\ln\left(\delta^{-1}\right) \ge \ln(2) = \Omega(1)$, the density~$f_{\Delta|T=\tau,R_1=r_1}(\delta)$
can be bounded from above by
\begin{alignat}{1}
      & \frac{\kappa\ln\left(\delta^{-1}\right)}{d^{d/2-1}\sqrt{r_1}}\int_{0}^{\tau}(\sqrt{d})^{d-3/2}\,dr_2
       + \frac{\kappa}{d^{d/2-1}\sqrt{r_1}}\int_{0}^{\tau}(\sqrt{d})^{d-3/2}\ln|2(r_1-r_2)-\delta|^{-1}\,dr_2\notag\\
   & + \frac{\kappa \ln\left(\delta^{-1}\right)}{d^{d/2-1}\sqrt{\tau r_1}} \int_{\tau}^{\sqrt{d}}(\sqrt{d})^{d-1} \,dr_2\notag\\
   & = \frac{O(1)}{\sqrt{r_1}}\cdot\ln\left(\delta^{-1}\right)
          +  \frac{O(1)}{\sqrt{r_1}}\cdot \int_{0}^{\tau}\ln|2(r_1-r_2)-\delta|^{-1}\,dr_2
          + \frac{O(1)}{\sqrt{\tau r_1}}\cdot\ln\left(\delta^{-1}\right). \label{eqn:CalculationIntegral1}
\end{alignat}

In order to bound the integral in the second term, we use the following lemma.
\begin{lemma}\label{lem:CalculationIntegralDetail}
Let~$f\colon \RR\to\RR$ be a linear function of the form~$f(x)=ax+b$ for arbitrary~$a,b\in\RR$ with~$|a|\ge 1$.
Furthermore, let~$c\in\RR$ and~$\varepsilon> 0$ be arbitrary. Then
\[
    \int_{c}^{c+\varepsilon}\ln\left(\frac{1}{|f(x)|}\right)\,dx \le \varepsilon \left(\ln\left(\frac{2}{\varepsilon}\right)+1\right).
\]
\end{lemma}
\begin{proof}
First we substitute~$z$ for $ax+b$ in the integral:
\begin{equation}\label{eqn:CalculationIntegralDetail1}
     \int_{c}^{c+\varepsilon}\ln\left(\frac{1}{|f(x)|}\right)\,dx
  = \int_{c}^{c+\varepsilon}\ln\left(\frac{1}{|ax+b|}\right)\,dx
  = \frac{1}{a}\int_{ac+b}^{a(c+\varepsilon)+b}\ln\left(\frac{1}{|z|}\right)\,dz.
\end{equation}
We first consider the case~$a>0$. In this case, the integral~$\int_{B}^{B+a\varepsilon}\ln(1/|z|)\,dz$
is maximized for~$B=-a\varepsilon/2$ because~$\ln(1/|z|)$ is symmetric around~$0$ and 
monotonically decreasing for $z>0$.
This yields
\begin{align*}
     &   \frac{1}{a}\int_{ac+b}^{a(c+\varepsilon)+b}\ln\left(\frac{1}{|z|}\right)\,dz
   \le \frac{1}{a}\int_{-a\varepsilon/2}^{a\varepsilon/2}\ln\left(\frac{1}{|z|}\right)\,dz
   = \frac{2}{a}\int_{0}^{a\varepsilon/2}\ln\left(\frac{1}{z}\right)\,dz.\\
    & = \frac{2}{a} \left[z(\ln(1/z)+1)\right]_{0}^{a\varepsilon/2}
    = \frac{2}{a} \cdot \frac{a\varepsilon}{2} \left(\ln\left(\frac{2}{a\varepsilon}\right)+1\right)
    = \varepsilon \left(\ln\left(\frac{2}{a\varepsilon}\right)+1\right).
\end{align*}
 
For~$a<0$, the last integral in~\eqref{eqn:CalculationIntegralDetail1} can be rewritten as follows:
\[
         \frac{1}{a}\int_{ac+b}^{a(c+\varepsilon)+b}\ln\left(\frac{1}{|z|}\right)\,dz
    =  \frac{1}{|a|}\int_{a(c+\varepsilon)+b}^{ac+b}\ln\left(\frac{1}{|z|}\right)\,dz.
\]
In this case the integral~$\int_{B+a\varepsilon}^{B}\ln(1/|z|)\,dz$
is maximized for~$B=-a\varepsilon/2$ because~$\ln(1/|z|)$ is symmetric around~$0$ and 
monotonically decreasing for $z>0$.
This yields
\begin{align*}
     &   \frac{1}{|a|}\int_{a(c+\varepsilon)+b}^{ac+b}\ln\left(\frac{1}{|z|}\right)\,dz
   \le \frac{1}{|a|}\int_{a\varepsilon/2}^{-a\varepsilon/2}\ln\left(\frac{1}{|z|}\right)\,dz
   = \frac{2}{|a|}\int_{0}^{|a|\varepsilon/2}\ln\left(\frac{1}{z}\right)\,dz.\\
    & = \frac{2}{|a|} \left[z(\ln(1/z)+1)\right]_{0}^{|a|\varepsilon/2}
    = \frac{2}{|a|} \cdot \frac{|a|\varepsilon}{2} \left(\ln\left(\frac{2}{|a|\varepsilon}\right)+1\right)
    = \varepsilon \left(\ln\left(\frac{2}{|a|\varepsilon}\right)+1\right).
\end{align*}
Altogether this proves the lemma because~$|a|\ge 1$.
\end{proof}

The previous lemma and~\eqref{eqn:CalculationIntegral1} imply that the density~$f_{\Delta|T=\tau,R_1=r_1}(\delta)$
is bounded from above by
\begin{alignat*}{1}
        &    \frac{O(1)}{\sqrt{r_1}}\cdot\ln\left(\delta^{-1}\right)
          +  \frac{O(1)}{\sqrt{r_1}}\cdot \tau \left(\ln\left(\frac{2}{\tau}\right)+1\right)           
          +  \frac{O(1)}{\sqrt{\tau r_1}}\cdot\ln\left(\delta^{-1}\right)\\
          &  = \frac{O(1)}{\sqrt{r_1}}\cdot\ln\left(\delta^{-1}\right)
          +  \frac{O(1)}{\sqrt{r_1}}           
          +  \frac{O(1)}{\sqrt{\tau r_1}}\cdot\ln\left(\delta^{-1}\right),
\end{alignat*}
where we used~$\tau\le\sqrt{d}=O(1)$ (which implies~$\tau\ln(2/\tau)=O(1)$) for the equality.
For a sufficiently large constant~$\kappa'$ we can bound the previous term from above by 
\begin{alignat*}{1}
          \frac{\kappa'}{\sqrt{\tau r_1}}\cdot\ln\left(\delta^{-1}\right),
\end{alignat*}
where we used $\ln\left(\delta^{-1}\right)\ge \ln(2) =\Omega(1)$ and~$\tau\le\sqrt{d}$.

For $\tau\le r_1$ we obtain
\begin{alignat*}{1}
   f_{\Delta|T=\tau,R_1=r_1}(\delta)
   \le & \int_{0}^{\tau}\frac{r_2^{d-1}}{d^{d/2-1}}\cdot 
  \frac{\kappa}{\sqrt{\tau r_2}}\left(\ln\left(\delta^{-1}\right)+\ln|2(\tau-r_2)-\delta|^{-1}\right)\,dr_2\\
   & + 
  \int_{\tau}^{\sqrt{d}}\frac{r_2^{d-1}}{d^{d/2-1}}\cdot\frac{\kappa}{\tau}\cdot\ln\left(\delta^{-1}\right)\,dr_2,
\end{alignat*}
where the integral in the first line corresponds to the case~$r_2\le \tau \le r_1$ of Lemma~\ref{lemma:densityDelta}
and the integral in the second line corresponds to the case~$\tau\le r_1$ and~$\tau\le r_2$. 
Analogously to the case~$r_1\le \tau$, this implies that the density~$f_{\Delta|T=\tau,R_1=r_1}(\delta)$ is bounded
from above by
\begin{alignat*}{1}
 & \frac{\kappa}{d^{d/2-1}\sqrt{\tau}}\int_{0}^{\tau}r_2^{d-3/2}\left(\ln\left(\delta^{-1}\right)+\ln|2(\tau-r_2)-\delta|^{-1}\right)\,dr_2\\
   & + \frac{\kappa}{d^{d/2-1}\tau}\cdot\int_{\tau}^{\sqrt{d}}r_2^{d-1}\ln\left(\delta^{-1}\right)\,dr_2\\
  \le & \frac{\kappa}{d^{d/2-1}\sqrt{\tau}}\int_{0}^{\tau}(\sqrt{d})^{d-3/2}\left(\ln\left(\delta^{-1}\right)+\ln|2(\tau-r_2)-\delta|^{-1}\right)\,dr_2\\
   & + \frac{\kappa}{d^{d/2-1}\tau}\cdot\int_{\tau}^{\sqrt{d}}(\sqrt{d})^{d-1}\ln\left(\delta^{-1}\right)\,dr_2\\   
  = &    \frac{O(1)}{\sqrt{\tau}}\cdot\ln\left(\delta^{-1}\right)
 +  \frac{O(1)}{\sqrt{\tau}}\cdot \int_{0}^{\tau}\ln|2(\tau-r_2)-\delta|^{-1}\,dr_2
 + \frac{O(1)}{\tau}\cdot\ln\left(\delta^{-1}\right).
\end{alignat*}
By Lemma~\ref{lem:CalculationIntegralDetail} this is bounded from above by
\[
     \frac{O(1)}{\sqrt{\tau}}\cdot\ln\left(\delta^{-1}\right)
 +  \frac{O(1)}{\sqrt{\tau}}\cdot  \tau \left(\ln\left(\frac{2}{\tau}\right)+1\right)
 + \frac{O(1)}{\tau}\cdot\ln\left(\delta^{-1}\right)
 \le \frac{\kappa'}{\tau}\ln\left(\delta^{-1}\right),
\]
for a sufficiently large constant~$\kappa'$.

b)  We can write the density of $\Delta$ under the conditions $T=\tau$
and $R_2=r_2$ as
\begin{equation}\label{eqn:DensityDeltaTauR2}
  f_{\Delta|T=\tau,R_2=r_2}(\delta) 
  = \int_{0}^{\sqrt{d}}\frac{r_1^{d-1}}{d^{d/2-1}}\cdot 
  f_{\Delta|T=\tau,R_1=r_1,R_2=r_2}(\delta)\,dr_1.
\end{equation}
For $r_2\le\tau$ and sufficiently large constants $\kappa'$ and $\kappa''$, we
obtain
\begin{alignat*}{1}
  f_{\Delta|T=\tau,R_2=r_2}(\delta) 
  & \le \int_{0}^{\tau}\frac{r_1^{d-1}}{d^{d/2-1}}\cdot 
  \frac{\kappa}{\sqrt{r_1r_2}}\left(\ln\left(\delta^{-1}\right)+\ln|2(r_1-r_2)-\delta|^{-1}\right)\,dr_1\\
  & + \int_{\tau}^{\sqrt{d}}\frac{r_1^{d-1}}{d^{d/2-1}}\cdot\frac{\kappa}{\sqrt{\tau r_2}}\left(\ln\left(\delta^{-1}\right)+\ln|2(\tau-r_2)-\delta|^{-1}\right)\,dr_1.  
\end{alignat*}
The integral in the first line corresponds to the case~$r_1\le \tau$ and~$r_2\le \tau$ of Lemma~\ref{lemma:densityDelta}
and the integral in the second line corresponds to the case~$r_2\le\tau\le r_1$. 
Using that~$\tau\le\sqrt{d}=O(1)$ and~$\ln\left(\delta^{-1}\right) \ge \ln(2) = \Omega(1)$ 
yields that the density~$f_{\Delta|T=\tau,R_2=r_2}(\delta)$ is bounded from above by
\begin{alignat*}{1}
  & \frac{O(1)}{\sqrt{r_2}}\ln\left(\delta^{-1}\right)+\frac{O(1)}{\sqrt{r_2}}\int_{0}^{\sqrt{d}}\ln|2(r_1\!-\!r_2)\!-\!\delta|^{-1}\,dr_1
  + \frac{O(1)}{\sqrt{\tau r_2}}\left(\ln\left(\delta^{-1}\right)+\ln|2(\tau\!-\!r_2)\!-\!\delta|^{-1}\right)\\
  & \le \frac{O(1)}{\sqrt{\tau r_2}}\left(\int_{0}^{\sqrt{d}}\ln|2(r_1-r_2)-\delta|^{-1}\,dr_1
                +\ln\left(\delta^{-1}\right)+\ln|2(\tau-r_2)-\delta|^{-1}\right).
\end{alignat*}
Together with Lemma~\ref{lem:CalculationIntegralDetail} the previous formula implies the following upper bound on the density~$f_{\Delta|T=\tau,R_2=r_2}(\delta)$:
\begin{alignat*}{1}
  & \frac{O(1)}{\sqrt{\tau r_2}}\left(\sqrt{d} \left(\ln\left(\frac{2}{\sqrt{d}}\right)+1\right)
                +\ln\left(\delta^{-1}\right)+\ln|2(\tau-r_2)-\delta|^{-1}\right)\\
   & \le \frac{\kappa'}{\sqrt{\tau r_2}}\left(
                \ln\left(\delta^{-1}\right)+\ln|2(\tau-r_2)-\delta|^{-1}\right),
\end{alignat*}
for a sufficiently large constant~$\kappa'$.

For $\tau\le r_2$ and a sufficiently large constant $\kappa'$, we obtain by~$\eqref{eqn:DensityDeltaTauR2}$
and Lemma~\ref{lemma:densityDelta}
\begin{align*}
  f_{\Delta|T=\tau,R_2=r_2}(\delta) 
  & \le \int_{0}^{\tau}\frac{r_1^{d-1}}{d^{d/2-1}}\cdot 
  \frac{\kappa}{\sqrt{\tau r_1}}\cdot\ln\left(\delta^{-1}\right)\,dr_1
  + \int_{\tau}^{\sqrt{d}}\frac{r_1^{d-1}}{d^{d/2-1}}\cdot 
   \frac{\kappa}{\tau}\cdot\ln\left(\delta^{-1}\right)\,dr_1.
\end{align*}
The first integral corresponds to the case~$r_1\le \tau \le r_2$ of Lemma~\ref{lemma:densityDelta}
and the second integral corresponds to the case~$\tau\le r_1$ and~$\tau\le r_2$. 
Using that~$\tau\le\sqrt{d}=O(1)$ yields that the previous term is bounded from above by
\begin{align*}
  & \frac{\kappa}{d^{d/2-1}\sqrt{\tau}}\cdot\ln\left(\delta^{-1}\right)\int_{0}^{\tau}r_1^{d-3/2}\,dr_1
  + \frac{\kappa}{d^{d/2-1}\tau}\cdot\ln\left(\delta^{-1}\right)\int_{\tau}^{\sqrt{d}}r_1^{d-1}\,dr_1\\   
  & \frac{\kappa}{d^{d/2-1}\sqrt{\tau}}\cdot\ln\left(\delta^{-1}\right)\int_{0}^{\tau}(\sqrt{d})^{d-3/2}\,dr_1
  + \frac{\kappa}{d^{d/2-1}\tau}\cdot\ln\left(\delta^{-1}\right)\int_{\tau}^{\sqrt{d}}(\sqrt{d})^{d-1}\,dr_1\\   
   \le & \frac{\kappa'}{\tau}\cdot\ln\left(\delta^{-1}\right),
\end{align*}
for a sufficiently large constant~$\kappa'$.
\end{proof}

Now we are ready to prove Lemma~\ref{lemma:UB:DensitiesDeltaZ}.
\begin{proof}[Lemma~\ref{lemma:UB:DensitiesDeltaZ}]
a) In order to prove part a), we integrate 
$f_{\Delta|T=\tau,R_1=r}(\delta)$ over all values $\tau$ that $T$ can 
take. We denote by $f_T$ the density of the length $T=\dist(O,P)$.
We have argued in the proof of Lemma~\ref{lemma:densityDelta2} that, for~$\tau\in[0,\sqrt{d}]$,
$f_{R_2}(\tau)=f_T(\tau)=\frac{\tau^{d-1}}{d^{d/2-1}}$.
We obtain, for a sufficiently large constant~$\kappa'$,
\begin{alignat*}{1}
  f_{\Delta|R_1=r}(\delta) & = \int_{0}^{\sqrt{d}}f_{T}(\tau)\cdot
  f_{\Delta|T=\tau,R_1=r}(\delta)\,d\tau\\
  & = \int_{0}^{\sqrt{d}}\frac{\tau^{d-1}}{d^{d/2-1}}\cdot
  f_{\Delta|T=\tau,R_1=r}(\delta)\,d\tau\\ 
  & \le \int_{0}^{r}\frac{\tau^{d-1}}{d^{d/2-1}}\cdot\frac{\kappa}{\tau}\cdot\ln\left(\delta^{-1}\right)\,d\tau
  + \int_{r}^{\sqrt{d}}\frac{\tau^{d-1}}{d^{d/2-1}}\cdot\frac{\kappa}{\sqrt{\tau r}}\cdot\ln\left(\delta^{-1}\right)\,d\tau\\
  & \le \int_{0}^{\sqrt{d}}\frac{\kappa(\sqrt{d})^{d-2}}{d^{d/2-1}}\cdot\ln\left(\delta^{-1}\right)\,d\tau
  + \int_{0}^{\sqrt{d}}\frac{(\sqrt{d})^{d-3/2}}{d^{d/2-1}}\cdot\frac{\kappa}{\sqrt{r}}\cdot\ln\left(\delta^{-1}\right)\,d\tau\\
  & \le O(1)\cdot  \ln\left(\delta^{-1}\right) + \frac{O(1)}{\sqrt{r}} \ln\left(\delta^{-1}\right)
   \le \frac{\kappa'}{\sqrt{r}}\cdot\ln\left(\delta^{-1}\right),
\end{alignat*}
where we used Lemma~\ref{lemma:densityDelta2}~a) for the first inequality,
and $0\le r\le\sqrt{d}=O(1)$ and~$\ln\left(\delta^{-1}\right)\ge \ln(2) =\Omega(1)$ for the other inequalities.

Furthermore, we integrate $f_{\Delta|T=\tau,R_2=r}(\delta)$ over all 
values $\tau$ that $T$ can take:
\begin{alignat*}{1}
  f_{\Delta|R_2=r}(\delta) & =
  \int_{0}^{\sqrt{d}}\frac{\tau^{d-1}}{d^{d/2-1}}\cdot
  f_{\Delta|T=\tau,R_2=r}(\delta)\,d\tau\\ 
  & \le \int_{0}^{r}\frac{\tau^{d-1}}{d^{d/2-1}}\cdot\frac{\kappa}{\tau}\cdot\ln\left(\delta^{-1}\right)\,d\tau\\
  & + \int_{r}^{\sqrt{d}}\frac{\tau^{d-1}}{d^{d/2-1}}\cdot\frac{\kappa}{\sqrt{\tau r}}(\ln\left(\delta^{-1}\right)+\ln|2(\tau-r)-\delta|^{-1})\,d\tau\\
  & \le O(1)\cdot\ln\left(\delta^{-1}\right)+ \frac{O(1)}{\sqrt{r}}\cdot\ln\left(\delta^{-1}\right) + \frac{O(1)}{\sqrt{r}}\sqrt{d} \left(\ln\left(\frac{2}{\sqrt{d}}\right)+1\right)\\
  & \le \frac{\kappa'}{\sqrt{r}}\cdot\ln\left(\delta^{-1}\right),
\end{alignat*}
where we used Lemma~\ref{lemma:densityDelta2}~b) for the first inequality, and
Lemma~\ref{lem:CalculationIntegralDetail}, $0\le r\le\sqrt{d}=O(1)$, and~$\ln\left(\delta^{-1}\right)\ge \ln(2) =\Omega(1)$
for the second and third inequalities.

b)
Let~$f_{R_1}(r)=\frac{r^{d-1}}{d^{d/2-1}}$ denote the density of the length $R_1=\dist(O,Q_1)$.
For a sufficiently large constant $\kappa'$,
\begin{alignat*}{1}
  f_{\Delta|T=\tau}(\delta) & =
  \int_{0}^{\sqrt{d}}f_{R_1}(r)\cdot
  f_{\Delta|T=\tau,R_1=r}(\delta)\,dr\\
  & = \int_{0}^{\sqrt{d}}\frac{r^{d-1}}{d^{d/2-1}}\cdot
  f_{\Delta|T=\tau,R_1=r}(\delta)\,dr\\
  & \le \int_{0}^{\tau}\frac{r^{d-1}}{d^{d/2-1}}\cdot
  \frac{\kappa}{\sqrt{\tau r}}\cdot\ln\left(\delta^{-1}\right)\,dr
  + \int_{\tau}^{\sqrt{d}}\frac{r^{d-1}}{d^{d/2-1}}\cdot
   \frac{\kappa}{\tau}\cdot\ln\left(\delta^{-1}\right)\,dr\\
  & \le \frac{O(1)}{\sqrt{\tau}} \cdot\ln\left(\delta^{-1}\right) + \frac{O(1)}{\tau} \cdot\ln\left(\delta^{-1}\right)  
  \le \frac{\kappa'}{\tau}\cdot\ln\left(\delta^{-1}\right).
\end{alignat*}
For the penultimate inequality we used~$0\le \tau\le\sqrt{d}=O(1)$ and~$\ln\left(\delta^{-1}\right)\ge \ln(2) =\Omega(1)$.

c)
Using part b), for a sufficiently large constant $\kappa'$,
\begin{alignat*}{1}
 f_{\Delta}(\delta) 
 & = \int_{0}^{\sqrt{d}}f_T(\tau)\cdot f_{\Delta|T=\tau}(\delta)\,d\tau\\
 & \le \int_{0}^{\sqrt{d}}\frac{\tau^{d-1}}{d^{d/2-1}}\cdot
 \frac{\kappa}{\tau}\cdot\ln\left(\delta^{-1}\right)\,d\tau
 \le \kappa'\cdot\ln\left(\delta^{-1}\right).
\end{alignat*}

d)
Let $f_{R_i}$ denote the density of $R_i$.
Using Lemma~\ref{lemma:densityZ}, we obtain
\begin{alignat*}{1}
  f_{Z_i|T=\tau}(z)  
  & = \int_{r=0}^{\tau}f_{R_i}(r)\cdot f_{Z|T=\tau,R=r}(z)\,dr \\ 
  & \le \int_{r=\frac{z+\tau}{2}}^{\tau}\frac{r^{d-1}}{d^{d/2-1}}\sqrt{\frac{2}{(\tau+z)(2r-\tau-z)}}\,dr
  + \int_{r=\tau}^{\sqrt{d}}f_{R_i}(r)\sqrt{\frac{2}{\tau^2-z^2}}\,dr.
\end{alignat*}
The lower limit of the first integral follows from the fact that,
according to Lemma~\ref{lemma:densityZ}, $z$ always takes a value in the interval $(-\tau,\min\{\tau,2R_i-\tau\})$.
Since~$z\le 2R_i-\tau$ is equivalent to~$R_i\ge\frac{z+\tau}{2}$,
we can bound~$f_{Z_i|T=\tau}(z)$ from above by
\begin{alignat*}{1}
  & \sqrt{\frac{2}{\tau+z}}d^{1/2}\int_{r=\frac{z+\tau}{2}}^{\tau}\sqrt{\frac{1}{2r-\tau-z}}\,dr
  +\sqrt{\frac{2}{\tau^2-z^2}}\int_{r=\tau}^{\sqrt{d}}f_{R_i}(r)\,dr\\
  & \le \sqrt{\frac{2}{\tau+z}}d^{1/2}\int_{r=\frac{z+\tau}{2}}^{\tau}\sqrt{\frac{1}{2r-\tau-z}}\,dr
  +\sqrt{\frac{2}{\tau^2-z^2}},
\end{alignat*}
where we used $r^{d-1} \le \tau^{d-1} \le (\sqrt{d})^{d-1}$ and the fact that the integral over a density is at most $1$.
Because
\[
   \int_{\frac{z+\tau}{2}}^{\tau}\sqrt{\frac{1}{2r-\tau-z}}\,dr
   = \frac{1}{2}\int_{x=0}^{\tau-z}\sqrt{\frac{1}{x}}\,dx
   \le \frac{1}{2}\int_{0}^{\sqrt{d}}\sqrt{\frac{1}{x}}\,dx
   = [\sqrt{x}]_{0}^{\sqrt{d}} = d^{1/4} = O(1),
\]
we can bound the conditional density of~$Z_i$ from above by
\begin{alignat*}{1}
  f_{Z_i|T=\tau}(z) 
   & \le \sqrt{\frac{2}{\tau+z}}d^{1/2}\cdot O(1)
  +\sqrt{\frac{2}{\tau^2-z^2}}\\ 
  & = \frac{O(1)}{\sqrt{\tau+z}}+\frac{O(1)}{\sqrt{\tau^2-z^2}}
  \le \frac{\kappa'}{\sqrt{\tau^2-z^2}},
\end{alignat*}
for a large enough constant $\kappa'$, where we used
\[
   \tau + z = \frac{\tau^2-z^2}{\tau-z} \ge \frac{\tau^2-z^2}{\sqrt{d}}      
\]
for the last inequality, which holds because $\tau\le\sqrt{d}$ and~$z\ge 0$.
\end{proof}

\section{Negatively Associated Random Variables\label{appendix:NegativelyAssociated}}

Dubhashi and Ranjan~\cite{DubhashiR98} define negatively associated random variables as follows.
\begin{definition}[\cite{DubhashiR98}, Definition 3]
The random variables $X_1,\ldots,X_n$ are negatively associated if for every two disjoint
index sets $I,J\subseteq[n]$,
\[
   \Ex{f(X_i,i\in I)\cdot g(X_j,j\in J)} \le \Ex{f(X_i,i\in I)}\cdot \Ex{g(X_j,j\in J)}, 
\] 
for all functions $f:\RR^{|I|}\to\RR$ and $g:\RR^{|J|}\to\RR$ that are both non-decreasing or both
non-increasing.
\end{definition}

In Section~\ref{sec:approximation}, we used the following result from Dubhashi and Ranjan's paper.
\begin{lemma}[\cite{DubhashiR98}, Proposition 6]
The Chernoff-Hoeffding bounds are applicable to sums of random variables that satisfy the negative association condition.
\end{lemma}

It remains to show that the random variables $X_1,\ldots,X_k$ defined in Section~\ref{sec:approximation}
satisfy the negative association condition. Remember that these variables come from a balls-into-bins process
in which $n$ balls are put independently into $k$ bins. Each ball has its own probability distribution on the $k$ bins
and the 0-1-variable $X_i$ indicates whether bin~$i$ contains at least one ball.

In order to show that the variables $X_1,\ldots,X_k$ are negatively associated, we follow the same line of arguments
as Lenzen and Wattenhofer~\cite{LenzenW10}, who showed the same statement for a balls-into-bins process in which the balls are put
uniformly at random into the bins. The proof is based on the following statements proven in~\cite{DubhashiR98}.
\begin{lemma}\label{lemma:NegativelyAssociated}
\begin{enumerate}
  \setlength{\itemsep}{0em}
  \renewcommand{\labelenumi}{\alph{enumi})}
  \item If $X_1,\ldots,X_n$ are 0-1-random variables with $\sum X_i = 1$, then $X_1,\ldots,X_n$ are negatively associated.
  \item If $X$ and $Y$ are sets of negatively associated random variables and if the random variables in $X$ and $Y$ are mutually independent,
        then $X\cup Y$ is also negatively associated.
  \item Assume that the random variables $X_1,\ldots,X_n$ are negatively associated and, for some $k\in\NN$, let $I_1,\ldots,I_k\subseteq[n]$ be mutually disjoint index
        sets. For $j\in[k]$, let $h_j:\RR^{|I_j|}\to\RR$ be functions that are all non-decreasing or all non-increasing,
        and define $Y_j = h_j(X_i, i\in I_j)$. Then the random variables $Y_1,\ldots,Y_k$ are also negatively associated.
\end{enumerate}
\end{lemma}

Based on this lemma, we prove the theorem about the balls-into-bins process.
\begin{theorem}
Consider a balls-into-bins process in which $n$ balls are put independently into $k$ bins. 
Each ball has its own probability distribution on the $k$ bins
and the 0-1-variable $X_i$ indicates whether bin~$i$ contains at least one ball.
The random variables $X_1,\ldots,X_k$ are negatively associated.
\end{theorem}
\begin{proof}
First we define for each bin $i\in[k]$ and each ball $j\in[n]$ a 0-1-variables $X_i^j$ indicating whether ball~$j$ ends up in bin~$i$.
For a ball $j\in[n]$, the random variables $X_1^j,\ldots,X_k^j$ are negatively associated according to Lemma~\ref{lemma:NegativelyAssociated}~a). 
Since the balls are put independently into the bins, all random variables $X_i^j$ for $i\in[k]$ and $j\in[n]$ are negatively associated according to Lemma~\ref{lemma:NegativelyAssociated}~b).

Now we define for each bin $i\in[k]$ the set $I_i=\{X_i^1,\ldots,X_i^n\}$ and the function 
\[
  h_i(X_i^1,\ldots,X_i^n) = \begin{cases}
     1 & \text{if $X_i^1+\cdots+X_i^n \ge 1$,} \\
     0 & \text{if $X_i^1+\cdots+X_i^n = 0$.}
  \end{cases} 
\]
Observe that $X_i = h_i(X_i^1,\ldots,X_i^n)$. As these functions are non-decreasing Lemma~\ref{lemma:NegativelyAssociated}~c)
implies that the random variables $X_1,\ldots,X_k$ are negatively associated.
\end{proof}

\end{appendix}

\bibliographystyle{plain}
\bibliography{Literature}

\end{document}